\documentclass[acmsmall]{acmart}
\AtBeginDocument{%
  \providecommand\BibTeX{{%
    \normalfont B\kern-0.5em{\scshape i\kern-0.25em b}\kern-0.8em\TeX}}}

\usepackage[T1]{fontenc}
\usepackage{booktabs}
\usepackage{subcaption}
\usepackage{graphicx}
\usepackage{adjustbox}
\usepackage{leftidx}
\usepackage{listings}
\usepackage{stmaryrd}
\usepackage{arydshln}
\usepackage{proof}
\usepackage{xspace}
\usepackage{multirow}
\usepackage{pifont}
\usepackage{enumitem}
\setlist{nosep,leftmargin=\parindent}

\usepackage{tikz}
\usepackage{tikz-qtree}
\usepackage[framemethod=tikz]{mdframed}
\usepackage{xcolor}
\usepackage{mathtools}
\usepackage{microtype}
\usepackage{stackengine}

\usepackage{hyperref}
\usepackage{cleveref}
\crefname{section}{\S}{\S\S}
\Crefname{section}{\S}{\S\S}
\crefformat{section}{\S#2#1#3}

\usepackage{thm-restate}
\newtheorem{theorem}{Theorem}

\usepackage{wrapfig}
\usepackage{svg}
\usepackage{refs}

\usepackage{bussproofs}
\usepackage{tikz-cd}
\usepackage[normalem]{ulem}
\usepackage{leftidx}


\usepackage{cleveref}

\newcommand{\llangle}{\langle\!\langle}
\newcommand{\rrangle}{\rangle\!\rangle}

\newcommand{\rone}{(\emph{i})~}
\newcommand{\rtwo}{(\emph{ii})~}
\newcommand{\rthree}{(\emph{iii})~}

\newcommand{\lbbar}{\{\kern-0.5ex|}
\newcommand{\rbbar}{|\kern-0.5ex\}}

\newcommand{\biglbbar}{\left\{\kern-0.5ex\left|}
\newcommand{\bigrbbar}{\right|\kern-0.5ex\right\}}

\newcommand{\revision}[1]{{{\color{magenta}#1}}}
\newcommand{\noappendix}[1]{{#1}}

\definecolor{noncrappyred}{RGB}{209, 25, 8}
\definecolor{noncrappyblue}{RGB}{0, 69, 245}

\def\cvc5{\textsc{cvc5}\xspace}

\newcommand{\type}{\tau}

\newcommand{\Etrue}{{{\mathsf{true}}}}
\newcommand{\Efalse}{{{\mathsf{false}}}}
\newcommand{\Et}{{{\mathsf{t}}}}
\newcommand{\Ef}{{{\mathsf{f}}}}
\newcommand{\Eifthenelse}[3]{{\mathsf{if}}\ {#1}\ {\mathsf{then}}\ {#2}\ {\mathsf{else}}\ {#3}}
\newcommand{\Eseq}[2]{{#1}{\mathsf{;}}\ {#2}}

\newcommand{\sem}[1]{\llbracket{#1}\rrbracket}       % [[ ]] 

\newcommand{\Eassign}[2]{{{#1}\ \mathsf{:=}\ {#2}}}

\newcommand{\Ewhile}[2]{\mathsf{while} \; {#1} \; \mathsf{do} \; {#2}}

\newcommand{\semantics}[1]{\llbracket {#1} \rrbracket}

\newcommand{\Omit}[1]{}

\usepackage[breakable]{tcolorbox}
\newenvironment{mybox}[1][gray!20]{
	\begin{tcolorbox}[   %% Adjust the following parameters at will.
		breakable,
		left=0pt,
		right=0pt,
		top=0pt,
		bottom=-1pt,
		colback=#1,
		colframe=#1,
		width=\dimexpr\textwidth\relax,
		boxsep=2pt,
		arc=0pt,outer arc=0pt,
		]
	}{
\end{tcolorbox}
}

\newcounter{resq}%[section]

\newcommand{\triple}[3]{\{{#1}\} \: #2 \: \{#3\}}

\newcommand{\utripleof}[4]{\lbbar {#1} \rbbar^{#4} \: #2 \: \lbbar #3 \rbbar^{#4}}
\newcommand{\utriple}[3]{\lbbar {#1} \rbbar \: #2 \: \lbbar #3 \rbbar}
\newcommand{\utripleag}[3]{\utripleof{#1}{#2}{#3}{\setsing}}
\newcommand{\utriplevs}[3]{\utripleof{#1}{#2}{#3}{\setweak}}

\newcommand{\Eskip}{\mathsf{skip}}

\newcommand{\vars}{\mathit{vars}}

\newcommand{\grmdisj}{\mathsf{GrmDisj}}

\newcommand{\set}[1]{\{{#1}\}}

\newcommand{\vinf}{V_{\mathsf{inf}}}

\definecolor{dgreen}{RGB}{0,128,0}
\definecolor{dred}{RGB}{200,0,0}

\newcommand{\gimp}{\mathit{G_{imp}}}
\newcommand{\defeq}{\ensuremath{\triangleq}}

\newcommand{\op}{\mathit{op}}

\newenvironment{proofsketch}{%
    \par\medskip\noindent%
    \textit{Proof Sketch.}%    
}{%    
    \hfill$\square$% Add a square symbol at the end of the proof sketch
    \par\medskip%
}

\newcommand{\Z}{\mathbb{Z}}
\newcommand{\N}{\mathbb{N}}
\newcommand{\divg}{\uparrow}
\def\state{\sigma}
\newcommand{\State}{\mathsf{State}}
\newcommand{\Prog}{\mathsf{Prog}}
\newcommand{\DState}{\mathsf{DState}}
\newcommand{\abs}[1]{\lvert {#1} \rvert}
\newcommand{\subs}[2]{[{#2}\mapsto{#1}]}

\newcommand{\stateentry}[2]{#1 \mapsto #2}
\newcommand{\statelist}[1]{({#1})}
\newcommand{\filter}[2]{\mathit{filter}({#1}, {#2})}

\newcommand{\subvec}[3]{\,{#1}_{{#2} \cdots {#3}}}
\newcommand{\pend}{++}

\def\domain{\tau}

\def\op{\alpha^{(n)}}
\def\varfun{\mathit{vars}}
\def\Tvars{\mathit{T}}
\def\Stmtvars{\mathit{Stmt}}
\def\Expvars{\mathit{Exp}}
\def\vars{\mathit{Var}}
\def\Boolvars{\mathit{BExp}}
\def\Tset{\mathbf{T}}
\def\Stmtset{\mathbf{Stmt}}
\def\Expset{\mathbf{Exp}}
\def\Boolset{\mathbf{BExp}}
\def\varspres{\mathit{vp}}
\def\vsvars{\mathit{EVS}}

\newcommand{\abstrp}[1]{\semprog{\cdot}^{#1}}
\newcommand{\abstrpnovars}[1]{\semprog{\cdot}^{#1}}
\newcommand{\abstrpstmt}[1]{\semprog{\cdot}^{#1}_{\Stmtvars}}
\newcommand{\abstrpbool}[1]{\semprog{\cdot}^{#1}_{\Boolvars}}
\newcommand{\abstrpexp}[1]{\semprog{\cdot}^{#1}_{\Expvars}}

\newcommand{\abstrpapp}[2]{\semprog{#2}^{#1}}

\newcommand{\abstrs}[1]{\semset{\cdot}^{#1}}
\newcommand{\abstrsnovars}[1]{\semset{\cdot}^{#1}}
\newcommand{\abstrsstmt}[1]{\semset{\cdot}^{#1}_{\Stmtvars}}
\newcommand{\abstrsbool}[1]{\semset{\cdot}^{#1}_{\Boolvars}}
\newcommand{\abstrsexp}[1]{\semset{\cdot}^{#1}_{\Expvars}}
\newcommand{\abstrst}[1]{\semset{\cdot}^{#1}_{T}}
\newcommand{\abstrsapp}[2]{\semset{#2}^{#1}}

\newcommand{\abstrstapp}[2]{\semset{#2}^{#1}_{T}}

\newcommand{\semset}[1]{\llangle{#1}\rrangle}
\newcommand{\semprog}[1]{\llbracket{#1}\rrbracket}
\newcommand{\semsingset}[1]{\semset{#1}^{\setsing}}
\newcommand{\semsingprog}[1]{\semprog{#1}^{\sing}}

\newcommand{\semsingsetgrn}[1]{\semset{#1}^{\setsinggrn}}
\newcommand{\semsingproggrn}[1]{\semprog{#1}^{\singgrn}}

\newcommand{\semsingsetyel}[1]{\semset{#1}^{\setsingyel}}
\newcommand{\semsingprogyel}[1]{\semprog{#1}^{\singyel}}

\newcommand{\semfullset}[1]{\semset{#1}^{\setfull}}
\newcommand{\semfullprog}[1]{\semprog{#1}^{\full}}
\newcommand{\semweakset}[1]{\semset{#1}^{\setweak}}

\newcommand{\semweaksetgrn}[1]{\semset{#1}^{\setweakgrn}}
\newcommand{\semweaksetgrnbad}[1]{\semset{#1}^{\setweakgrn_{\mathbf{fine}}}}
\newcommand{\semweakproggrnbad}[1]{\semprog{#1}^{\setweakgrn_{\mathbf{fine}}}}
\newcommand{\semweakproggrn}[1]{\semprog{#1}^{\weakgrn}}

\newcommand{\semweaksetgrnnc}[1]{\semset{#1}^{\setweakgrn_{\mathbf{coarse}}}}

\newcommand{\semweaksetyel}[1]{\semset{#1}^{\setweakyel}}
\newcommand{\semweakprogyel}[1]{\semprog{#1}^{\weakyel}}
\newcommand{\semulagset}[1]{\abstrsapp{La}{#1}}
\newcommand{\semulvsset}[1]{\abstrsapp{Lv}{#1}}

\newcommand{\green}{divergence\text{-}aware\xspace}
\newcommand{\greenCaps}{Divergence\text{-}Aware\xspace}
\newcommand{\yellow}{divergence\text{-}agnostic\xspace}
\newcommand{\yellowCaps}{Divergence\text{-}Agnostic\xspace}
\newcommand{\greensup}{t\xspace}
\newcommand{\yellowsup}{\xspace}
\newcommand{\sing}{\mathit{a}}
\newcommand{\singyel}{\mathit{a\yellowsup}}
\newcommand{\singgrn}{\mathit{a\greensup}}
\newcommand{\setsing}{\mathbf{a}}
\newcommand{\setsingyel}{\mathbf{a\yellowsup}}
\newcommand{\setsinggrn}{\mathbf{a\greensup}}
\newcommand{\full}{\mathit{aw}}
\newcommand{\setfull}{\mathbf{aw}}
\newcommand{\weak}{\mathit{v}}
\newcommand{\weakyel}{\mathit{v\yellowsup}}
\newcommand{\weakgrn}{\mathit{v\greensup}}
\newcommand{\setweak}{\mathbf{v}}
\newcommand{\setweakyel}{\mathbf{v\yellowsup}}
\newcommand{\setweakgrn}{\mathbf{v\greensup}}

\newcommand{\prog}{c}
\newcommand{\progset}{C}
\newcommand{\Red}{\mathit{reduce}}

\newcommand{\itrp}{\mathtt{interpret}}

\newcommand{\opn}{\alpha}
\newcommand{\opset}{\mathcal{A}}
\renewcommand{\pod}[1]{\allowbreak\mathchoice
  {\if@display \mkern 18mu\else \mkern 8mu\fi (#1)}
  {\if@display \mkern 18mu\else \mkern 8mu\fi (#1)}
  {\mkern4mu(#1)}
  {\mkern4mu(#1)}
}



\begin{document}
% for arxiv
% \pagestyle{plain}

\title{Semantics of Sets of Programs}

\author{Jinwoo Kim}
\authornote{Jinwoo Kim and Shaan Nagy contributed equally to this paper.}
\orcid{0000-0002-3897-1828}
\affiliation{%
  \institution{University of California-San Diego}
  \city{San Diego}
  \country{USA}
}
\email{pl@ucsd.edu}

\author{Shaan Nagy}
\authornotemark[1]
\orcid{0000-0001-8015-5421}
\affiliation{%
  \institution{University of California-San Diego}
  \city{San Diego}
  \country{USA}
}
\email{shnagy@ucsd.edu}

\author{Thomas Reps}
\orcid{0000-0002-5676-9949}
\affiliation{%
  \institution{University of Wisconsin-Madison}
  \city{Madison}
  \country{USA}
}
\email{reps@cs.wisc.edu}

\author{Loris D'Antoni}
\orcid{0000-0001-9625-4037}
\affiliation{%
  \institution{University of California-San Diego}
  \city{San Diego}
  \country{USA}
}
\email{ldantoni@ucsd.edu}

%\title{Semantics of Sets of Programs}
%\author{Jinwoo Kim}
%\authornote{The first two authors contributed equally to this paper.}
%\orcid{0000-0002-3897-1828}
%\affiliation{
%    \institution{University of California-San Diego}
%    \country{USA}
%}
%\email{pl@ucsd.edu}
%\author{Shaan Nagy}
%\orcid{0000-0001-8015-5421}
%\affiliation{
%    \institution{University of California-San Diego}
%    \country{USA}
%}
%\email{shnagy@ucsd.edu}
%\author{Thomas Reps}
%\orcid{0000-0002-5676-9949}
%\affiliation{
%    \institution{University of Wisconsin-Madison}
%    \country{USA}
%}
%\email{reps@cs.wisc.edu}
%\author{Loris D'Antoni}
%\orcid{0000-0001-9625-4037}
%\affiliation{
%    \institution{University of California-San Diego}
%    \country{USA}
%}
%\email{ldantoni@ucsd.edu}

%\renewcommand{\shortauthors}{}

\begin{abstract}
Applications like program synthesis sometimes require proving that a property holds for all of the 
infinitely many programs described by a grammar---i.e., an inductively defined set of programs. 
Current verification frameworks overapproximate programs' behavior when sets of programs contain loops, including two Hoare-style logics that
fail to be relatively complete
when loops are allowed. In this work, we prove that compositionally verifying simple properties for infinite sets of programs requires tracking distinct program behaviors over unboundedly many executions. Tracking this information is both necessary and sufficient for verification. We prove this fact in a general, reusable theory of denotational semantics that 
can model the expressivity and compositionality of
verification techniques
over infinite sets of programs. We construct the minimal compositional semantics that captures simple properties of sets of programs and use it to
derive the first sound and relatively complete Hoare-style logic for infinite sets of programs. Thus, our methods 
can be used to design minimally complex, compositional verification techniques for sets of programs.

\end{abstract}
\maketitle
% also for arxiv
% \thispagestyle{plain}

\section{Introduction}
\label{sec:intro}
The central problem of this work is to understand what formal semantics are needed when verifying infinite sets of imperative programs.

\begin{example} [Verifying a Set of Imperative Programs]
\label{ex:motive}
Consider the following grammar, which defines a set $L(W)$ of loops:
\begin{align*}
    W& ::= \Ewhile{x < E}{x:=x + 1} \\
    E& ::= 0~|~E + 2 
\end{align*}

\noindent Can we prove that, for every loop $w \in L(W)$, running $w$ on the initial value $x=0$ results in a final state in which $x$ is even?

\end{example}

The need to verify
a property of an infinite set of programs
appears in several domains. 
For example, in language-wide verification, one may wish to prove that every program in a given domain-specific language 
terminates or satisfies some safety property.
Another example is verifying programs that contain constructs like \texttt{eval}, 
which load code as a string/tree at run time and dynamically execute it. 
In this setting, one may need to reason about 
infinitely many  possible strings/trees that represent programs and are passed to \texttt{eval}~\cite{Thesis:Malmkjaer93}.
A third example, and the main motivation behind our work, is proving unrealizability of program-synthesis problems~\cite{uls}, 
where the goal is to determine whether all programs in a search space satisfy the negation of a given specification---i.e., no program in the search space matches the desired specification. 
Unrealizability of program-synthesis problems can be used to speed up program synthesis and also to prove optimality of a program $\prog$ with respect to a given metric~\cite{HuD18}---i.e., by
showing that the problem of synthesizing a program that is better than $\prog$ is unrealizable.

While there are many techniques to prove that a \emph{particular} loop $w \in L(W)$ satisfies the specification given in Example~\ref{ex:motive}, no verification technique that we are aware of can prove that the specification holds for
\emph{all} $w \in L(W)$.

Our two main goals in this paper are: \rone to present a fundamental theoretical
obstacle that
makes it hard to design effective frameworks for reasoning about extensional properties (i.e., input-output properties) of sets of deterministic imperative programs with loops (e.g. \Cref{ex:motive}), and \rtwo in light of this obstruction, to present a provably simplest semantic basis for designing verification techniques. To ground our study, we first consider reasoning about 
individual imperative programs.

\paragraph{Denotational Semantics for Single Imperative Programs}
Formal verification can always be defined with respect to a denotational semantics---i.e.,
a map that assigns a mathematical object, called a denotation, to each program. 
These denotations 
capture the program information one is interested in studying.
A familiar example of
a denotation is a
\textit{state transformer}---i.e., a function (or  relation) that captures how a program maps input states to output states.

In this paper, we are interested in understanding how powerful denotations must be 
to capture extensional properties for sets of programs. 
For a \textit{single program} $\prog$, an 
\textit{extensional property} 
is a property of the relation between inputs and outputs captured by the state transformer of $\prog$.~\footnote{In this paper, we use $c$ to refer to arbitrary single programs (i.e., program statements or expressions). We use $s$ to refer specifically to program statements (derivable from $\Stmtvars$ in the grammar given in \Cref{fig:gimp}). Similarly, we use $C$ to refer to sets of programs, and we use $S$ to refer to sets of program statements.}
For example, monotonicity and commutativity are extensional properties---they are 
direct properties of the mathematical relation that the state transformer assigned to $\prog$ induces.
On the other hand, program text size or running time of the program are \textit{not} extensional properties---they are unrelated to the input-output relation described by $\prog$.

To reason about extensional properties
%\footnote{An extensional property of a program $\prog$ is a property of $\sem{\prog}$, the (partial) function that $\prog$ computes. Thus, monotonicity is an extensional property ($\forall \sigma_1, \sigma_2.~ \sigma_1[x_1] \leq \sigma_2[x_1] \rightarrow \sem{\prog}(\sigma_1)[x_1] \leq \sem{\prog}(\sigma_2)[x_1]$), but program size and runtime are not.}
of individual deterministic imperative programs, one often uses a semantics that maps every syntactic program to a partial function mapping input program states $\sigma$ to output program states~\cite{Book:Stoy77,Book:Schmidt86}, illustrated in the following example:
\begin{equation}\label{eq:indiv_semantics}
    \sem{\Eassign{x}{10}} = \lambda \sigma. \sigma\subs{10}{x}
\end{equation} 

\noindent
This semantics is the simplest one that expresses single-input extensional properties (i.e., the output
state of the program on a given single input state $\sigma$). Despite its simplicity, this denotational semantics boasts two key properties that enable effective verification techniques to be built on top of it, the second of which enables the first:

\begin{description}
    \item[Compositionality.] The denotational semantics of a complex program can be expressed precisely in terms of the semantics of its subprograms, as shown by the following example:  
    \[\sem{\Eseq{\Eassign{x}{10}}{\Eassign{y}{12}}} = \sem{\Eassign{y}{12}} \circ \sem{\Eassign{x}{10}}\] 

    \item[Expression of Multiple Executions.] The semantics of a program is capable of describing how the program acts on multiple different input states. For example, to understand how $\prog$ acts on two states rather than just one, one need not define a new semantics. Instead, 
    one can apply $\sem{\prog}$ to each state:
    \begin{equation}\label{eq:express_multi_exec_indiv}
        (\sigma_1, \sigma_2) \mapsto (\sem{\prog}(\sigma_1), \sem{\prog}(\sigma_2))
    \end{equation} 
\end{description}

Both traits 
are desirable for program analysis and verification.
Compositionality enables scalable program analysis:
a program can be analyzed by decomposing the original analysis into simpler ones over 
smaller subprograms. 
Expression of multiple executions enables compositionality over constructs that repeatedly execute code (e.g., while loops):
to determine the behavior of a while loop ``$\Ewhile{b}{s}$'' on an input state $\sigma$ compositionally, 
one must understand how $b$ and $s$ behave on every state that appears in a loop iteration. 
For example, to see that 
$\sem{\Ewhile{x < 2}{x := x + 1}}\statelist{\stateentry{x}{0}}$ $=\statelist{\stateentry{x}{2}}$, 
one must observe multiple executions of the guard and body:

\begin{center}
\setlength{\tabcolsep}{8pt}
\begin{tabular}{c c c} 
    $\sem{x {< 2}}\statelist{\stateentry{x}{0}} ~=~ t$
    &$\longrightarrow$
    &$\sem{x := x {+} 1}\statelist{\stateentry{x}{0}} ~=~ \statelist{\stateentry{x}{1}}$\\
    & $\swarrow$ &\\
    $\sem{x {< 2}}\statelist{\stateentry{x}{1}} ~=~ t$
    &$\longrightarrow$
    &$\sem{x := x {+} 1}\statelist{\stateentry{x}{1}} ~=~ \statelist{\stateentry{x}{2}}$\\
    & $\swarrow$ &\\
    $\sem{x {< 2}}\statelist{\stateentry{x}{2}} ~=~ f$& &
\end{tabular}
\end{center}

Although one might take the power to express
multiple executions for granted in this single-program semantics, 
it turns out to be surprisingly hard to design a semantics of \emph{sets} of programs that can express multiple executions for multiple programs \textit{at once}. 
The reason is that 
``natural'' extensions lifting single-program semantics to
semantics for sets of programs tend to mix the outputs of 
different programs together, making it hard to tell 
when two possible outputs on two different inputs can be produced by the same program. 
As a result, it is hard to design a compositional semantics for a set of programs when the programs can contain loops.

\paragraph{Challenges in Designing a Denotational Semantics for Sets of Imperative Programs}

Taking inspiration from the properties of denotational semantics of individual programs, this paper addresses the following objective:

\begin{mybox}
    \textit{Design a \textbf{compositional} semantics for \textbf{sets of programs} that can express \textbf{extensional properties}.}
\end{mybox}
  The simplest extensional properties are 
  \emph{single-input extensional properties}, 
  which for a single program $c$, 
  are the facts
  % properties
  that can be
  % decided
  deduced
  by computing $\sem{c}(\state)$ on some
  state $\state$.
  To generalize single-input extensional properties to a set of programs $C$, our first attempt is a straightforward lifting of the single-program semantics in \Cref{eq:indiv_semantics}---which we call $\semsingset{\cdot}$, the ``program-agnostic'' semantics for sets of programs. 
Whereas $\sem{\prog}$ has the signature $\State \rightarrow \State$, 
the signature of 
$\semsingset{\progset}$ must be 
$2^{\State} \rightarrow 2^{\State}$ 
because $\progset$ may generate multiple outputs for a single input 
(i.e., one for each $\prog \in \progset$).
% Consequently, 
To maintain compositionality on, e.g., sequential composition ($\Eseq{S_1}{S_2}$),
$\semsingset{\progset}$ must have signature 
$2^{\State} \rightarrow 2^{\State}$---that is, 
\emph{single-input} extensional properties require reasoning about 
\emph{sets} of inputs in the case of sets of programs.
Then on a given set of input states $P$, $\semsingset{\progset}$ 
returns the set of all output states that programs in $\progset$ can generate: 
\begin{equation}
\label{eq:agnostic-sem-intro}    
\semsingset{\progset} = \lambda P. \{\sem{\prog}(p) \mid \prog \in \progset, p \in P\}
\end{equation}
(In this paper, we will use $\abstrs{}$---typically with a distinguishing superscript, such as 
$\semsingset{\progset}$---to denote a 
semantics that operates over a set of programs.)
$\semsingset{\progset}$, as defined in \Cref{eq:agnostic-sem-intro},
has been used implicitly in prior work as the basis of relatively incomplete 
Hoare-style logics for sets of programs without loops~\cite{uls, ulw}. 
Crucially, \Cref{eq:agnostic-sem-intro} is the \emph{least-expressive} semantics
(formally, the \emph{coarsest-grained} semantics [see \Cref{def:granularity}]) 
that captures single-input extensional properties of sets of programs.
In other words, \Cref{eq:agnostic-sem-intro} is the least-expressive semantics 
that can specify the set of all possible output states 
that can be produced by running some program in $\progset$ on a single given input state $\sigma$. 
Thus, it is reasonable to say that a useful verification technique 
for sets of programs should capture $\semsingset{\cdot}$ at a minimum.

%As with single programs, achieving compositionality for loops requires a semantics to be capable of expressing multiple executions. 
% 
%For sets of programs, however, the ability to 
%express multiple executions is not free; 
%instead, it demands a surprising strengthening of our 
%semantics' expressiveness to avoid entangling the outputs of distinct programs.

As with $\sem{\cdot}$, the semantics for sets 
$\semsingset{\cdot}$ must also be capable of 
expressing multiple executions to achieve compositionality for loops.
However, for sets of programs, $\semsingset{\cdot}$ is surprisingly \emph{too weak} to be 
compositional for loops: it cannot capture the semantics of multiple executions with sufficient precision.
The first contribution of this paper is a
formal
proof that  although the semantics $\semsingset{\cdot}$ 
is compositional for sets of loop-free programs (\Cref{thm:ag_loop_free_comp}),
$\semsingset{\cdot}$ is not compositional for sets of programs that are allowed to contain loops (\Cref{thm:sing_noninductive}).
The non-compositionality of $\semsingset{\cdot}$ is due to the fact that, 
for a set of loops $\Ewhile{B}{S}$, 
it is not always possible 
to derive $\semsingset{\Ewhile{B}{S}}$ from only $\semsingset{B}$ and $\semsingset{S}$, 
because $\semsingset{\cdot}$ cannot \textit{precisely express multiple executions} 
of individual programs in $B$ or $S$.

To see why $\semsingset{\cdot}$ cannot precisely express multiple executions, suppose that the set $S = \{{\color{orange}s_1}, {\color{cyan}s_2}\}$ contains two statements, 
and we are given two input states $\sigma_i$ and $\sigma_j$. 
We want to combine $\semsingset{S}(\sigma_i)$ and $\semsingset{S}(\sigma_j)$ to produce the set $\{(\sem{{\color{orange}s_1}}(\sigma_i), \sem{{\color{orange}s_1}}(\sigma_j)), (\sem{{\color{cyan}s_2}}(\sigma_i), \sem{{\color{cyan}s_2}}(\sigma_j))\}$, which tells us what can happen if we run each program in $S$ on both $\sigma_i$ and $\sigma_j$.
% 
% To mimic the strategy in \Cref{eq:express_multi_exec_indiv}, one would want to match 
% each output of $\semsingset{S}(\sigma_1)$ to an output of $\semsingset{S}(\sigma_2)$. 
In \Cref{eq:express_multi_exec_indiv}, our semantics produced a single output state, so we directly combined the results into a tuple of states. Because $\semsingset{\cdot}$ returns a \emph{set} of output states, we must take the Cartesian product of $\semsingset{S}(\{\sigma_i\})$ and $\semsingset{S}(\{\sigma_j\})$ to get a set of tuples of states:
\begin{align*}
(\sigma_i, \sigma_j) \mapsto& \semsingset{S}(\{\sigma_i\}) \times \semsingset{S}(\{\sigma_j\})\\
&= \{\sem{{\color{orange}s_1}}(\sigma_i), \sem{{\color{cyan}s_2}}(\sigma_i)\} \times \{\sem{{\color{orange}s_1}}(\sigma_j), \sem{{\color{cyan}s_2}}(\sigma_j)\}\\
&= \{(\sem{{\color{orange}s_1}}(\sigma_i), \sem{{\color{orange}s_1}}(\sigma_j)), (\sem{{\color{orange}s_1}}(\sigma_i), \sem{{\color{cyan}s_2}}(\sigma_j)), (\sem{{\color{cyan}s_2}}(\sigma_i), \sem{{\color{orange}s_1}}(\sigma_j)), (\sem{{\color{cyan}s_2}}(\sigma_i), \sem{{\color{cyan}s_2}}(\sigma_j))\}    
\end{align*}

The {\color{orange}orange}-only $(\sem{{\color{orange}s_1}}(\sigma_i), \sem{{\color{orange}s_1}}(\sigma_j))$ and {\color{cyan}cyan}-only $(\sem{{\color{cyan}s_2}}(\sigma_i), \sem{{\color{cyan}s_2}}(\sigma_j))$ outcomes reflect the programs {\color{orange}$s_1$} and {\color{cyan}$s_2$}, but the outcomes with two colors---i.e., $(\sem{{\color{orange}s_1}}(\sigma_i), \sem{{\color{cyan}s_2}}(\sigma_j))$ and $(\sem{{\color{cyan}s_2}}(\sigma_i), \sem{{\color{orange}s_1}}(\sigma_j))$---do not correspond to any single program in $S$. 
The semantics $\semsingset{\cdot}$ does not track which programs produce which output states (i.e., it is program-agnostic), so it cannot tell whether or not a given pair of output states $(\sigma'_i, \sigma'_j) \in \semsingset{S}(\{\sigma_i\}) \times \semsingset{S}(\{\sigma_j\})$ is produced by the same program $s \in S$.

As a result, $\semsingset{\cdot}$ is not compositional over sets of programs containing loops.
In practice, this
property
means that, given $\semsingset{B}$ and $\semsingset{S}$, 
one cannot distinguish the set of loops $\Ewhile{B}{S}$ 
from the singleton set containing the nondeterministic program $\Ewhile{choose(B)}{choose(S)}$, 
which nondeterministically chooses an expression from $B$ 
and program statement from $S$ on every loop iteration. 
Note that, in the absence of loops, sets of programs can usually be modeled as a single 
nondeterministic program~\cite{nope}.

\paragraph{Compositional Denotational Semantics for Sets of Imperative Programs}
To resolve the noncompositionality of $\semsingset{\cdot}$, we must design a more expressive semantics that explicitly tracks
the effects of programs on multiple inputs.  
Our next contribution is a ``vector-state semantics,''
denoted by $\semweakset{\cdot}$,
which meets the above goal by mapping sets of input vectors of states $P$ to sets of output 
vectors that programs in $\progset$ can produce when applied to each entry of an input vector:
\[\semweakset{\progset} = \lambda P. \{[\sem{\prog}(\sigma_1), \cdots, \sem{\prog}(\sigma_n)] \mid \prog \in \progset, [\sigma_1, \cdots, \sigma_n] \in P\}\]

Unlike the program-agnostic semantics $\semsingset{\cdot}$, the vector-state semantics $\semweakset{\cdot}$ is capable of precisely expressing multiple executions.
The key difference is that the vectorized input-states group inputs together, preventing the loss of precision that occurred by taking Cartesian products when we attempted to 
express multiple executions using $\semsingset{\cdot}$.
As an illustration, 
taking the set of statements $S = \{{\color{orange}s_1}, {\color{cyan}s_2}\}$ as before, the vectorized semantics yields the following behavior:
\begin{align*}
    (\sigma_1, \sigma_2) \mapsto& \semweakset{S}(\{[\sigma_1, \sigma_2]\}) \\
    &= \{[\sem{s}(\sigma_1), \sem{s}(\sigma_2)] \mid s \in S\}\\
    &= \{[\sem{{\color{orange}s_1}}(\sigma_1), \sem{{\color{orange}s_1}}(\sigma_2)], 
         [\sem{{\color{cyan}s_2}}(\sigma_1), \sem{{\color{cyan}s_2}}(\sigma_2)]\}
\end{align*}

The main contribution of this paper is to show that $\semweakset{\cdot}$ is compositional, even over sets of loops (\Cref{thm:weak_compositional_yel}) and is the coarsest-grained (i.e., least expressive)
compositional semantics that is stronger than $\semsingset{\cdot}$ (\Cref{thm:weak_minimality_yel}). Thus, the vector-state semantics $\semweakset{\cdot}$ is the simplest semantics that can serve as the foundation of a compositional verification framework for sets of programs.

Taken together, our results show that a complete, compositional verification framework for sets of programs that can capture single-input extensional properties must capture at least $\semweakset{\cdot}$. This result sets a lower bound on the expressivity of complete, compositional verification frameworks for sets of programs. 
Whether a framework abandons completeness, abandons compositionality, removes loops from the language under consideration, or adopts the complexity of the vector-state semantics $\semweakset{\cdot}$ is a design choice to be made on a case-by-case basis.

\paragraph{Relationship to Program Synthesis}
Readers familiar with program synthesis may wonder how single-input extensional properties 
tie in with program-synthesis problems, where specifications are often written over 
multiple examples, as in programming-by-example (PBE)~\cite{sketch, flashfill}.
Single-input extensional properties correspond to single-example synthesis problems. 
Later, in \Cref{sec:weak-vector-sem}, we discuss how 
% the development of 
the vector-state semantics $\semweakset{\cdot}$ 
can be viewed as an extension of $\semsingset{\cdot}$ that supports reasoning about finitely many inputs simultaneously, 
allowing the vector-state semantics $\semweakset{\cdot}$ to model multiple-input PBE problems as well.
% yields a natural extension of $\semsingset{\cdot}$ 
% % capable of 
% that naturally supports reasoning about a fixed number of
% inputs simultaneously, 
% which in turn allows the semantics to model multiple-input PBE problems.
% In this paper, we provide a comprehensive study of denotational semantics for sets of imperative programs.
% % 
% We discuss two flavors of the ideas discussed earlier:
% \rone semantics for capturing single-input extensional properties of sets of programs ($\semsingsetyel{\cdot}$ and $\semweaksetyel{\cdot}$), and
% \rtwo semantics for capturing both single-input extensional properties and nontermination  ($\semsingsetgrn{\cdot}$ and $\semweaksetgrn{\cdot}$).

\paragraph{Contributions} 
% To summarize, we make the following contributions in this paper:
Although we focus on single-input extensional properties, our work presents a \emph{general framework} for determining the minimum amount of information that must be tracked for a verification technique to be both compositional and relatively complete. In particular, we make the following contributions:

\begin{itemize}
    \item We introduce formalisms for reasoning about denotational semantics of sets of programs, including a notion of \emph{granularity} to capture the relative expressivity of various semantics (\S\ref{Se:prelim}).
    \item We introduce two program-agnostic
    semantics: $\semsingsetyel{\cdot}$ and $\semsingsetgrn{\cdot}$. The $\semsingsetyel{\cdot}$ semantics captures exactly single-input extensional properties as described above, and $\semsingsetgrn{\cdot}$ extends $\semsingsetyel{\cdot}$ with the ability to capture nontermination (\S\ref{Se:prelim}).
    \item We show that, because $\semsingsetyel{\cdot}$ and $\semsingsetgrn{\cdot}$ cannot express multiple executions, neither semantics is compositional over sets of programs when loops are allowed (\S\ref{sec:compositional-semantics}).
    \item We define two more-expressive semantics that
    operate over vector-states,
    $\semweaksetyel{\cdot}$ and $\semweaksetgrn{\cdot}$, and prove that $\semweaksetyel{\cdot}$ and $\semweaksetgrn{\cdot}$ are the least-expressive \textit{compositional} semantics 
    that are at least as expressive as
    $\semsingsetyel{\cdot}$ and $\semsingsetgrn{\cdot}$, respectively (\S\ref{Se:modestgoal}).
    \item We use the semantics $\semweaksetyel{\cdot}$ and $\semweaksetgrn{\cdot}$ to design the \emph{first two relatively complete, compositional Hoare-style 
    % program
    logics for sets of programs}, a goal not achieved
    % that has not been achieved
    by previous work (\S\ref{sec:proof_systems})~\cite{uls, ulw}.
    % \item \revision{Although we focus on single-input extensional properties}, our work constitutes a \emph{general framework} 
    % for determining 
    % the minimum amount of information that must be tracked for \revision{a verification technique to be both compositional and relatively complete}.

\end{itemize}
\noindent
Given a verification technique for sets of programs, our framework lets us analyze its expressivity and compositionality by considering the semantics defined by the set of facts that the technique can prove.
% Our framework
% % over semantics
% can 
% % also
% be used to analyze the expressivity and compositionality of verification techniques
% for sets of programs 
% by considering a semantics  
% \revision{comprised} of the set of facts that
% % can be proven using 
% the verification technique \revision{can prove}.
In \Cref{sec:case_studies}, we give two concrete examples, where 
we expose and resolve the limitations of two different verification techniques by considering 
their associated semantics.

% \revision{
% Our results and framework over semantics can be used to 
% analyze \snnchanged{arbitrary} verification techniques. 
% Specifically, a given technique induces a semantics
% corresponding to the set of facts the technique can prove.
% \S\ref{sec:case_studies} shows how this works, 
% taking two 
% % examples of 
% verification techniques for sets of programs and
% % by inducing semantics from the techniques, 
% using our theory to expose the techniques' limitations and motivate solutions.}
\Cref{Se:related} discusses related work. \Cref{Se:conclusion} concludes. 
The relative granularities of the semantics we introduce are 
depicted in \Cref{fig:granularity-lattice}.
%\noappendix{
%  An extended version of this paper is available on arXiv~\cite{semarxiv}, which 
%  includes separate appendices that provides proofs of major theorems and 
%  rigorous formalizations of some concepts
%  %(such as $\semsingsetgrn{\cdot}$)
%  that are simplified in this paper for 
%  presentation.
%  We have clearly marked places where referring to the full version will provide more details.
%}
\noappendix{
%Proofs of major theorems,
%are included with the extended version of the paper
%as a separate appendix, available on Arxiv~\cite{semarxiv}.}
Proofs of major theorems appear in 
Appendices~\ref{app:proofs} and \ref{app:lattice}. 
Appendices~\ref{app:interpreter} and \revision{\ref{app:VectorStateSemantics}}
elaborate ideas introduced 
in the text about interpreters and the $\semsingsetgrn{\cdot}$ semantics, respectively.
%\noappendix{\xspace (Appendices are available in the extended version of the paper~\cite{semarxiv}.)}
}

\section{From Verification Frameworks to Semantics and Back}\label{sec:case_studies}

In this section, we demonstrate the utility of our theory by deriving semantics 
for two verification techniques for sets of programs: program logics~\cite{uls,ulw} for sets of programs 
(\Cref{subsec:setlogic}) and Hoare-style reasoning about an interpreter on sets of programs (\Cref{subsec:interpreter}).
In particular, our theory can 
\rone prove non-obvious tradeoffs between expressivity and completeness for these verification strategies, and 
\rtwo inspire new complete versions of these techniques using the vector-state semantics $\semweakset{\cdot}$. 

%\begin{remark}
  We note that in this paper we use the word `program' to 
  refer to not only statements, but also Boolean and integer expressions---e.g., both 
  the statement $\Eassign{x}{x + 3}$ 
  and the expression $x + 3$ are referred to as programs.
  Later, when we define semantics for programs (\Cref{def:semantics_prog}) and 
  sets of programs (\Cref{def:semantics_set}), we mean that the 
  semantics is not limited to statements.
  %In addition, the semantics we will define in this paper for Boolean and integer expressions 
  %will be of type $\State \rightarrow \State$ instead of 
  %type $\State \rightarrow \set{\Et, \Ef}$ and $\State \rightarrow \mathbb{Z}$, 
  %with the convention that Boolean and integer expressions instead respectively 
  %update special reserved variables $b_t$ and $e_t$.
  %For example, 
  %$\sem{x }$
%\end{remark}

\begin{figure}
{\centering
    % \vspace{-20mm}
    \includegraphics[width=0.7\linewidth]{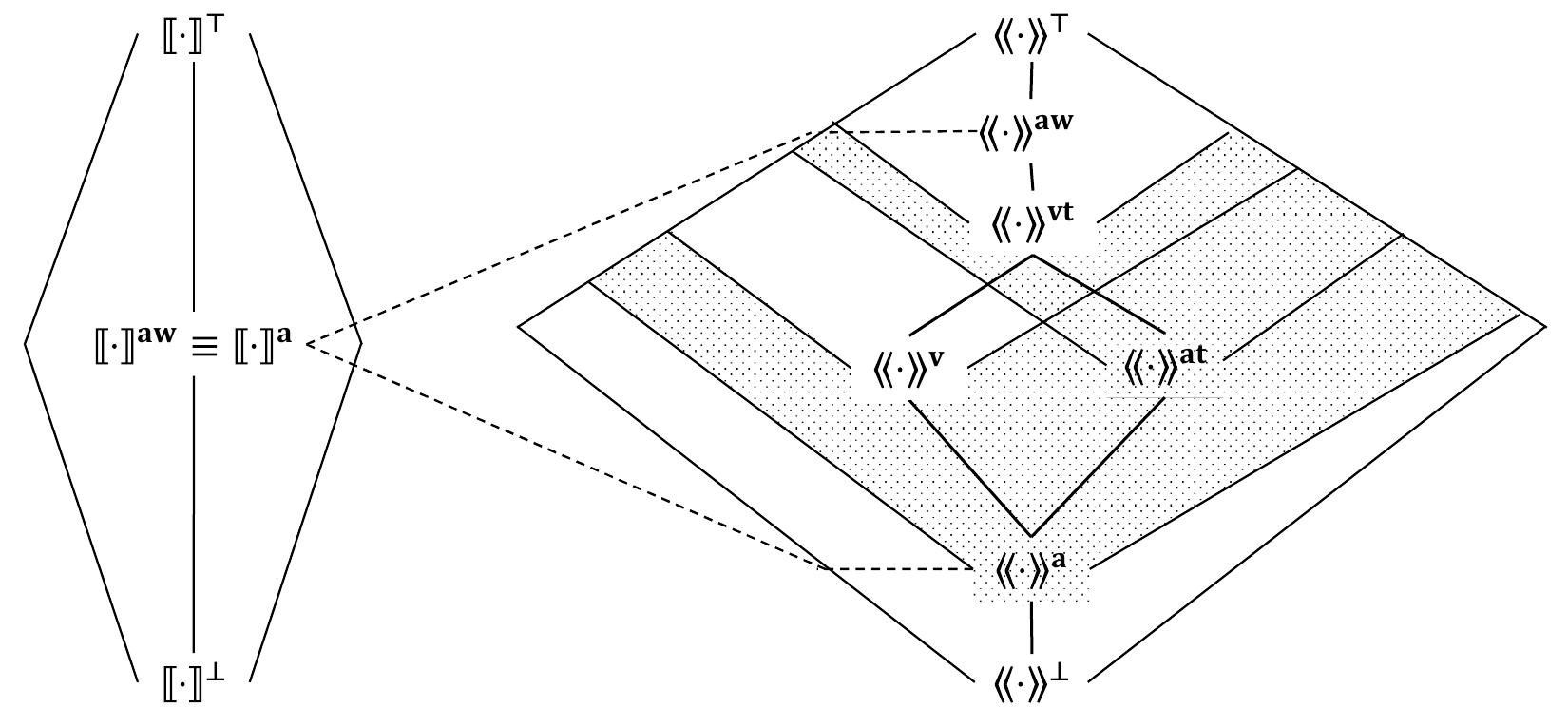}
    % \vspace{-27mm}
    \caption{
      Partial orders that relate the different semantics discussed in the paper using granularity (finest at the top).
      The lattice on the left is for semantics of single programs; the lattice on the right is for semantics of sets of programs.
      Each semantics in each lattice is allowed to have a different output type.
        Semantics are ordered based on the granularity of the equivalence classes their denotations define 
        %induces an equivalence class on its set of inputs (either programs or sets of programs)
        % (either programs or sets of programs; 
        (see \Cref{def:granularity}).
        % for a formal definition).     
      %The different semantics are still comparable based on \emph{granularity}:
      %e.g., each semantics $\sem{\cdot}$ partitions the set of programs into equivalence classes of programs that have identical behaviors under $\sem{\cdot}$.
      %The granularity ordering is \emph{partition refinement}: a semantics that produces a
      %more-refined partition 
      %\jinwoochanged{
      % is finer.
      %}
      %(On the right, partitions are over the set of sets of programs.)
      The top semantics $\sem{\cdot}^{\top} $ and $\abstrs{\top}$ are the respective (syntactic) identity functions---i.e., $\sem{\cdot}^{\top} = \lambda \prog. \prog$ and $\abstrs{\top} = \lambda \progset. \progset$.
      The bottom semantics $\sem{\cdot}^{\bot}$ and $\abstrs{\bot}$ indicate constant semantics that return the same element for all (sets of) programs.
      The partial order for the semantics of sets of programs has finer structure than the partial order for the semantics of single programs:
      there exist elements between $\semfullset{\cdot}$ and $\semsingset{\cdot}$.
      Moreover, 
      between the \yellow semantics $\semsingsetyel{\cdot}$ and $\semweaksetyel{\cdot}$ and between the \green semantics $\semsingsetgrn{\cdot}$ and $\semweaksetgrn{\cdot}$, there exist two overlapping crosshatched zones in which all semantics are noncompositional.
      % \sn{Minor point -- $\semfullset{\cdot} \neq \semstrongset{\cdot}$ when sets have infintely many variables. Let's correct the diagram.}
      \label{fig:granularity-lattice}
      \vspace{-3mm}
}
    }
\end{figure}
\subsection{A Program Logic for Sets of Programs}
\label{subsec:setlogic}

% 1. Define the goal
A program logic (e.g., Hoare Logic~\cite{hoare}) is a proof system for reasoning about program properties. 
In this case study, the goal is to design a program logic $L$ for reasoning about 
single-input extensional properties of sets of programs.
Previous attempts at designing such a logic~\cite{uls, ulw} have resulted in incomplete logics, 
due to implicitly basing the logics on $\semsingset{\progset}$ from \Cref{sec:intro}; we illustrate 
why exactly this is the case, providing an outline here and more formal details in \Cref{sec:proof_systems}.

To design such a logic, one needs to define 
\rone some kind of meaningful judgments ($J(\progset)$) about sets of programs $\progset$, and 
\rtwo inference rules that allow the logic to prove new judgments
from known judgments or axioms, for example,
\begin{equation}
\label{eq:while-eq}
\infer[]
{\vdash J_3(\Ewhile{B}{S})}
{\vdash J_1(B) \quad \vdash J_2(S)}
\end{equation}
where $B$ is a set of Boolean expressions, $S$ is a set of statements.
%and $\Ewhile{B}{S} = \{\Ewhile{b}{s} \mid b \in B \land s \in S\}$. 
Critically, the inference rules should be compositional:
%in the sense that, 
for example, 
any valid judgment about the set $\Ewhile{B}{S}$ 
should be derivable from valid judgments about $B$ and $S$, as in \Cref{eq:while-eq}.

% 2. State existing work & deficiencies
\paragraph{Existing Logics' Judgments are Noncompositional}
In the logic proposed by Nagy et al.~\cite{ulw}, judgments are triples, 
denoted by $\utriple{P}{\progset}{Q}$, where $\progset$ is a set of programs, and $P$ and $Q$ are sets of states. 
In terms of semantics,
the meaning of
this kind of triple can be written as follows:
\begin{equation}
  \label{Eq:TripleToSingSet}
  \utripleag{P}{\progset}{Q} \defeq \semsingset{\progset}(P) \subseteq Q
\end{equation}
Other prior work interprets triples slightly differently but relies on the same intuition~\cite{uls}.

Observe that a denotational semantics over a set of programs is just a map that assigns to each 
set of programs $\progset$ a mathematical object 
associated with $\progset$
(i.e., its denotation).
Thus, one can also define a
``logic-based''
semantics $\semulagset{\cdot}$ that associates with $\progset$ 
 the set of valid judgments about $\progset$:
% \[
%   \semulagset{\progset} = \set{\utripleag{P}{-}{Q} \mid~ \utripleag{P}{\progset}{Q} \text{ holds}}
% \]
\begin{equation}
  \semulagset{\progset} \defeq \set{(P,Q) \mid~ \utripleag{P}{\progset}{Q} \text{ holds}}
\end{equation}
%\[
%  \semulagset{\progset} = \{\utripleag{P}{\progset}{Q} \mid~ \utripleag{P}{\progset}{Q} \text{ is true}\},~\textit{where $\utripleag{P}{\progset}{Q}$ is defined by \Cref{Eq:TripleToSingSet}}
%\]
% % 
% To capture all basic extensional properties, $\abstrs{L}$ must capture $\semsingset{\cdot}$. 
% % 
% Indeed, this has been done in prior work 
\noindent
Under this definition, 
the logic-based semantics
$\semulagset{\cdot}$ and 
the program-agnostic semantics $\semsingset{\cdot}$ carry the same information: 
$\semulagset{\progset}$ can be derived from $\semsingset{\progset}$ and vice versa.
%(i.e., they are equally granular).

Previous work has failed to find relatively complete, compositional inference rules for a logic over such triples. Our theory proves that such rules do not exist: because both the program-agnostic semantics $\semsingset{\cdot}$ and the logic-based semantics $\semulagset{\cdot}$ carry the same information, noncompositionality of the program-agnostic semantics $\semsingset{\cdot}$ 
(see \Cref{thm:sing_noninductive}) implies noncompositionality of 
the logic-based semantics $\abstrs{La}$.
The non-compositionality of $\semulagset{\cdot}$ implies that there is
no complete system of compositional inference rules for such triples (see \Cref{thm:no-compo-prog-ag-while})!

More concretely, there are sets of Boolean guards $B$ and loop bodies $S$ for which there exist true triples $\utripleag{P}{\Ewhile{B}{S}}{Q}$ that cannot be proven from any true triples about $B$ and $S$.
For example, if $B = \{x == 1, \neg(x == 1)\}$ then we cannot compositionally prove the triple
\begin{equation}
  \label{Eq:UnprovableTriple}
  \utripleag{x=0}{(\Ewhile{\{x == 1, \neg(x == 1)\}}{\Eassign{x}{x+1}})}{x=0 \lor x=1}
\end{equation}
Without the ability to track the effects of programs on multiple inputs, the set of guards $B$ is indistinguishable from the set of guards $B' = \{x == 2, \neg(x == 2)\}$,
because both sets can produce $\Etrue$ and $\Efalse$ on any input state (i.e., $\semsingset{B} = \semsingset{B'}$ so $\semulagset{B} = \semulagset{B'}$).
In other words, every triple that is true of $B$ is true of $B'$.
However, changing the set of guards from $B$ to $B'$ invalidates the triple in \Cref{Eq:UnprovableTriple} (i.e., $\utripleag{x=0}{(\Ewhile{B'}{\Eassign{x}{x+1}})}{x=0 \lor x=1}$ is false because $x=2$ is a possible outcome). A compositional logic that cannot distinguish
between
$B$ and $B'$ also cannot distinguish
between
$\Ewhile{B}{\Eassign{x}{x+1}}$ and $\Ewhile{B'}{\Eassign{x}{x+1}}$, so no compositional logic can prove the
triple given in \Cref{Eq:UnprovableTriple}.

% 3. Solve the deficiencies
\paragraph{Compositional Triples via Vector-State Semantics}
To design a compositional logic, we need to strengthen our triples. Our new triples must yield a semantics that is at least as expressive as the vector-state semantics $\semweakset{\cdot}$---i.e., the coarsest compositional semantics stronger than the program-agnostic semantics $\semsingset{\cdot}$---otherwise the resulting semantics will be noncompositional like $\semulagset{\cdot}$ (\Cref{thm:weak_minimality_yel}).
Following \Cref{Eq:TripleToSingSet},
we can take $P$ and $Q$ to be sets of vector states, defining triples by \Cref{Eq:TripleToWeakSet} and inducing the vectorized logic-based semantics in \Cref{eq:ul_vs_sem}:
\begin{align}
  \utriplevs{P}{&\progset}{Q} \defeq \semweakset{\progset}(P) \subseteq Q\label{Eq:TripleToWeakSet}\\[0.5em]
\abstrsapp{Lv}{\progset} \defeq &\{(P,Q) \mid~ 
  \utriplevs{P}{\progset}{Q} \text{ holds}
  % \semweakset{\progset}(P) \subseteq Q 
  \}\label{eq:ul_vs_sem}
\end{align}

With this definition, the vectorized logic-based semantics $\semulvsset{\cdot}$ is equivalent to the vector-state semantics $\semweakset{\cdot}$, 
and compositional inference rules follow directly from our constructive proof that 
$\semweakset{\cdot}$ is compositional (\Cref{thm:weak_compositional_yel}). 
Moreover, the above definition is the simplest notion of triples that permits a compositional inference rule for sets of while loops 
(given in \S\ref{sec:vs_have_while_rules}). 
To capture hyperproperties, prior work has proposed similar---but different---semantics over vectors,
but has only produced incomplete program logics~\cite{uls, ulw}.\footnote{There is a flaw in the proof of completeness given in~\citet{uls}, which has been discussed with the authors.\label{note1}
}

\subsection{Analyzing Sets of Programs Through an Interpreter}
\label{subsec:interpreter}
 
\noindent
\lstset{%
    deletekeywords={do, if, then},
    keywordstyle=\color{blue},
    morekeywords={match, with, if, then}
}
\begin{wrapfigure}{r}{0.4\textwidth}
  \vspace{-5mm}
  \begin{lstlisting}[linewidth=\linewidth]
State interpret (t: Prog, x: State) {
  match t with 
    ...
    case While b do s: {
      State r;
      State b1 = interpret(b, x);
      if (b1.bt) {
        State x1 = interpret(s, x);
       (*@ \color{red}{r = interpret(While b do s, x1);}\label{line:recWhile}@*)
      }
      else {
        r = x;
      }
      return r;
    }
    ...
}
  \end{lstlisting}
  \vspace{-4mm}
  \caption{An interpreter for
  % $\Ewhile{B}{S}$ for 
  the language described in \Cref{fig:gimp}.
  % We assume that $\mathtt{State}$ and $\mathtt{Term}$ that 
  % respectively capture the types of states and terms.
  The symbol $\mathtt{bt}$ indicates a reserved variable in $\mathtt{State}$ 
  that stores the result of evaluating Boolean expressions.
}
  \label{fig:interpreter}
  \vspace{-4mm}
\end{wrapfigure}
Readers might feel that compositionality is an issue only because we are defining our program logic over sets of programs. 
One way to avoid lifting our reasoning to sets is to instead perform a Hoare-style analysis of an interpreter 
% that takes as input \lorischanged{programs in} the set $\progset$
(e.g., $\itrp$ in \Cref{fig:interpreter}).\footnote{
    Because $\itrp$ is a  recursive function, one technically needs to consider proofs of infinite size or apply recursive Hoare logic~\cite{nipkow}. These approaches require additional steps (e.g., introducing an induction hypothesis). We ignore these details to simplify the presentation; 
    %\noappendix{the extended version of the paper~\cite{semarxiv} 
    %gives a more thorough treatment.}
    \Cref{app:interpreter} gives a more thorough treatment.
    %\noappendix{\xspace in the extended version of the paper~\cite{semarxiv}}.
  }
Rather than analyzing $\progset$ directly as in 
\Cref{subsec:setlogic}, one would carry out an inductive argument about the behavior of $\itrp$ over $\progset$.
The hope is that because $\itrp$ is a single program,
an analysis that is compositional over $\itrp$ would avoid the complexity of an analysis that is compositional over $\progset$ as in \Cref{subsec:setlogic}.

Using the theory presented in this paper, we can \textit{prove} that this approach does not avoid the hardness of compositional reasoning over sets of programs.
In particular, we consider a Hoare-style analysis of the interpreter $\itrp$
where the preconditions are 
limited to ``Cartesian'' predicates: 
predicates of the form 
$\varphi_1(t) \wedge \varphi_2(x)$, where 
$\varphi_1$ is forbidden from referencing the state $x$ and 
$\varphi_2$ is forbidden from referencing the program $t$.
Maintaining this separation is necessary to
prevent one from writing 
preconditions that reason about interpretation 
explicitly,
which would defeat the purpose of analyzing an interpreter.

In this section, we demonstrate that Hoare-style analysis limited to Cartesian preconditions induces a semantics $\abstrs{\itrp_a}$ similar to $\abstrs{La}$. 
The completeness of Hoare-style analysis limited to Cartesian preconditions would imply 
compositionality of the induced semantics $\abstrs{\itrp_a}$, but $\abstrs{\itrp_a}$ is not compositional. 
Thus, such an analysis
%an analysis over Cartesian preconditions of sets of programs 
over recursive interpreters of the type
$\itrp: \Prog \times \State \rightarrow \State$ is bound to be incomplete.
On the other hand, the %aforementioned 
analysis may be complete if
%analysis with Cartesian preconditions is complete if 
we lift the interpreter to act on vector-states 
(i.e., $\itrp: \Prog \times \State^* \rightarrow \State^*$), 
because the existence of a compositional vector-state semantics as mentioned in \Cref{sec:intro} 
implies that a compositional analysis must exist.
We sketch the main idea in this section,
%\noappendix{a rigorous formalization may be found in the extended version of the paper~\cite{semarxiv}.}
and 
provide further details in \Cref{app:interpreter}.
%\noappendix{\xspace of the extended version of the paper~\cite{semarxiv}}.

\paragraph{Noncompositionality of a Simple Interpreter}
To analyze the extensional properties of a syntactically defined set of programs $\progset$, 
one can compute $\sem{\itrp}(\progset \times P)$ 
given an arbitrary set of inputs $P$.
This approach is equivalent to deriving a Hoare triple of the form:
\[
  \triple{t \in \progset \land x \in P}{y = \itrp(t, x)}{y \in Q}.
\]
For example, to understand the possible outputs of the set of loops $W$ from \Cref{ex:motive} on the input $\statelist{\stateentry{x}{0}}$, one can 
% compute \[\sem{\itrp}(W \times \{\statelist{\stateentry{x}{0}}\}) = \{\statelist{\stateentry{x}{2n}} \mid n \in \mathbb{N}\}\] or 
prove
\[
  \triple{t \in W \land x = 0}{y = \itrp(t, x)}{y = 0\pmod 2}.
\]

However, reasoning about such triples can be complicated.
Cartesian preconditions of the form $t \in \progset \land x \in P$
can check membership of an input $(t, x)$ in the precondition by checking 
$t \in \progset$ and $x \in P$ \emph{independently}.
Critically, because $\progset$ is defined syntactically and 
$P$ cannot reference $t$, one cannot relate the syntax of $t$ to its semantics $\sem{t}$ inside the precondition.
Observe that Cartesian preconditions have a direct correspondence with the program-agnostic semantics $\semsingset{\cdot}$: 
a triple
$\triple{t \in \progset \land x \in P}{\itrp(t, x)}{Q}$ holds iff $\semsingset{\progset}(P) \subseteq Q$.
Unfortunately, when $\mathtt{t}$ contains a loop, such Cartesian preconditions often cannot be proven without using substantially more complicated 
``non-Cartesian'' preconditions that relate $\mathtt{t}$ and $\mathtt{x}$.
In this situation, we no longer have a direct correspondence between Hoare triples and $\semsingset{\cdot}$, which corresponds to reasoning about complex properties outside the program-agnostic semantics $\semsingset{\cdot}$.
% implies that one has to reason about complex properties outside what is expressible with the program-agnostic semantics $\semsingset{\cdot}$.

For example, consider proving the triple $\triple{t \in \Ewhile{B}{S} \land x \in P}{y = \itrp(t, x)}{y \in Q}$.
Analysis is straightforward until \Cref{line:recWhile} of \Cref{fig:interpreter}, emphasized in {\color{red}red}, where the semantics of a loop 
requires that the same $\mathtt{b} \in B$ and $\mathtt{s} \in S$ be 
\emph{repeated} for each iteration of a loop.
Upon reaching \Cref{line:recWhile}, there is a relation formed with the input state and the program:
$\mathtt{x}_1$
is created using a specific combination of $\mathtt{b} \in B$ and $\mathtt{s} \in S$, such that 
$\sem{\mathtt{b}}(\mathtt{x}_0) = \Etrue$ 
and
$\sem{\mathtt{s}}(\mathtt{x}_0) = \mathtt{x}_1$,
for some input state $\mathtt{x}_0$.
Thus, the precondition over which the recursive call to $\itrp$ must operate becomes:
%can be described as follows:
\begin{equation}
  \label{eq:interpret_precondition_two}
(\mathtt{t}, \mathtt{x}) \in \set{
  (\mathtt{while \ b \ do \ s}, \texttt{x}_1) \mid 
   \mathtt{b} \in B \wedge \mathtt{s} \in S \wedge 
   (\exists \mathtt{x}_0 \in P. 
   \sem{\mathtt{b}}(\mathtt{x}_0) = \Etrue
   \wedge
   \sem{\mathtt{s}}(\mathtt{x}_0) = \mathtt{x}_1)
   }
\end{equation}

\noindent
Although $t \in \Ewhile{B}{S} \land x \in P$ was initially a Cartesian precondition, 
\Cref{eq:interpret_precondition_two} is not---the additional condition 
$\exists \mathtt{x}_0 \in P.     
\sem{\mathtt{b}}(\mathtt{x}_0) = \Etrue \wedge \sem{\mathtt{s}}(\mathtt{x}_0
) = \mathtt{x}_1$ 
forces our precondition to reason about both the syntax of our programs (e.g., $s \in S$) as well as their semantic effects (i.e., $\sem{\mathtt{s}}(\mathtt{x}_0
) = \mathtt{x}_1$). 
This entanglement requires a complex notion of `interpretation' $\sem{\cdot}$ in our precondition that one would expect analysis of an interpreter to avoid.
% relates the input program ``$\Ewhile{b}{s}$'' to the input state $x_1$.

% The necessity of reasoning 
Reasoning about non-Cartesian preconditions 
is actually an unavoidable feature of every 
correct recursive implementation of $\itrp: \Prog \times \State \rightarrow \State$
(\Cref{thm:itrp_noncart}).
%\noappendix{(\Cref{thm:itrp_noncart}).}
% By recursive, we mean that $\itrp$ only accesses $\mathtt{t}$ by \rone 
% matching on the top-level constructor of $\mathtt{t}$, and 
% \rtwo calling $\itrp$ recursively on the subprogams of $\mathtt{t}$. 
To understand why this observation holds,
consider an interpreter-based semantics 
$\abstrs{\itrp_a}$ akin to $\abstrs{La}$ in \Cref{subsec:setlogic}:
%we have a version of vectorized-interpret semantics (\itrp_v) later
% \begin{equation}
%     \abstrsapp{\itrp_a}{\progset} = \set{\triple{\mathtt{t} \in\ - \wedge P(\mathtt{x})}{y = \itrp(t, x)}{Q} \mid \text{the triple is true when }-=\progset}
% \end{equation}
    \begin{equation}
        \abstrsapp{\itrp_a}{\progset} = \set{(P, Q) \mid \triple{\mathtt{t} \in\ \progset \wedge \mathtt{x} \in P}{y = \itrp(t, x)}{Q} \text{ is true}}
    \end{equation}

If one could prove every Cartesian triple with a Hoare-logic proof using only Cartesian triples, 
one would be able to derive a compositional characterization of $\abstrs{\itrp_a}$. 
However, $\abstrs{\itrp_a}$ is noncompositional
because it is equivalent to the noncompositional program-agnostic semantics $\semsingset{\cdot}$, 
just like the logic-based semantics 
$\semulagset{\cdot}$.
Consequently, reasoning about non-Cartesian preconditions like \Cref{eq:interpret_precondition_two} 
is unavoidable for \emph{any} recursive implementation of $\itrp: \Prog \times \State \rightarrow \State$ 
%(\noappendix{see the extended version of the paper~\cite{semarxiv}} for a formal proof).
(see \Cref{thm:itrp_noncart} for a proof).
%\noappendix{\xspace in the extended version of the paper~\cite{semarxiv}} for a proof).

\paragraph{Complete Interpreter Analysis}
To precisely reason about all facts of the form $\sem{\itrp}(\progset \times P)$ with Hoare logic, one can either 
\rone accept the complexity of non-Cartesian preconditions or 
\rtwo extend the signature of $\itrp$ so that $\mathtt{x}$ is a vector-state rather than a single input 
(i.e., $\itrp_v : \Prog \times \State^* \rightarrow \State^*$). 
In the second case, this new interpreter induces a semantics:
\begin{equation}
    \abstrsapp{\itrp_v}{\progset} = \set{(P, Q) \mid \triple{\mathtt{t} \in\ \progset \wedge \mathtt{x} \in P}{\itrp_v(t, x)}{Q} \text{ is true}},
\end{equation}
which is equivalent to the vector-state semantics
$\semweakset{\cdot}$.
Because 
%there exists a compositional characterization of 
$\semweakset{\cdot}$ is compositional (\Cref{thm:weak_compositional_yel}),
there will also exist a Hoare-style analysis of $\itrp_v$ 
that stays within the realm of Cartesian preconditions---at the cost of introducing vector-states.
%(\Cref{thm:itrp_vs_cart})---at the cost of introducing vector-states.
% 
Alternatively, one could forego completeness in exchange for 
a simple and compositional analysis.
For example, in \Cref{eq:interpret_precondition_two}, one could disallow
conditions such as 
$\exists \mathtt{x}_0 \in P.     
\sem{\mathtt{b}}(\mathtt{x}_0) = \Etrue \wedge \sem{\mathtt{s}}(\mathtt{x}_0) = \mathtt{x}_1$
and only consider preconditions that are Cartesian.
This approach would enlarge the possible input states to subsequent calls of $\itrp$, resulting in an imprecise, but still sound, analysis.

% \pagebreak
\section{Semantics of Sets of Programs}
\label{Se:prelim}

% The goal of this section is to formalize the following high-level question:
% \textit{To what extent does a formal reasoning system for sets of programs need to be able to distinguish the semantics of distinct programs in a set (i.e., synchronize programs and inputs as discussed in \Cref{Se:intro})?}

Having illustrated 
% how our results on the complexity of semantics relate to 
the implications of our semantic theory for the hardness of
practical 
verification techniques, we now turn to formalizing the ideas presented in 
Sections~\ref{sec:intro} and~\ref{sec:case_studies}.
We start by formally defining the sets of programs we are interested in (\S\ref{sec:ind_sets}), 
proceed to formalize what a semantics over (sets of) programs is mathematically 
(\S\ref{sec:semantics} and \S\ref{sec:full_single_semantics}), then 
define a notion of \emph{granularity} which allows us to formally compare the 
expressive power of different semantics (\S~\ref{sec:gran_and_comp}).

\subsection{Inductively Defined Sets of Programs} \label{sec:ind_sets}

% The definition of the G_imp
\begin{figure}[bt!]
  {\small
  \[
  \begin{array}{lll}
     \Stmtvars & ::= &
      \Eassign{\vars}{\Expvars} \mid \Eseq{\Stmtvars}{\Stmtvars} \mid\\
      &&\quad\Eifthenelse{\Boolvars}{\Stmtvars}{\Stmtvars}  \mid \Ewhile{\Boolvars}{\Stmtvars}\\
    \Expvars & ::= &
      \vars \mid 1 \mid 0 \mid \Expvars + \Expvars \mid \Expvars - \Expvars\\    
    \Boolvars & ::= &
      \Et \mid \Ef \mid \lnot \Boolvars \mid \Boolvars \wedge \Boolvars \mid \Expvars < \Expvars \mid \Expvars == \Expvars\\
      \vars & ::= & 
      x \mid x_1 \mid x_2 \mid \cdots 
  \end{array}\]}
  \vspace{-2mm}
  \caption{The grammar $\gimp$, which defines the entire set of terms considered in this paper.}
  \label{fig:gimp}
\end{figure}

In this paper, we consider a small deterministic imperative language that contains loops, 
defined via $\gimp$ in \Cref{fig:gimp}.
In particular, we are interested in sets of programs (i.e., subsets of $\gimp$)
that are inductively defined by a regular tree grammar.

\def\alphabet{\mathcal{A}}
\def\nonterminals{\mathcal{N}}
\begin{definition}[Regular Tree Grammar]
A \emph{(typed) regular tree grammar} (RTG) is a tuple $G = (\nonterminals,\alphabet, \mathit{Start},a,\delta)$,
where $\nonterminals$ is a finite set of nonterminals;
$\alphabet$ is a ranked alphabet; $\mathit{Start} \in \nonterminals$ is an initial nonterminal;
$a$ is a type assignment that gives types for members of $\alphabet \cup \nonterminals$;
and $\delta$ is a finite set of productions of the form
$N_0 \rightarrow \alpha^{(i)}(N_1,...,N_i)$,
where for $0 \leq j \leq i$, each $N_j \in \nonterminals$ is a nonterminal
such that if $a_{\alpha^{(i)}} = (\type_0,\type_1,...,\type_i)$ then $a_{N_j} = \type_j$.
\end{definition}

\begin{definition} [Inductively Defined Set of Programs]
  \label{def:inductive-programs}
    Given a set of programs $\progset$ and an RTG $G$ with start nonterminal $N$, we say that $\progset$ is \emph{inductively defined by} $G$
    if $\progset = L(N)$.
\end{definition}

For this paper, the alphabet of an RTG consists of constructors for each of the constructs of $\gimp$.
For simplicity of presentation, 
we elide the arity of constructors and use infix notation (e.g., $\Eassign{x}{0}$, rather than $:=^{(2)} (x^{(0)}, 0^{(0)})$)) and 
also allow inlining of nonterminals 
(e.g., $S' ::= \Eassign{x_1}{x_1}$ instead of the three productions 
$S' ::= \Eassign{V_x}{E_x}$, $E_x ::= V_x$, and $V_x ::= x_1$). 
We assume grammars use productions that differ from the ones in $\gimp$ only due to the nonterminal name
(i.e., the constructor $\op$ of a production must be from $\gimp$ and have the same arity).
We call a production $N_0 ::= \alpha^{(i)}(N_1,\cdots,N_i)$ \emph{valid with respect to $\gimp$} if replacing each nonterminal $N_j$ with the nonterminal of the same type in the set $\{\Stmtvars, \Expvars, \Boolvars\}$ yields a
production in $\gimp$ (e.g., production $S_1 ::= S_2; S_3$ is valid with respect to $\gimp$ because $\Stmtvars ::= \Stmtvars; \Stmtvars$
% replacing each $S_i$ with $\Stmtvars$ yields
is a production in $\gimp$).

Given a nonterminal $N$ of a grammar $G$, we will use $N$ itself as shorthand to refer to $L(N)$ 
(the set of terms derivable via the grammar $G$ from the nonterminal $N$). 
Thus, $\Stmtvars$ refers also to the set of statements derivable from $\Stmtvars$. 
Given a term $\prog$, 
we define $\varfun(\prog)$ to be the set of variables appearing in $\prog$---e.g., 
$\varfun(\Eassign{x_1}{x_2})=\{x_1,x_2\}$.
For a set of programs $\progset$, we say 
$\varfun(\progset) = \bigcup\limits_{\prog \in \progset} \varfun(\prog)$. 

We use the (bolded) name
$\Stmtset = \{S \mid S \in 2^\Stmtvars \land \abs{\varfun(S)} < \infty\}$ 
to denote the set of all sets of statements $S$ that each contain finitely many variables. 
$\Expset$ and $\Boolset$ are defined analogously, for expressions 
and Boolean expressions, respectively.
This paper only considers sets of programs $\progset$ where 
$\abs{\varfun(\progset)} < \infty$---sets whose programs reference an infinite number of variables are rare 
and out of scope for this paper.

\subsection{Semantics of Programs} \label{sec:semantics}

We now formalize the definition of a \emph{semantics} of a \emph{single} program.

% In this paper, we are especially interested in comparing different semantics, 
% a goal that will benefit from starting with a broad definition of semantics.

%To do so, we give a very broad definition of semantics for programs and sets of programs. By ``semantics'', we mean a formalism which captures information about a program or set of programs. 
%It can be nice to think about semantics as abstractions over programs/sets of programs.

\begin{definition}[Semantics of a Program]
  \label{def:semantics_prog}
  Let $A(\domain_{\Stmtvars}, \domain_{\Expvars}, \domain_{\Boolvars})$ be a 3-tuple with an associated name $\mathit{A}$ where $\domain_{\Stmtvars}$, $\domain_{\Expvars}$, and $\domain_{\Boolvars}$ are the types of the denotations that the semantics assigns to programs of type $\Stmtvars$, $\Expvars$, and $\Boolvars$, respectively.
  Then a \emph{semantics over programs} $\abstrpnovars{A}$ is 
  a 3-tuple $\abstrpnovars{A} = (\abstrpstmt{A}, \abstrpexp{A}, \abstrpbool{A})$, where:
   $\abstrpstmt{A}: \Stmtvars \rightarrow \domain_{\Stmtvars}$, $\abstrpexp{A}: \Expvars \rightarrow \domain_{\Expvars}$, and $\abstrpbool{A}: \Boolvars \rightarrow \domain_{\Boolvars}$.
   We call $A(\domain_{\Stmtvars}, \domain_{\Expvars}, \domain_{\Boolvars})$ the \emph{signature} of $\abstrpnovars{A}$.
\end{definition}

In this paper, we will only be defining semantics where 
$\domain_\Stmtvars = \domain_\Expvars = \domain_\Boolvars$, by 
interpreting integer and Boolean expressions as state transformers as well.
Specifically, an expression $e \in \Expvars$ is considered to update a reserved variable $e_t$
instead of returning an integer value; 
Boolean expressions update a separate reserved variable $b_t$ in a similar manner.
This formalism makes it easier to relate the result of an expression with its input state, 
which will become convenient later when constructing compositional characterizations for sets of programs, 
as in \Cref{fig:single_state_semantics}.
When $\domain_\Stmtvars = \domain_\Expvars = \domain_\Boolvars = \domain$, we abbreviate the signature $A(\domain_\Stmtvars, \domain_\Expvars, \domain_\Boolvars)$ as $A(\domain)$, 
and we often omit subscripts for $\Stmtvars$, $\Expvars$, and $\Boolvars$ when they are clear from context.

\Cref{def:semantics_prog} allows us to consider multiple different semantics for a program 
by changing the signature $A(\tau)$
and the definition of $\abstrpnovars{A}$.
We will first define a collecting big-step semantics for programs in $\gimp$, called $\semsingprogyel{\cdot}$, which ignores nontermination and is such that $\tau {=} 2^{\State} \rightarrow 2^{\State}$.
We will then introduce a variant of $\semsingprogyel{\cdot}$, denoted by 
$\semsingproggrn{\cdot}$, which handles nontermination explicitly 
by introducing a non-terminating value $\divg$.
We start by defining program states.
\begin{definition}[$\State$ and $\DState$]
\label{def:state}
  A \emph{state} $\sigma$ is defined as a tuple $(h, e_t, b_t)$, where 
  \rone $h$ is a map $h: \vars \rightarrow \Z$ from variables in $\vars$ to integer values, 
  \rtwo $e_t \in \Z$, and \rthree $b_t \in \set{\Et, \Ef}$.
  We denote the set of all possible states as $\State$.
  For a state $\state = (h, e_t, b_t)$, 
  we write $\state[x_j]$ to denote $h(x_j)$, 
  $\state[e_t]$ to denote $e_t$, and $\state[b_t]$ to denote $b_t$.
  $\DState$ denotes the set of all possible states including the divergence behavior $\divg$, 
  and we define it as $\DState = \State \cup \set{\divg}$.
\end{definition}
In \Cref{def:state}, $h$ acts as the `usual' state that maps variables to values---note that 
$h$ is defined over $\vars$, the entire set of possible variables in $\gimp$, as opposed to 
a subset of variables that occur in the program;
this
convention simplifies the definitions of our semantics as well the proofs in \Cref{Se:modestgoal}.
When we only care about the values of specific variables, we
write states as, for example, $\statelist{\stateentry{x_1}{4}, \stateentry{e_t}{3}}$ 
to denote any $\state \in \State$ where $\state[x_1] = 4$ and $\state[e_t] = 3$.
$e_t$ and $b_t$ act as auxiliary variables
that ``store'' the results of evaluating an expression and Boolean expression, respectively.

The value $\divg$ denotes nontermination (i.e., diverging behavior) and is used in the semantics
$\semsingproggrn{\cdot}$.
For example, $\semsingproggrn{\Ewhile{x = 1}{\Eassign{x}{x}}}(\{[x \mapsto 1]\}) = \{\divg\}$.
To make $\semsingproggrn{\prog}$ total, we assume 
$\semsingproggrn{\prog}(\{\divg\}) =\{\divg\}$ for every program $\prog$.

We are now ready to formally define a big-step collecting semantics for terms in $\gimp$.\footnote{We consider 
collecting semantics because the focus on this paper is on the semantics of sets of programs, which do not admit a compositional characterization when limited to single-state input and outputs.}
This semantics is
the single-program version of the ``program-agnostic'' semantics $\semsingset{\cdot}$ from \Cref{sec:intro},
which map
sets of inputs to sets of outputs.
We call this semantics \emph{program-agnostic} because
it does not attempt to track how the program maps inputs to outputs.
This distinction will
become clearer when we switch to talking about sets of programs in \Cref{sec:full_single_semantics}.
\begin{figure*}
{\small
\[
\begin{array}{c}
  \semsingprog{0}(\state)=\state\subs{0}{e_t} \quad
  \semsingprog{1}(\state)=\state\subs{1}{e_t} \quad
  \semsingprog{x_j}(\state)=\state\subs{\state[x_j]}{e_t} \quad
  \semsingprog{\Et}(\state)=\state\subs{\Et}{b_t} 
  \\[2.5mm]
  \semsingprog{\Ef}(\state)=\state\subs{\Ef}{b_t} \quad
  \semsingprog{\neg b}(\state)=\state\subs{\neg (\semsingprog{b}(\state)[b_t])}{b_t} 
  \\[2.5mm]
  \semsingprog{t_1 \odot t_2}(\state)=\state\subs{\state_1[q_{t}] \odot \state_2[q_{t}]}{r_{t}},
  \text{where } \state_1 = \semsingprog{t_1}(\state) \text{ and } \state_2 = \semsingprog{t_2}(\state)
  \\[2.5mm]
  \semsingprog{\Eassign{x}{e}}(\state)=\state\subs{(\semsingprog{e}(\state)[e_t])}{x}\qquad
  \semsingprog{s_1;s_2}(\state) = \semsingprog{s_2}(\semsingprog{s_1}(\state)) \quad 
  \\[2.5mm]
  \semsingprog{\Eifthenelse{b}{s_1}{s_2}}(\state) = \Eifthenelse{\semsingprog{b}(\state)}{\semsingprog{s_1}(\state)}{\semsingprog{s_2}(\state)} 
  \\[2.5mm]
  \semsingprogyel{\Ewhile{b}{s}}(\state) = \Eifthenelse{\semsingprogyel{b}(\state)}{\semsingprogyel{\Ewhile{b}{s}}(\semsingprogyel{s}(\state))}{\state}
  \\[2.5mm]
  \semsingproggrn{\Ewhile{b}{s}}(\state) = 
  \begin{cases}
      \semsingprogyel{\Ewhile{b}{s}}(\state) & \semsingprogyel{\Ewhile{b}{s}}(\state) \neq \emptyset\\
      \{\divg\} & \text{else}\\
  \end{cases}
\end{array}
\]
}
\caption{
  Definition of $\semsingprogyel{\cdot}$ and $\semsingproggrn{\cdot}$ where $b \in \Boolset$, $e, e_1, e_2 \in \Expset$, 
  and $s, s_1, s_2 \in \Stmtset$.
  Recall that $e_t$ and $b_t$ are reserved variables that 
  hold the results of evaluated arithmetic and Boolean expressions, respectively.
  The symbol $\odot$ denotes a binary operator in $\gimp$ (either $+, -, <, ==,$ or $\wedge$),
  $t_1$ and $t_2$ refer to expressions of the appropriate types, and ${r_{t}}$ and $q_{t}$ refer to auxiliary variables of the appropriate type.
  $\semsingprogyel{\cdot}$ and $\semsingproggrn{\cdot}$ are respectively 
  of type $2^{\State} \rightarrow 2^{\State}$ and 
  $2^{\DState} \rightarrow 2^{\DState}$;
  this figure only lists definitions for singleton inputs for readability.
  The definition of $\semsingproggrn{\cdot}$ is 
  identical to that of $\semsingprogyel{\cdot}$
  for all constructors
  other than
  while. 
  The semantics of $\semsingprogyel{\Ewhile{b}{s}}$ is taken to be the least fixed point of the given equation, as is standard.
  % For while, the definition of $\semsingprogyel{\Ewhile{b}{s}}$ is the least fixed-point of the given expression, and $\semsingproggrn{\Ewhile{b}{s}}$ is defined in terms of $\semsingprogyel{\Ewhile{b}{s}}$ \sn{TODO: Clarify fixed points with yellow and green while semantics. Define yel with standard lfp definition, and define grn in terms of yellow by replacing $\emptyset$ with $\divg$.}
}
\label{fig:single_state_individual_semantics}
\end{figure*}

\begin{definition}[Program-Agnostic Semantics of a Program] \label{def:sing_state_sem_prog}
    The \emph{program-agnostic} semantics for single programs $\semsingprogyel{\cdot}$ over the signature
    $\singyel(2^\State \rightarrow 2^\State)$
    is defined as
    $\semsingprogyel{\prog}(X) = \bigcup\limits_{\state \in X} \semsingprogyel{\prog}(\{\state\})$, 
    where the rules for computing 
    $\semsingprogyel{\prog}(\set{\state})$ 
    are defined in Figure~\ref{fig:single_state_individual_semantics}.
    Similarly, the program-agnostic semantics for single programs $\semsingproggrn{\cdot}$ over the signature
    $\singgrn(2^\DState \rightarrow 2^\DState)$
    is defined as
    $\semsingproggrn{\prog}(X) = \bigcup\limits_{\state \in X} \semsingproggrn{\prog}(\{\state\})$ following the rules defined in Figure~\ref{fig:single_state_individual_semantics}.
    
    In Figure~\ref{fig:single_state_individual_semantics}
    and other parts of this paper, 
    for a program $\prog$ and a single state $\state \in \DState$, 
    we write $\semsingprog{\prog}(\state)$ to mean $\semsingprog{\prog}(\{\state\})$ 
    (the same holds for other kinds of semantics defined in this paper).
\end{definition}
% More unusual semantics over programs might include one which counts the number of characters in a program, or one which gives some complexity-theoretic property of a program.

The difference between 
$\semsingprogyel{\cdot}$ and $\semsingproggrn{\cdot}$ 
lies in their treatment 
of divergence: 
$\semsingprogyel{\Ewhile{\Et}{\Eassign{x}{x}}}(\state) = \emptyset$ 
whereas $\semsingproggrn{\Ewhile{\Et}{\Eassign{x}{x}}}(\state) = \{\divg\}$. 
For single programs, 
these two semantics are equivalent in their distinguishing power: 
given two programs $c_1$ and $c_2$, 
$\semsingprogyel{c_1} = \semsingprogyel{c_2}$ iff
$\semsingproggrn{c_1} = \semsingproggrn{c_2}$.
However, the difference in the treatment of divergence 
will create a difference when considering semantics for sets of programs, and we will see that explicitly tracking divergence as in $\semsingproggrn{\cdot}$
is necessary to 
precisely capture how different programs can diverge.

Next, we define the program-aware semantics $\semfullprog{\cdot}$, whose denotations are \emph{sets of} functions, as opposed to functions between sets of states (like $\semsingprogyel{\cdot}$).
%and it explicitly 
%tracks through individual functions how a program maps sets of inputs to sets of outputs.
% 
This distinction is irrelevant for single programs 
(because there is only one function) 
but becomes crucial later for
sets of programs.
% }

\begin{definition} [Program-Aware Semantics ($\semfullprog{\cdot}$)]
\label{def:full-sem-prog}
    The \emph{program-aware semantics} for programs $\semfullprog{\cdot}$ over the signature 
    $\full(2^{2^\DState \rightarrow 2^\DState})$
    % $\full(2^{2^\DState \rightarrow 2^\DState}, 2^{2^\DState \rightarrow 2^\DState}, 2^{2^\DState \rightarrow 2^\DState})$
    is defined so that, for each given program $c$,
    $$\semfullprog{\prog} \defeq \{\semsingproggrn{\prog}\}.$$
\end{definition}

For example, $\semfullprog{\Eassign{x}{5}} = \{\lambda \state. \state\subs{5}{x}\}$ 
because $\semsingproggrn{\Eassign{x}{5}} = \lambda \state. \state\subs{5}{x}$.
While $\semsingproggrn{\cdot}$ and $\semfullprog{\cdot}$ 
contain the same information for single programs, they are different mathematical objects.
When generalizing these
semantics to sets of programs, this difference in representation allows 
$\semfullprog{\cdot}$ to keep the semantics of different programs separate, 
whereas $\semsingproggrn{\cdot}$ cannot do the same.

%are obviously equivalent, they represent information differently because $\semfullprog{\prog}$ maintains $\semsingproggrn{\prog}$ as an independent object. When we define our
%semantics of sets of programs, we will see that generalizing $\semfullprog{\cdot}$ lets us keep the semantics of different programs separate, whereas generalizing $\semsingproggrn{\cdot}$ does not. 

\pagebreak
\subsection{Semantics of Sets of Programs} \label{sec:full_single_semantics}

In this section, we extend Definition~\ref{def:semantics_prog} to define the semantics
of a \emph{set of programs}.

\begin{definition} [Semantics of a Set of Programs]
  \label{def:semantics_set}
  Let $A(\domain_{\Stmtvars}, \domain_{\Expvars}, \domain_{\Boolvars})$ be a 3-tuple with an associated name $A$ where $\domain_{\Stmtvars}$, $\domain_{\Expvars}$, and $\domain_{\Boolvars}$ are the types of the denotations that the semantics assigns to sets of programs
  of type $\Stmtset$, $\Expset$, and $\Boolset$, respectively.
  Then a \emph{semantics over a set of programs}
  $\abstrsnovars{A}$ is defined 
  as the 3-tuple $\abstrsnovars{A} = (\abstrsstmt{A}, \abstrsexp{A}, \abstrsbool{A})$, where:
   $\abstrsstmt{A}: \Stmtset \rightarrow \domain_{\Stmtvars}$, $\abstrsexp{A}: \Expset \rightarrow \domain_{\Expvars}$, and $\abstrsbool{A}: \Boolset \rightarrow \domain_{\Boolvars}$.
   We call $A$ the \emph{signature} of $\abstrsnovars{A}$, and when $\domain_\Stmtvars = \domain_\Expvars = \domain_\Boolvars$, we write the signature as simply $A(\domain)$.\footnote{
    When we consider a semantics over sets of programs (e.g. $\abstrstapp{A}{\cdot}$), we assume that it carries with it $\varfun(\cdot)$. Thus, $\abstrstapp{A}{\cdot}$ really means $\abstrstapp{A}{\cdot} \times \varfun(\cdot)$. 
    This assumption is a reasonable one (because $\varfun$ can be computed immediately from an RTG),
    and the assumption will be needed in later proofs (\S~\ref{Se:modestgoal}). 
} 
We omit the subscripts $\Stmtvars, \Expvars, \Boolvars$ if clear from context.
\end{definition}

Next, we extend the three semantics over programs ($\semsingprogyel{\cdot}$ and $\semsingproggrn{\cdot}$ from \Cref{def:sing_state_sem_prog} and $\semfullprog{\cdot}$ from \Cref{def:full-sem-prog})
to semantics over sets of programs ($\semsingsetyel{\cdot}$, $\semsingsetgrn{\cdot}$, and $\semfullset{\cdot}$, respectively), by taking the union of the semantics of all programs in the sets.

\begin{definition}[Program-Agnostic and Program-Aware Semantics of a Set of Programs] \label{def:sing_state_sem_set}
\label{def:full_sem_set}
   For a given 
   set of programs $\progset$ and an input set of states $X$, each of the semantics $\abstrp{A}$ for $A = \singyel, \singgrn, \full$ can be lifted to a semantics over sets of programs $\abstrs{A}$ over the signature
   $A(2^\State \rightarrow 2^\State)$,
   % $\setsing(2^\DState \rightarrow 2^\DState, 2^\DState \rightarrow 2^\DState, 2^\DState \rightarrow 2^\DState)$
   defined as:
   \[
     \abstrsapp{A}{\progset}(X) = \bigcup_{\prog \in \progset} \abstrpapp{A}{\prog}(X).
   \]
\end{definition}

As discussed in \Cref{sec:intro}, these program-agnostic semantics determine exactly the single-input extensional properties of $\progset$ (and nontermination in the case of $\semsingsetgrn{\cdot}$). 
The \yellow semantics are also noticeably less precise than the
program-aware semantics of sets of programs, $\semfullset{\cdot}$, which explicitly gives the 
semantics of each program $\prog \in \progset$, as illustrated in \Cref{ex:ag_and_aw} below.

\begin{example}
  \label{ex:ag_and_aw}
  Consider two sets of programs, $S_1 = \set{\Eassign{x}{1}, \Eassign{x}{2}}$, and 
  $S_2 = \set{\Eifthenelse{x==0}{\Eassign{x}{1}}{\Eassign{x}{2}}, \Eifthenelse{\neg(x==0)}{\Eassign{x}{2}}{\Eassign{x}{1}}}$.
  Then for an arbitrary input state $\sigma$, 
  $\semsingset{S_1}(\sigma) = \semsingset{S_2}(\sigma) = \set{\sigma\subs{1}{x}, \sigma\subs{2}{x}}$.
  More generally, it can be stated that
  $\semsingset{S_1} = \semsingset{S_2} = \lambda \Sigma. \bigcup_{\sigma \in \Sigma} \set{\sigma\subs{1}{x}, \sigma\subs{2}{x}}$.

  In contrast, $\semfullset{S_1} \neq \semfullset{S_2}$.
  Consider $\semantics{\Eassign{x}{1}} = \lambda \state. \state \subs{1}{x}$.
  Then $\sem{\Eassign{x}{1}} \in \semfullset{S_1}$, but 
  $\sem{\Eassign{x}{1}} \not \in \semfullset{S_2}$: each function in $S_2$ assigns $2$ to $x$ for some $\state$.
  Thus $\semfullset{S_1} \neq \semfullset{S_2}$.
\end{example}

\subsection{Granularity} 
% \subsection{Granularity and Compositionality} 
\label{sec:gran_and_comp}

While the two program-agnostic and
the
program-aware semantics are equi-expressive for single programs, 
the situation changes
for sets of programs.
The following definition of \emph{granularity} 
formalizes the notion of the relative expressivity of a semantics to capture this discrepancy.
% formally captures this discrepancy 
% and formalizes 
% the notion of the expressiveness of a semantics that we have previously used informally.
% \input{Figures/GranularityFigure}

\begin{definition}[Granularity of a Semantics Over Sets of Programs]
  \label{def:granularity}
  Let $A$
  % (\domain_{A1}, \domain_{A2}, \domain_{A3})$
  and $B$
  % (\domain_{B1}, \domain_{B2}, \domain_{B3})$
  be two signatures.
  We say that a semantics for sets of programs 
  $\abstrs{A}$ is \emph{finer} than another semantics $\abstrs{B}$
  (or that $\abstrs{B}$ is \emph{coarser} than $\abstrs{A}$), denoted by $\abstrs{B}\preceq \abstrs{A}$, iff 
  for each type $T$ ($\Stmtvars$, $\Expvars$, or $\Boolvars$) there exists a function $f_T$ such that, for all $C \in \boldsymbol{T}$,
  $f_T(\abstrstapp{A}{C}) = \abstrstapp{B}{C}$.
\end{definition}
When $\abstrs{B}\preceq \abstrs{A}$, we say that $\abstrs{A}$ is finer than $\abstrs{B}$.
When $\abstrs{B}\preceq \abstrs{A}$ but $\abstrs{A}\not\preceq \abstrs{B}$, we say that 
$\abstrs{A}$ is strictly finer than $\abstrs{B}$, denoted by $\abstrs{B}\prec \abstrs{A}$. 
%\noappendix{This granularity preorder induces a complete lattice, 
%which we prove in the extended version of the paper~\cite{semarxiv}.}
This granularity preorder induces a complete lattice, 
which we prove in \Cref{app:lattice}.
%In \Cref{app:lattice}\noappendix{\xspace of the extended version of the paper~\cite{semarxiv}}, we prove this granularity preorder induces a complete lattice.

  Note that $f_{T}$ is essentially $\abstrstapp{\widehat{B}}{(\abstrstapp{A}{\cdot})^{-1}}$, using the convention that, when $\abstrst{B}: T \rightarrow \domain_T$, we say $\abstrst{\widehat{B}} : 2^T \rightarrow 2^{\domain_T}$.
  
\begin{wrapfigure}{l}{0.25\textwidth}
    \vspace{-4mm}
  \begin{tikzcd}
        {T} && {A_{T}} \\
        \\
        {T} && {B_{T}}
        \arrow["id"', from=1-1, to=3-1]
        \arrow["{\abstrst{A}}", from=1-1, to=1-3]
        \arrow["{\abstrst{B}}"', from=3-1, to=3-3]
        \arrow["{f_{T}(\cdot)}", dashed, from=1-3, to=3-3]
    \end{tikzcd}
    \vspace{-4mm}
\end{wrapfigure}
%\begin{minipage}{0.68\linewidth}
  Definition~\ref{def:granularity} can be understood via the 
  diagram on the left: $f_T$ is such that this diagram commutes. 
  This demands that  
  $(\abstrst{A})^{-1}$ refines 
  $(\abstrst{B})^{-1}$; 
  a semantics $\abstrst{A}$
  is finer than $\abstrst{B}$ if the partition induced by 
  $(\abstrst{A})^{-1}$ is finer than the partition induced by 
  $(\abstrst{B})^{-1}$.
%\end{minipage}

Intuitively, for semantics $\abstrs{A}$ and $\abstrs{B}$, $\abstrs{B} \preceq \abstrs{A}$ if
for all $\progset$, $\abstrsapp{A}{\progset}$
lets us uniquely identify $\abstrsapp{B}{\progset}$.
Recall that in this paper, a semantics $\abstrsapp{A}{\cdot}$ is assumed to implicitly
carry $\varfun(\cdot)$, a function that returns the set of variables used by a set of programs. Thus, when considering $\abstrsapp{A}{\progset}$, the function $f_T$ has access to $\varfun(\progset)$ through its argument as well.

From \Cref{def:granularity},
one can show that 
$\semsingsetyel{\cdot} \prec \semsingsetgrn{\cdot}$ and $\semsingsetgrn{\cdot} \prec \semfullset{\cdot}$.
% $\semfullset{\cdot} \succ \semsingsetgrn{\cdot}$ and that
% $\semsingsetgrn{\cdot} \succ \semsingsetyel{\cdot}$. 
% This is true despite the equivalence of $\semfullprog{\cdot}$ and $\semsingprog{\cdot}$ noted earlier.

\begin{theorem} [$\semfullset{\cdot}$ is strictly finer than $\semsingsetgrn{\cdot}$]
% {thm}{singlessthanaware} \label{thm:sing_less_than_aware}
I.e.,
$\semsingsetgrn{\cdot} \prec \semfullset{\cdot}$.
\end{theorem}

\begin{proof}
To see that $\semsingsetgrn{\cdot}\preceq\semfullset{\cdot}$,
simply take 
$f(\semfullset{\progset}) = \bigcup_{\semsingproggrn{\prog} \in \semfullset{\progset}} \semsingproggrn{\prog}$.

To see that $\semfullset{\cdot} \not \preceq \semsingsetgrn{\cdot}$, recall the two sets 
of statements
$S_1$ $= \{ \Eassign{x_1}{2}$, $\Eassign{x_1}{1} \}$ and $S_2$ $= \{ \Eifthenelse{x_1==0}{\Eassign{x_1}{1}}{\Eassign{x_1}{2}}$,  
$\ \Eifthenelse{\neg(x_1==0)}{\Eassign{x_1}{2}}{\Eassign{x_1}{1}} \}$ 
from \Cref{ex:ag_and_aw}.
As illustrated in \Cref{ex:ag_and_aw}, $\semsingsetgrn{S_1} = \semsingsetgrn{S_2}$,
but because 
$\semfullset{S_1} \neq \semfullset{S_2}$,
there cannot exist $f_{\Stmtvars}$
where
$f_{\Stmtvars}(\semsingsetgrn{S}) = \semfullset{S}$ for all $S \in \Stmtset$.
\end{proof}

\vspace{-1.9ex}
\begin{restatable}[$\semsingsetgrn{\cdot}$ is strictly finer than $\semsingsetyel{\cdot}$]{thm}{yellessthangrn}
\label{thm:yel_less_than_grn}
I.e.,
$\semsingsetyel{\cdot}\prec\semsingsetgrn{\cdot}$.
\end{restatable}

\Cref{thm:yel_less_than_grn} can be proven by 
observing that 
$\semsingsetyel{\cdot}$ cannot distinguish between 
an empty set of programs and a set $C_{\mathit{divg}}$ 
that only contains diverging programs, 
because
$\semsingsetyel{\emptyset} = \semsingsetyel{C_{\mathit{divg}}} = \emptyset$, 
but $\semsingsetgrn{\cdot}$ can.
The full proof can be found in 
%\noappendix{the extended version of the paper~\cite{semarxiv}}.
\Cref{app:small-proofs}.
%\noappendix{\xspace in the extended version of the paper~\cite{semarxiv}}.
% , submitted as 
% supplementary material.

\section{Compositionality of the Semantics of Sets of Programs} 
\label{sec:compositional-semantics}

Having established a formal notion of semantics in \Cref{sec:semantics}, we now 
turn to the issue of \emph{compositionality}.
As argued in \S\ref{sec:intro}, the compositionality of a semantics is crucial
for
supporting practical reasoning methods.
% which, as argued in \S\ref{sec:intro}, is a desirable property for a semantics because it provides a starting point for
% practical reasoning methods to be implemented.
% 
In this section, we formally define what it means for a semantics to be 
compositional, and then prove that $\semfullset{\cdot}$ 
\emph{is} compositional, 
whereas $\semsingset{\cdot}$ and $\semsingsetgrn{\cdot}$ 
\emph{are not} compositional.\footnote{
  In \Cref{Se:prelim}, none of $\semsingset{\cdot}$, $\semsingsetgrn{\cdot}$, and $\semfullset{\cdot}$ are stated in a compositional manner.
 \Cref{thm:relational_abstr_inductive} shows that $\semfullset{\cdot}$ can be characterized compositionally.
}

The compositionality of a semantics is 
inspired by
% ---and in fact, almost identical---to the idea of 
\emph{grammar-flow analysis} (GFA)~\cite{SAGA:MW91}.
GFA is a formalism for recursively 
specifying, and computing, aggregate properties of the nonterminals 
in a given regular tree grammar.
A compositional semantics for sets of programs does something very similar: 
it computes properties
over a given inductive set of programs 
(which is, in this paper, defined by a regular tree grammar).
Thus, the
notion that ``a semantics for a set of programs $\abstrs{A}$ is compositional''
can be taken to mean that $\abstrs{A}$ should be computable in a fashion similar to grammar-flow analysis 
(i.e., the semantics of a nonterminal should be computed from the semantics of its productions, 
and the semantics of each production is computed from the semantics of its terms), 
an analogy that may be drawn for many other analyses over sets of 
programs as well (as illustrated in \Cref{sec:case_studies}).
We capture this notion of compositionality by requiring that the semantics $\abstrs{A}$ satisfy the 
following intuitive properties for any input set of programs $C$, defined by a regular tree
grammar $G_C$:
\begin{itemize}
  \item
    For each nonterminal $N \in G_C$, defined by the production right-hand sides
    (RHSs)
    $R_1, \cdots, R_n$, the denotation $\abstrsapp{A}{N}$ 
    is obtained by applying some aggregation function 
    $f^A_{\cup}$ over the denotations of each production RHS: i.e., 
    $\abstrsapp{A}{N} = f^A_{\cup}(\abstrsapp{A}{R_1}, \cdots, \abstrsapp{A}{R_2})$.
    
  \item For a constructor $\opn$ and
  an RHS $R$
  of the form $\opn(\progset_1, \cdots, \progset_n)$, 
  the denotation $\abstrsapp{A}{R}$ is obtained by applying some function 
  $f^A_{\alpha}$ to the denotations of the ``operand'' sets:
  i.e., $\abstrsapp{A}{R} = f^A_{\alpha}(\abstrsapp{A}{\progset_1}, \cdots, \abstrsapp{A}{\progset_n})$.
\end{itemize}
For example, given a set of programs like $\Eifthenelse{x_1 == 2}{x_2:=E}{S}$, a compositional semantics is
able to compute
the semantics of this set using only the semantics of $E$ and $S$, 
without individually considering all programs in $\Eifthenelse{x_1 == 2}{x_2:=E}{S}$.
In contrast,
a characterization like 
$\abstrsapp{A}{\Eifthenelse{x_1 == 2}{x_2:=E}{S}} = 
f(\{\abstrpapp{A}{\Eifthenelse{x_1 == 2}{x_2:=e}{s}} \mid e \in E, s \in S\})$ 
\emph{is noncompositional:}
it enumerates all programs in the set.
\Cref{def:inductivity} formalizes this intuition.

%Now, we would like to know what constitutes a semantics we can use in practice. We are interested in the most coarse-grained semantics that might be practically useful in computing full and/or program-agnostic semantics over inductively-defined sets of programs.
%A ``good'' semantics in this sense should be able to be defined compositionally. 
\begin{definition}[Compositionality]
\label{def:inductivity}
Let $\opset$ denote a finite set of constructors, and $G$ be a regular tree grammar
that only uses constructors in $\opset$.
Let $N$ be a nonterminal in $G$, for which the productions are defined as follows ($N_{1, 1}, \cdots, N_{k, n_k}$ are nonterminals of $G$):
\[
  N ::= \alpha^{(n_1)}_1(N_{1,1}, \ldots, N_{1,n_1}) \mid \ldots \mid \alpha^{(n_k)}_k(N_{k,1}, \ldots, N_{k,n_k})  
\]

We say that a semantics for sets of programs $\abstrs{A}$ 
\emph{admits a compositional characterization over $\opset$}, or simply that
$\abstrs{A}$ is \emph{compositional over $\opset$}, 
if there exists a function $f^A_{\opn_i}$ for each $\opn_i \in \opset$, and an aggregate function 
$f^A_{\cup}$, such that the following holds for any $G$ and $N$:
\[
  \abstrsapp{A}{L(N)} = f^A_{\cup}(f^A_{\opn_1}(\abstrsapp{A}{L(N_{1, 1})}, \ldots, \abstrsapp{A}{L(N_{1, n_1)}}), 
                          \ldots, f^A_{\opn_k}(\abstrsapp{A}{L(N_{k, 1})}, \ldots, \abstrsapp{A}{L(N_{k, n_k})}))
\]
We simply say that $\abstrs{A}$ admits a compositional characterization, or is compositional, 
if it is compositional over the set of all operators in $\gimp$.
We say that $\abstrs{A}$ is compositional over sets of loop-free programs if it is compositional over 
the set of all operators in $\gimp$ except $\mathsf{While}$.
\end{definition}

In the rest of the section, we show that the program-aware semantics $\semfullset{\cdot}$ 
admits a compositional characterization,
whereas the program-agnostic semantics $\semsingsetyel{\cdot}$ and $\semsingsetgrn{\cdot}$ 
admit compositional characterizations only over sets of loop-free programs. 

\begin{theorem} [Compositionality of $\semfullset{\cdot}$] \label{thm:relational_abstr_inductive}
    The program-aware semantics $\semfullset{\cdot}$ admits a compositional characterization.
\end{theorem}

\begin{proof}
    Given $\op$ and $\semfullset{\progset_1}, \cdots, \semfullset{\progset_n}$, we want to determine $\semfullset{\op(\progset_1, \cdots, \progset_n)}$.

    We observe that $\semfullset{\op(\progset_1, \cdots, \progset_n)}$ is exactly $\{\op(\semsingproggrn{\prog_1}, \cdots, \semsingproggrn{\prog_n}) \mid \forall j \leq n \,.\, \semsingproggrn{\prog_j} \in \semfullset{\progset_j}\}$, where applying $\op$ to $\semsingproggrn{\prog_j}$ means to use the standard compositional 
    definition of $\semsingproggrn{\cdot}$ over $\op$,
    shown in \Cref{fig:single_state_individual_semantics} for every constructor besides while.

    For example, for the sequence constructor, $\semfullset{S_1;S_2} = \{\semsingproggrn{s_2} \circ \semsingproggrn{s_1} \mid \semsingproggrn{s_1} \in \semfullset{S_1} \land \semsingproggrn{s_2} \in \semfullset{S_2}\}$.
    For the $\mathsf{while}$ constructor, taking $\mu$ as the usual 
    least fixed-point operator:
    \[\semfullset{\Ewhile{B}{S}} = \left\{
    \mu f. \begin{cases}
        f(\semsingproggrn{s}(\sigma)) & \text{if } \semsingproggrn{b}(\sigma)\\
        \sigma & \text{else}
    \end{cases}
    % \semsingproggrn{\Ewhile{\semsingproggrn{b}}{\semsingproggrn{s}}}
    \,\middle\vert\,
    \semsingproggrn{b} \in \semfullset{B} \land \semsingproggrn{s} \in \semfullset{S}\right\}\]

    More intuitively (and abusing notation), we can say $\semfullset{\Ewhile{B}{S}} = \{\Ewhile{\semsingproggrn{b}}{\semsingproggrn{s}} \mid \semsingproggrn{b} \in \semfullset{B} \land \semsingproggrn{s} \in \semfullset{S}\}$, i.e., we consider all the functions that can be built as loops whose guard computes a function in $\semfullset{B}$ and whose body computes a function in $\semfullset{S}$.
    % Here ``$\mu f. \Eifthenelse{\semsingproggrn{b}(\sigma)}{f(\semsingproggrn{s}(\sigma))}{\sigma}$'' is the usual semantics of a while loop whose guard computes $\semsingproggrn{b}$ and whose body computes $\semsingproggrn{s}$. 
     For unions, if $\progset = \progset_1, \cdots, \progset_m$, then $\semfullset{\progset} = \bigcup\limits_{j \leq m}\semfullset{\progset_j}$.

     Consequently,
     $\semfullset{\cdot}$ is compositional.
\end{proof}

\begin{figure*}
{\small
\[
\begin{array}{c}
  \semsingset{0}(\state)=\{\state\subs{0}{e_t}\} \qquad
  \semsingset{1}(\state)=\{\state\subs{1}{e_t}\} \qquad
  \semsingset{x_j}(\state)=\{\state\subs{x_j}{e_t}\}
  \\[2.5mm]
  \semsingset{\Et}(\state)=\{\state\subs{\Et}{b_t}\} \qquad
  \semsingset{\Ef}(\state)=\{\state\subs{\Ef}{b_t}\} \\[2.5mm]
  \semsingset{\neg B}(\state)=\{\sigma\subs{\neg \sigma_b[b_t]}{b_t} \mid \sigma_b \in \semsingset{B}(\state)\} \qquad \semsingset{\bigcup\limits_{j < n} \progset_j}(\state)=\bigcup\limits_{j < n} \semsingset{\progset_j}(\state) 
  \\[2.5mm]
  \semsingset{A_1 \odot A_2}(\state)=\{\state\subs{\state_1[a_t] \odot \state_2[a_t]}{r_t} \mid \state_1 \in \semsingset{A_1}(\state), \state_2 \in \semsingset{A_2}(\state)\}
  \\[2.5mm]
  \semsingset{\Eassign{x}{E}}(\state)=\{\sigma\subs{\sigma_e[e_t]}{x} \mid \sigma_e \in \semsingset{E}(\state)\} \qquad
  \semsingset{S_1;S_2}(\state) = \semsingset{S_2}(\semsingset{S_1}(\state)) \quad 
  \\[2.5mm]
  \semsingset{\Eifthenelse{B}{S_1}{S_2}}(\state) = \begin{cases}
                                                      \semsingset{S_1}(\state) & \set{\state_b[b_t] \mid \state_b \in \semsingset{B}(\state)} = \set{\Et} \\
                                                      \semsingset{S_2}(\state) & \set{\state_b[b_t] \mid \state_b \in \semsingset{B}(\state)} = \set{\Ef} \\
                                                      \semsingset{S_1}(\state) \cup \semsingset{S_2}(\state) & \text{otherwise}
                                                   \end{cases}
\end{array}
\]
}
\caption{
  Compositional {definitions} for the program-agnostic semantics $\semsingset{\cdot}$ (\Cref{def:sing_state_sem_set}).
  In the fourth line, $\odot$ denotes a binary operator (either $+, -, <, ==,$ or $\wedge$), and 
  $a_t$ and $r_t$ are auxiliary variables of the appropriate type. Note that $\semsingset{\progset}(X) = \bigcup\limits_{\state \in X}\semsingset{\progset}(\state)$.
}
\label{fig:single_state_semantics}
\end{figure*}

\noindent 
Next, we show that $\semsingsetyel{\cdot}$ and $\semsingsetgrn{\cdot}$ 
are compositional for loop-free programs.
\begin{theorem} [Loop-Free Compositionality of $\semsingset{\cdot}$ and $\semsingsetgrn{\cdot}$] 
\label{thm:ag_loop_free_comp}
     The program-agnostic semantics 
     $\semsingsetyel{\cdot}$ and $\semsingsetgrn{\cdot}$ 
     admit compositional characterizations over loop-free programs.
\end{theorem}
\begin{proof}
    Figure~\ref{fig:single_state_semantics} provides a compositional characterization of both semantics; 
    equivalence to Definition~\ref{def:sing_state_sem_set}
    can be proven by structural induction.
\end{proof}
  However, $\semsingset{\cdot}$ and 
  $\semsingsetgrn{\cdot}$ are no longer compositional when programs are allowed 
  to have loops.
\vspace{-5mm}
\begin{theorem} [Noncompositionality of $\semsingset{\cdot}$ and 
$\semsingsetgrn{\cdot}$] \label{thm:sing_noninductive}
    Both $\semsingset{\cdot}$ and $\semsingsetgrn{\cdot}$
    do not admit compositional characterizations when loops are permitted.
\end{theorem}
\begin{proof}
    We prove that $\semsingset{\cdot}$ is 
    not compositional for loops; the proof is identical 
    for $\semsingsetgrn{\cdot}$.
    
    We will show that there are two sets of loop guards, $B_1$ and $B_2$, 
    with the \emph{same} program-agnostic semantics but which can be used to construct 
    otherwise identical sets of loops with \emph{different} program-agnostic semantics.
    Let 
    \begin{align*}
        B_1 = \{x == 1, \neg (x == 1)\} &\hspace{4mm} B_2 = \{x == 2, \neg (x == 2)\}.
    \end{align*}
    Then observe for any input state $\sigma \in \State$ and $B= B_1$ or $B_2$, 
    there is a guard $\mathsf{guard}_{\Etrue} \in B$ such that 
    $\semsingprog{\mathsf{guard}_{\Etrue}}(\state)[b_t] = \Et$
    and a guard $\mathsf{guard}_{\Efalse} \in B$ such that 
    $\semsingprog{\mathsf{guard}_{\Efalse}}(\state)[b_t] = \Ef$.
    Thus, on every set of input states $\Sigma$,
    $\semsingset{B_1}(\Sigma) = \semsingset{B_2}(\Sigma) = 
    \Sigma\subs{\Et}{b_t} \cup \Sigma\subs{\Ef}{b_t}$.
    
    Now consider two sets of loops, $W_1$ and $W_2$, 
    that share the same body:
    \begin{align*}
        W_1 = \Ewhile{B_1}{x := x + 1} &\hspace{4mm} W_2 = \Ewhile{B_2}{x := x + 1}.
    \end{align*}
    Because $\semsingset{B_1} = \semsingset{B_2}$, if $\semsingset{\cdot}$
    is
    compositional, it must be the case that $\semsingset{W_1} = \semsingset{W_2}$.
    However, a simple check proves they are not:
    \begin{align*}
        \semsingset{W_1}&(\{\sigma \in \State \mid \sigma[x] = 0\})\\ &= \{\sigma \in \State \mid \sigma[x] = 0\} \cup \{\sigma \in \State \mid \sigma[x] = 1\}\\
        \semsingset{W_2}&(\{\sigma \in \State \mid \sigma[x] = 0\})\\ &= \{\sigma \in \State \mid \sigma[x] = 0\} \cup \{\sigma \in \State \mid \sigma[x] = 2\}
    \end{align*}

    Because $\semsingset{W_1} \neq \semsingset{W_2}$, we have that $\semsingset{\cdot}$ and $\semsingsetgrn{\cdot}$
    are not compositional over sets of programs containing loops.
\end{proof}

Intuitively, 
\Cref{thm:sing_noninductive}
means that $\semsingset{\cdot}$ cannot always be
determined compositionally when sets of loops are allowed. 
Thus, compositional verification techniques with semantics equivalent to
$\semsingset{\cdot}$ (see \Cref{subsec:setlogic,subsec:interpreter}) will be
unable to prove many properties about such sets of programs. 
For example, no such technique can prove any property 
that distinguishes
$W_1$ from $W_2$, as defined above. 
More generally, properties of sets of loops
that
involve multiple iterations of a loop
%require executing a guard/body more than once to check
are in general unprovable 
via compositional techniques over $\semsingset{\cdot}$
(e.g., checking that some loop in $W_2$ can map a
$\state$ 
in which
$x \mapsto 0$ to $\state\subs{2}{x}$ requires executing the loop 
$\Ewhile{\neg(x==2)}{x := x + 1}$ for two iterations). 
Put another way, $\semsingset{\cdot}$ cannot ``express multiple executions'' 
as explained in \Cref{sec:intro}, 
so verification techniques built on 
$\semsingset{\cdot}$ cannot check many properties 
that require repeated execution of the same code.

\section{Vector-State Semantics}
\label{Se:modestgoal}

Having proved that 
%the program-agnostic semantics 
$\semsingsetyel{\cdot}$ and $\semsingsetgrn{\cdot}$ 
do not admit compositional characterizations, 
the natural question to ask is:
What is the coarsest \emph{compositional} semantics that is at least as fine as $\semsingsetyel{\cdot}$
(or $\semsingsetgrn{\cdot}$)?
In this section, we answer this question by proposing the 
\emph{vector-state semantics} $\semweaksetyel{\cdot}$ 
(and $\semweaksetgrn{\cdot}$), the coarsest compositional semantics 
that is at least as fine-grained as $\semsingsetyel{\cdot}$ 
(or as $\semsingsetgrn{\cdot}$).
\subsection{Intuition for Vector-States} \label{sec:weak_vs_intuition}

The fundamental challenge of describing the behavior of sets of programs is ``synchronizing'' programs with inputs 
to express multiple executions, as described in \Cref{sec:intro}.
To enable this synchronization, we lift our semantics from sets of individual input states to sets of sequences of input states. 
This approach allows us to apply the same program to each state in an input sequence, 
so that the corresponding output sequence is the result of executing a single program $\prog \in \progset$.

The noncompositionality of the program-agnostic semantics 
$\semsingsetyel{\cdot}$ and $\semsingsetgrn{\cdot}$ can be attributed to their inability to express multiple executions. 
In particular, for a set of guards $B$ and set of statements $S$, each loop 
$(\Ewhile{b}{s}) \in (\Ewhile{B}{S})$ may execute $b$ and $s$ arbitrarily many times. 
Consequently, to answer questions about the program-agnostic semantics of $\Ewhile{B}{S}$, we need to be able to answer questions about multiple executions
of specific pairs $(b, s) \in B \times S$.
For example, to know that
$p_n =~ \statelist{\stateentry{x}{10}}$ is a possible final state of 
$\semsingsetyel{\Ewhile{B}{S}}(\{\statelist{\stateentry{x}{0}}\})$, 
we would need to know that there is a guard $b \in B$, a loop body $s \in S$, and a sequence of states $[p_1, \cdots, p_n]$ 
for which
$b$ is true on all $p_j$ except $p_n$, and $s$ maps each $p_j$ to $p_{j+1}$.

More generally, given a set of programs $\progset$ and sequences $[p_1,\cdots, p_n]$ and $[q_1,\cdots, q_n]$ of inputs and outputs, 
one must be able to ask whether there
exists a \emph{single} program $\prog \in \progset$ such that $\semsingprog{\prog}(p_j) = q_j$ for all $j$. 
Questions of this sort cannot be answered from $\semsingsetyel{\progset}$, 
yet they are necessary to understand the program-agnostic behavior of sets of loops defined in terms of $\progset$.

The goal of this section is to introduce
two new, \emph{vector-state semantics}, denoted by $\semweaksetyel{\cdot}$ and $\semweaksetgrn{\cdot}$, which
capture exactly this kind of information about $\progset$
(with and without the ability to capture nontermination, respectively).
These vector-state semantics admit compositional characterizations over 
\textsf{while}, and they are at least as fine as $\semsingsetyel{\cdot}$ and $\semsingsetgrn{\cdot}$, respectively. 
In addition, the vector-state semantics are the \emph{coarsest} such semantics, 
and thus they represent the minimum amount of information about sets of programs one must track to enable compositional reasoning.

\subsection{\texorpdfstring{$\semweaksetyel{\cdot}$}{}: \yellowCaps Vector-State Semantics}
\label{sec:weak-vector-sem}
We begin by developing 
$\semweaksetyel{\cdot}$---a \yellow vector-state semantics.
  % $\semweaksetyel{\cdot}$ is significantly easier to develop compared to $\semweaksetgrn{\cdot}$, the divergence-aware 
  % vector-state semantics,
  % because $\semweaksetyel{\cdot}$ can ignore divergence, and thus 
  % can ignore
  % diverging loop bodies and infinite loop executions.
We first provide
some intuition about how $\semweakset{\cdot}$ should operate on 
input vector-states, and formalize the definition of $\semweaksetyel{\cdot}$.
Then, we prove that
\rone $\semweaksetyel{\cdot}$ has a compositional characterization, and 
\rtwo $\semweaksetyel{\cdot}$ is
the coarsest-grained compositional semantics that is 
finer than $\semsingset{\cdot}$.
As discussed,
we will use 
vector-states to capture the behavior of a set of programs 
(e.g., a set of loop bodies) under multiple executions.

Following this intuition, $\semweaksetyel{\progset}$ 
should let each $\prog \in \progset$ operate over 
each entry in the vector
to generate an output vector-state.
For example, given the set 
$S = \set{\Eassign{x}{x + 2}, \Eassign{x}{x + 10}}$, 
the vector-state semantics
$\semweaksetyel{S}([\statelist{\stateentry{x}{2}}, \statelist{\stateentry{x}{4}}])$ yields 
$\{[\statelist{\stateentry{x}{4}}, \statelist{\stateentry{x}{6}}], 
   [\statelist{\stateentry{x}{12}}, \statelist{\stateentry{x}{14}}]
\}$.

%Because $\semsingsetyel{\cdot}$ ignores divergence, 
We define $\semweaksetyel{\cdot}$ so that, when a program diverges on 
\emph{any} state in a vector-state, 
the \emph{entire} output vector is discarded. This is
% to simply return $\emptyset$ on diverging combinations of inputs and outputs,
similar to how Hoare logic returns $\Efalse$ 
as the postcondition of a diverging loop. For example, for the singleton set 
$S = \set{\Ewhile{x < 2}{\Eassign{x}{x - 1}}}$, the vector-state semantics
$\semweaksetyel{S}([\statelist{\stateentry{x}{2}}, \statelist{\stateentry{x}{4}}])$ yields $\{[\statelist{\stateentry{x}{2}}, 
\statelist{\stateentry{x}{4}}]\}$ 
because the loop does not diverge on any input in the vector, 
but $\semweaksetyel{S}([\statelist{\stateentry{x}{2}}, \statelist{\stateentry{x}{4}}, \statelist{\stateentry{x}{1}}])$ 
yields $\emptyset$: 
the \emph{entire output vector-state} is eliminated, due to the existence of a diverging 
element in the input vector-state.

\subsubsection{Formalization of $\semweaksetyel{\cdot}$}
In the following, $\State^*$ denotes the set of finite sequences over $\State$.

\begin{definition} [\yellowCaps Vector-State Semantics of a Program] \label{def:weak_sem_prog_yel}
    The \emph{vector-state semantics} for programs $\semweakprogyel{\cdot}$ 
    over the signature $\weak(2^{\State^*} \rightarrow 2^{\State^*})$ is defined as follows.
    We first define the semantics for singleton vector states.
    For a program $\prog \in T$ and finite vector state $[v_1,\cdots,v_n] \in \State^*$, the semantics 
    $\semweakprogyel{\prog}(\{[v_1,\cdots,v_n]\}))$ is defined as follows:
    \begin{enumerate}
        \item If for every $v_i$ the program $\prog$ does not diverge---i.e., $\semsingprogyel{\prog}(v_i)\neq \emptyset$---then 
        \[\semweakprogyel{\prog}(\{[v_1,\ldots,v_n]\}) = \{[\semsingprogyel{\prog}(v_1), \cdots, \semsingprogyel{\prog}(v_n)]\}.\]
        \item If for some $v_i$ the program $\prog$ does diverge---i.e., $\semsingprogyel{\prog}(v_i) = \emptyset$---then
        \[\semweakprogyel{\prog}(\{v\}) = \emptyset.\]
    \end{enumerate}

    \noindent
    Note that entries of the output vectors are technically singleton sets; we unpack them as states.

    \noindent
    Then for a set of vector states $V \subseteq \State^*$, we define 
    $\semweakprogyel{\prog}(V) = \bigcup\limits_{u \in V}\semweakprogyel{\prog}(\{u\})$.
\end{definition}

Intuitively, $\semweakprogyel{\prog}(\{v\})$ (often abbreviated as $\semweakprogyel{\prog}(v)$) 
gives the output of $\prog$ on $v$ for each entry of $v$ unless $\prog$ diverges on some state in $v$. 
If $\prog$ diverges along $v$, then no vector-state is returned. 
Thus, $\semweakprogyel{\prog}(V)$ will show no trace of vectors $v \in V$ 
for which
$\semweakprogyel{\prog}(v) = \emptyset$. 
% as long as $\semweakprogyel{\prog}(v') \neq \emptyset$ for some $v' \in V$. 
Generalizing $\semweakprogyel{\cdot}$ to sets of programs is done by unioning the semantics of all programs in the set,
similar to \Cref{def:sing_state_sem_set}.
% how $\semsingset{\cdot}$ generalizes $\semsingprog{\cdot}$ to sets of programs.

\begin{definition} [\yellowCaps Vector-State Semantics of a Set of Programs] \label{def:weak_sem_set_yel}
    The \textit{vector-state semantics} for sets of programs, denoted $\semweaksetyel{\cdot}$,
    over the signature
    $\setweak(2^{\State^*} \rightarrow 2^{\State^*})$
    is defined 
    for an input set $V \subseteq \State^*$ as follows:
    \[\semweaksetyel{\progset}(V) = \bigcup\limits_{\prog \in \progset} \semweakprogyel{\prog}(V).\] 
\end{definition}

\begin{example} %[Understanding $\semweaksetyel{\cdot}$]
\label{ex:vector-states}
    Consider $S = \{\Eassign{x}{x + n} \mid n \in \N\}$.
    Given $v = [\statelist{\stateentry{x}{1}}, \statelist{\stateentry{x}{2}}, \statelist{\stateentry{x}{3}}]$, we get 
    $\semweaksetyel{S}(v) = \{[\statelist{\stateentry{x}{1+n}}, \statelist{\stateentry{x}{2+n}}, \statelist{\stateentry{x}{3+n}}] \mid n \in \N\}$.
\end{example}

\subsubsection{Properties of $\semweaksetyel{\cdot}$}
%We now show that 
$\semweaksetyel{\cdot}$ admits a compositional characterization
(Theorem~\ref{thm:weak_compositional_yel}), and
%and that $\semweaksetyel{\cdot}$ 
is the coarsest compositional semantics 
that is finer-grained than $\semsingsetyel{\cdot}$ (Theorem~\ref{thm:weak_minimality_yel}).

\begin{restatable} [Compositionality of $\semweaksetyel{\cdot}$]{thm}{yelwkcomp} \label{thm:weak_compositional_yel}
    $\semweaksetyel{\cdot}$ admits a compositional characterization.
\end{restatable}

\begin{figure*}
\begin{center}
{ \small
\[
\begin{array}{c}
  \semweakset{1}(V)= \Red(\bigcup_{v \in V} \set{v\subs{1}{e_t}}) \quad
  \\[2.5mm]
  \semweakset{x_j}(V)= \Red(\bigcup_{v \in V} \set{v\subs{v[x_j]}{e_t}})\quad
  \semweakset{\Et}(V)= \Red(\bigcup_{v \in V} \set{v\subs{\Et}{b_t}}) 
  \\[2.5mm]
  \semweakset{\Ef}(V)= \Red(\bigcup_{v \in V} \set{v\subs{\Ef}{b_t}}) \quad
  \semweakset{\neg B}(V)=
    \set{v^{b}\subs{\neg v^{b}[b_t]}{b_t} \mid v^{b} \in \semweakset{B}(V)}
  \\[2.5mm]
  \semweakset{A_1 \odot A_2}(V)=
      \set{v^{2}\subs{v^{1}[q_t] \odot v^{2}[q_t]}{r_t} \mid v^{1} \in \semweakset{A_1}(V), v^{2} \in \semweakset{A_2}(V)}
  \\[2.5mm]
  \semweakset{\Eassign{x}{E}}(V)=\Red(\bigcup_{v \in V} \set{v^{e}\subs{v^{e}[e_t]}{x} \mid v^{e} \in \semweakset{E}(v)}) 
  \\[2.5mm]
  \semweakset{S_1;S_2}(V) = \semweakset{S_2}(\semweakset{S_1}(V)) \quad 
  \\[2.5mm]
  \semweakset{\Eifthenelse{B}{S_1}{S_2}}(V)= \\\left\{\mathit{interleave}(v^{s_1},v^{s_2},v^{b}) 
  \,\middle\vert\,
  v^{b} \in \semweakset{B}(V), v^{s_1} \in \semweakset{S_1}(\filter{V}{v^{b}}), v^{s_2} \in \semweakset{S_2}(\filter{V}{\neg v^{b}})\right\}
  \\[2.5mm]
  \semweakset{\Ewhile{B}{S}}(V) = f^{\setweak}_{while}(\semweakset{B}, \semweakset{S}, V)
  \\[2.5mm]
  \semweakset{\bigcup_{j < n} S_j}(V) = \Red(\bigcup_{j < n}\semweakset{S_j}(V))
\end{array}
\]
}
\end{center}
\vspace{-1mm}
\caption{
  A compositional characterization
  for the \yellow vector-state semantics $\semweaksetyel{\cdot}$ defined in
  \Cref{def:weak_sem_set_yel}.
  % (and the \green vector-state semantics 
  % $\semweaksetgrn{\cdot}$ introduced in \Cref{sec:grn_vs}).
  % Because the characterizations of $\semweaksetyel{\cdot}$ and 
  % $\semweaksetgrn{\cdot}$ have nearly the same form, we write them only for 
  % $\semweaksetyel{\cdot}$ above. 
  The only differences between the characterizations of 
  $\semweaksetyel{\cdot}$ and $\semweaksetgrn{\cdot}$ (the \green vector-state semantics 
   introduced in \Cref{sec:grn_vs}) are that \rone for 
  $\semweaksetyel{\cdot}$, one does not apply $\Red$ 
  %or $\rdset{\cdots}$ (defined below) 
  and \rtwo the functions $f_{\mathit{while}}$ are different.
  %We define a shorthand operator that reduces over sets of vector-states as
  %$\rdset{v \in V}(\phi(v))=\Red(\bigcup_{v \in V} \phi(v))$.
  The symbol $\odot$ denotes a binary operator in $\gimp$, and 
  $q_t$ and $r_t$ refer to auxiliary variables of the appropriate type. 
  Substitutions $v\subs{a}{b}$ to vectors are applied elementwise. 
  %Thus, when $a$ is a value, $v\subs{a}{b}$ sets $b$ to $a$ in every $v[i]$. When $a$ is itself a vector, we do $v[i]\subs{a[i]}{b}$ on each vector entry. Each
  The functions $f_{\mathit{while}}$, $\mathit{interleave}$, and $\mathit{filter}$ are
  defined in the proofs of Theorems~\ref{thm:weak_compositional_yel} and \ref{thm:weak_inductive}\noappendix{, available in the extended version of the paper~\cite{semarxiv}}.
  Briefly, $\mathit{interleave}(v^{s_1}, v^{s_2}, v^b)$ produces an array using entries of $v^{s_1}$ where $v^b$ is true and $v^{s_2}$ where $v^b$ is false, and $\mathit{filter}(V, v^b)$ deletes entries from vectors in $V$ where $v^b$ is false. On singleton vectors, the function $f_{\mathit{while}}$ produces $\semweakset{\Ewhile{B}{S}}([\state])$ by finding all realizable traces of loop executions on $\state$ according to $\semweakset{B}$ and $\semweakset{S}$.
  %is defined in the proofs of Theorems~\ref{thm:weak_compositional_yel} and \ref{thm:weak_inductive}. The operations $\mathit{interleave}$ and $\mathit{filter}$ are defined in the proofs of Theorems~\ref{thm:weak_compositional_yel} and \ref{thm:weak_inductive} as well, although here we extend $\mathit{filter}$ as $\filter{V}{b} = \set{\filter{v}{b} \mid v \in V}$.
  % The operation $\mathit{interleave}$ produces an array using entries of $s_1$ where $b$ is true and $s_2$ where $b$ is false. The operation $\mathit{filter}$ deletes entries from vectors in $V$ where $b$ is false. Both definitions appear in the same proof ofThm~\ref{thm:weak_inductive}, although here we extend $\mathit{filter}$ as $\filter{V}{b} = \set{\filter{v}{b} \mid v \in V}$.
  % For $\semweakset{\Ewhile{B}{S}}(v)$ we do as follows: For each distinct $\semsingprog{b}$ (which we can determine from $\semsingset{B}$), determine the sequences of queries to $S$ that would cause the loop with guard $\semsingprog{b}$ to converge and diverge. $\semweakset{S}$ on each sequence gives the loop behavior.
  %   \item $\semweakset{\bigcup_{j < n} S_j}(v) = Red(\bigcup_{j < n}\semweakset{S_j}(v))$
}
\label{fig:weak_compositional}
\vspace{-4mm}
\end{figure*}
\Cref{fig:weak_compositional} gives compositional characterizations 
of $\semweakset{\cdot}$ over unions and all constructors except
\textsf{while}.
%\noappendix{The full proof is provided in the extended version of the paper~\cite{semarxiv}; we give a sketch below.}
The full proof is given in \Cref{app:weak_compositionality_yel}.
%\noappendix{\xspace in the extended version of the paper~\cite{semarxiv}}. 
We give a sketch below.

\begin{proofsketch}
    We sketch only the most complex cases---\textsf{if-then-else} 
    and \textsf{while}.
    % \removable{will} focus our discussion on 
    % \textsf{if-then-else} and \textsf{while}, \snnchanged{which are the most complex cases.}
    % because they are the most complex operators when considering 
    % vector-state semantics,
    % and
    We restrict our explanation to singleton sets 
    of inputs for simplicity of presentation.

    For $\Eifthenelse{B}{S_1}{S_2}$, 
    we can determine the ways that guards in $B$ can split our inputs by evaluating 
    $\semweaksetyel{B}(v) = \{v^b \mid b \in B\}$. For each such $v^b$, we pass the part of $v$ where $v^b$ is true to $\semweakset{S_1}$ and the false part to $\semweakset{S_2}$. We then reassemble the results according to $v^b$.

    For $\Ewhile{B}{S}$ and vectors $[\state]$, we employ a ``generate and test'' strategy.
    The execution of a loop can be written as the sequence of states reached after
    every iteration. 
    To find all possible loop executions, we consider each finite vector $t$ of states 
    beginning with $\state$ and ask \rone whether any guard can send all
    $t$'s states to true except the final state based on $\semweaksetyel{B}$, 
    and \rtwo whether any body maps each state in $t$ to its successor based on $\semweaksetyel{S}$.
    If both checks pass, the vector $t$ represents a possible loop execution and its final 
    state is a possible output 
    of
    $\Ewhile{B}{S}$ on $\state$. 
    Generalizing to input vectors of length $k$ requires considering concatenations of $k$ 
    guesses as a single vector.
    
\end{proofsketch}

The \yellow vector-state semantics $\semweaksetyel{\cdot}$
is the \emph{coarsest} compositional semantics that captures the 
program-agnostic semantics $\semsingset{\cdot}$.
In other words, $\semweaksetyel{\cdot}$ sets a lower bound on the granularity 
required for a compositional semantics over sets of programs that captures $\semsingsetyel{\cdot}$ -- i.e., 
a lower bound on the complexity, of reasoning about sets of programs 
in a compositional fashion.

\begin{restatable}
    [$\semweaksetyel{\cdot}$ is Coarsest Compositional for $\semsingsetyel{\cdot}$]{thm}{yelwkmin} \label{thm:weak_minimality_yel}
    For every compositional semantics 
    $\abstrs{A}$ such that $\semsingsetyel{\cdot}\preceq \abstrs{A}$, we have that $\semsingsetyel{\cdot} \preceq \semweaksetyel{\cdot}\preceq \abstrs{A}$.
\end{restatable}

\begin{proofsketch}
    Clearly,
    $\semsingsetyel{\cdot} \preceq \semweaksetyel{\cdot}$ as
    $\semweaksetyel{\progset}$ on vectors of length $1$ 
    simulates $\semsingsetyel{\progset}$.
    % Clearly, $\semweaksetyel{\progset}$ determines $\semsingsetyel{\progset}$ because we can evaluate $\semweaksetyel{\progset}$ on vectors of length $1$ to simulate $\semsingsetyel{\progset}$. 
    % So $\semsingsetyel{\cdot} \preceq \semweaksetyel{\cdot}$.

    We must then show that $\semweaksetyel{\cdot}$ is the coarsest compositional semantics at least as fine as $\semsingsetyel{\cdot}$. To do so, we observe that
    any compositional semantics $\abstrs{A}$ such that $\semsingsetyel{\cdot} \preceq \abstrs{A}$ 
    must determine sets of statements $S$ well enough to reason about the program-agnostic semantics of sets of loops in which $S$ appears in the body. Given a set of statements $S$, for each pair $(v,u)$ of input and output vector-states, we construct sets of loops $W_{v,u}$ that loop over the indices of our vectors and check the behavior of $S$ on each entry (i.e., $\Ewhile{b}{s_1;S;s_2}$, where $s_1$ sets the input state to $v[i]$, $s_2$ checks if the output state is $u[i]$ and breaks if not, and $b$ iterates through $v$). The proof concludes by showing the program-agnostic semantics of such sets of loops determines whether $u \in \semweaksetyel{S}(v)$. 
    Thus, $\semweaksetyel{\cdot} \preceq \abstrs{A}$ on sets of statements. A similar argument carries over to expressions.
\end{proofsketch}
    The full proof of \Cref{thm:weak_minimality_yel} is given
    %\noappendix{in the extended version of the paper~\cite{semarxiv}}.
    \Cref{app:weak_coarseness_yel}.
    %\noappendix{, in the extended version of the paper~\cite{semarxiv}}.

\vspace{-2mm}
\paragraph{Multiple-Input Extensional Properties}
The vector-state semantics $\semweaksetyel{\cdot}$ extends the ability of 
$\semsingsetyel{\cdot}$ to cover $k$-input extensional properties for unbounded $k$, 
which is useful for, e.g., 
capturing programming-by-example (PBE) program-synthesis problems~\cite{sketch,flashfill}.
Specifications in PBE are given as a set of input-output examples 
$\set{(\mathit{in}_1, \mathit{out}_1), \cdots, 
(\mathit{in}_k, \mathit{out}_k)}$; 
whether or not 
such a PBE specification 
has a solution in a set of programs $C$
can be checked by checking whether $[\mathit{out}_1, \cdots, \mathit{out}_k] \in 
\semweaksetyel{C}([\mathit{in}_1, \cdots, \mathit{in}_n])$.

\subsection{\texorpdfstring{$\semweaksetgrn{\cdot}$}{}: Divergence-Aware Vector-State Semantics}\label{sec:grn_vs}

We now introduce $\semweaksetgrn{\cdot}$, 
the \green vector-state semantics, 
which is the coarsest compositional semantics at least as fine as the 
\green program-agnostic semantics $\semsingsetgrn{\cdot}$.
Whereas $\semweaksetyel{\cdot}$ was a straightforward vectorization of $\semsingsetyel{\cdot}$ (\Cref{def:weak_sem_prog_yel}), nontermination requires more careful treatment:
na\"ively extending $\semsingsetgrn{\cdot}$ to vector-states will 
result in a semantics that is more fine-grained than necessary.
In this section, we provide 
intuition about
why nontermination imposes significant 
challenges on the construction of $\semweaksetgrn{\cdot}$.
%\noappendix{The extended version of the paper~\cite{semarxiv} 
%contains a formal treatment of $\semweaksetgrn{\cdot}$.}
\Cref{app:vec_grn} formalizes $\semweaksetgrn{\cdot}$.
%\noappendix{\xspace in the extended version of the paper~\cite{semarxiv}}.

\subsubsection{The Na\"ive Approach: \texorpdfstring{$\semweaksetgrnbad{\cdot}$}{}}\label{sec:vs_grn_naive}
The natural approach to handle nontermination is to extend $\semweaksetyel{\cdot}$ to operate over infinite vectors and record divergence on each state in a vector as follows:

\begin{example}[$\semweaksetgrnbad{\cdot}$ Semantics: Too Fine] \label{ex:bad_vs_grn}
    Consider a semantics for single programs $\semweakproggrnbad{\cdot}$ that extends 
    $\semsingproggrn{\cdot}$ towards vector-states, such that 
    $\semweakproggrnbad{\prog}(\set{[v_1, \cdots, v_n]}) = \set{[\semsingproggrn{\prog}(v_1), \cdots, \semsingproggrn{\prog}(v_n)]}$, 
    for a program $\prog$ and an input vector-state $[v_1, \cdots, v_n]$.
    Define $\semweaksetgrnbad{\progset}(V) = \bigcup_{\prog \in \progset} \semweakproggrnbad{\prog}(V)$ for a set of programs $\progset$ and 
    a set of input vector-states $V$.
    Then for a set of programs $S = \{w_1, w_2\}$ where 
    $w_1 = \Ewhile{x = 1}{\Eassign{x}{x}}$ and $w_2 = \Ewhile{x = 2}{\Eassign{x}{x}}$,
    
    {
    % \allowdisplaybreaks
     \small
    \begin{align*}
        %\semweaksetgrnbad{S}(\{[\statelist{\stateentry{x}{1}}]\}) &= \%{[\sem{w_1}\statelist{\stateentry{x}{1}}], [\sem{w_2}\statelist{\stateentry{x}{1}}]\} \\ 
        %&= \{[\divg], [\statelist{\stateentry{x}{1}}]\}\\  
        \semweaksetgrnbad{S}(
          \{[\statelist{\stateentry{x}{1}}, \statelist{\stateentry{x}{2}}]\}) &= 
          \{[\sem{w_1}\statelist{\stateentry{x}{1}}, 
             \sem{w_1}\statelist{\stateentry{x}{2}}], 
            [\sem{w_2}\statelist{\stateentry{x}{1}}, 
             \sem{w_2}\statelist{\stateentry{x}{2}}]\} \\ 
        &= \{[\divg, \statelist{\stateentry{x}{2}}], 
             [\statelist{\stateentry{x}{1}}, \divg]\}\\
        \semweaksetgrnbad{S}(
          \{[\statelist{\stateentry{x}{2}}, 
             \statelist{\stateentry{x}{3}}, \cdots]\}) &= 
          \{[\sem{w_1}\statelist{\stateentry{x}{2}}, 
             \sem{w_1}\statelist{\stateentry{x}{3}}], 
            [\sem{w_2}\statelist{\stateentry{x}{2}}, 
             \sem{w_2}\statelist{\stateentry{x}{3}}]\} \\ 
        &= \{[\statelist{\stateentry{x}{2}}, \statelist{\stateentry{x}{3}}, \cdots], 
             [\divg, \statelist{\stateentry{x}{3}}, \cdots]\}
    \end{align*}}
    Observe how $\divg$ may appear both before, after, and in between other entries.
\end{example}
Although $\semweaksetgrnbad{\cdot}$ is compositional and finer than $\semsingsetgrn{\cdot}$, 
it is not the \emph{coarsest} such semantics.
%Applications of $\semweaksetgrnbad{\progset}$ provide more information than is needed.
%Thus, a compositional verification technique built atop $\semweaksetgrnbad{\cdot}$ would be more expressive, 
%and thus more complex, than necessary to answer questions about $\semsingsetgrn{\cdot}$.
  %Having understood that $\semweaksetgrnbad{\cdot}$ is too fine, 
  %one may try to define a different, simple semantics 
  An alternative, coarser approach might be
  $\semweaksetgrnnc{\cdot}$, where
  an input vector containing at least one diverging entry is mapped entirely to $\divg$: 
  e.g., 
  $\semweaksetgrnnc{\Ewhile{x < 1}{\Eassign{x}{x - 1}}}
  ([\statelist{\stateentry{x}{0}}, 
    \statelist{\stateentry{x}{1}}]) = \: \divg$ instead of 
    $[\divg, \statelist{\stateentry{x}{1}}]$.
  However, $\semweaksetgrnnc{\cdot}$ turns out to be too \emph{coarse}:
  $\semweaksetgrnnc{\cdot}$ is non-compositional, for reasons 
  similar to \Cref{thm:sing_noninductive} 
  (%\noappendix{see the extended version of the paper~\cite{semarxiv} for details).}
  see Appendix~\ref{app:small-proofs}
  %\noappendix{\xspace in the extended version of the paper~\cite{semarxiv} }
  for details).
  % Thusdesign of a coarsest compositional \green semantics for 
  % vector-states must lie in between these two na\"ive approaches.
  The coarsest compositional \green semantics for 
  vector-states lies in between these two approaches.

\subsubsection{The Right Approach: \texorpdfstring{$\semweaksetgrn{\cdot}$}{}}
To define a semantics with the appropriate granularity,
we adopt the idea of \Cref{ex:bad_vs_grn}, but simplify the output of $\semweaksetgrnbad{S}(P) =  X$ 
to remove unnecessary information about $S$.
Given the output vector-state set $X$,
our approach
boils down to first \rone truncating all vectors in $X$ at the first occurrence of $\divg$, denoted by $\mathit{truncate}(X)$, 
and second \rtwo using the following two rules to reduce $X$ (denoted by $\Red(X)$):
    \begin{enumerate}
        \item If $[a_1, \cdots, a_n, \divg] \in X$ and $[a_1, \cdots, a_n, a_{n+1}, \cdots, a_m, \divg] \in X$, then $[a_1, \cdots, a_n, a_{n+1}, \cdots, a_m, \divg]$ does not appear in the reduction of $X$.
        \item If $[a_1, \cdots, a_n, \divg] \in X$ and an infinite vector $[a_1, \cdots, a_n, a_{n+1}, \cdots] \in X$, then $[a_1, \cdots, a_n, a_{n+1}, \cdots]$ does not appear in the reduction of $X$.
    \end{enumerate}
The coarsest compositional divergent-aware semantics $\semweaksetgrn{\cdot}$ can then be defined, for a set of programs $\progset$ and a set of input vector-states $V$,
$\semweaksetgrn{\progset}(V) = \Red(\mathit{truncate}(\semweaksetgrnbad{\progset}(V)))$.
\begin{example}[Correct $\semweaksetgrn{\cdot}$ Semantics] \label{ex:good_vs_grn}
We illustrate the behavior of $\mathit{truncate}(X)$, $\Red(X)$, and $\semweaksetgrn{\cdot}$, 
by illustrating their application to the 
set of programs and inputs from \Cref{ex:bad_vs_grn}.
{
 \small
    \begin{align*}
%        \semweaksetgrn{S}(\{[\statelist{\stateentry{x}{1}}]\}) &= \Red(\mathit{truncate}(\semweaksetgrnbad{S}(\{[\statelist{\stateentry{x}{1}}]\})))\\
%        &= \Red(\mathit{truncate}(\{[\divg], [\statelist{\stateentry{x}{1}}]\}))\\
%       &= \{[\divg], [\statelist{\stateentry{x}{1}}]\}\\ 
        \semweaksetgrn{S}(\{[\statelist{\stateentry{x}{1}}, \statelist{\stateentry{x}{2}}]\}) &= \Red(\mathit{truncate}(\semweaksetgrnbad{S}(\{[\statelist{\stateentry{x}{1}}, \statelist{\stateentry{x}{2}}]\}) ))\\
        &= \Red(\mathit{truncate}(\{[\divg, \statelist{\stateentry{x}{2}}], [\statelist{\stateentry{x}{1}}, \divg]\})\\ 
        &= \Red(\{[\divg], [\statelist{\stateentry{x}{1}}, \divg]\})\\
        &= \{[\divg]\}\\ 
        \semweaksetgrn{S}(\{[\statelist{\stateentry{x}{2}}, \statelist{\stateentry{x}{3}}, \cdots]\}) &= \Red(\mathit{truncate}(\semweaksetgrnbad{S}(\{[\statelist{\stateentry{x}{2}}, \statelist{\stateentry{x}{3}}, \cdots]\})))\\
        &= \Red(\mathit{truncate}(\{[\statelist{\stateentry{x}{2}}, \statelist{\stateentry{x}{3}}, \cdots], [\divg, \statelist{\stateentry{x}{3}}, \cdots]\})) \\
        &= \Red(\{[\statelist{\stateentry{x}{2}}, \statelist{\stateentry{x}{3}}, \cdots], [\divg]\})) \\
        &= \{[\divg]\}
    \end{align*}
    }
\end{example}
Truncation and reduction correspond to removing all information about $S$ that could not be determined from the 
\green program-agnostic semantics of sets of loops $W$ defined in terms of $S$. 
% Intuitively, the compositionality of $\semweaksetgrn{\cdot}$ and fact that $\semweaksetgrn{\cdot} \preceq \semsingsetgrn{\cdot}$ mean that $\semweaksetgrn{S}$ need only contain the information about $S$ that can be determined from such $\semsingsetgrn{W}$.
This idea is clarified in 
%\noappendix{the extended version of the paper~\cite{semarxiv}}.
\Cref{app:vec_grn}.
%\noappendix{, available in the extended version of the paper~\cite{semarxiv}}.
Removing this unnecessary information makes $\semweaksetgrn{\cdot}$ the coarsest-grained compositional semantics that is finer than $\semsingsetgrn{\cdot}$.
% The definition of $\semweaksetgrn{\cdot}$ seems complex---but as previously stated, it is both compositional, and the coarsest-grained compositional semantics 
% that is finer than $\semsingsetgrn{\cdot}$.

\begin{restatable}[Compositionality of $\semweaksetgrn{\cdot}$]{thm}{grnwkcomp}\label{thm:weak_inductive}
    The \green vector-state semantics over sets of programs $\semweaksetgrn{\cdot}$ admits a compositional characterization.
\end{restatable}

The proof of \Cref{thm:weak_inductive} very closely mirrors the proof of \Cref{thm:weak_compositional_yel}, 
with the slight differences that \rone as a shortcut, we can determine the semantics of guards in $B$ directly 
by checking $\semweaksetgrn{B}(x)$ where $x$ is an enumeration of $\State\restriction_{\varfun(B)}$, and 
\rtwo we must handle infinite (and diverging) traces in addition to finite, nondiverging traces.

\begin{restatable}[$\semweaksetgrn{\cdot}$ is Coarsest Compositional for $\semsingsetgrn{\cdot}$]{thm}{grnwkmin} \label{thm:weak_minimality}
    For every compositional semantics 
    $\abstrs{A}$ such that $\semsingsetgrn{\cdot}\preceq \abstrs{A}$,
    we have that $\semsingsetgrn{\cdot} \preceq \semweaksetgrn{\cdot}\preceq \abstrs{A}$.
\end{restatable}

The proof of \Cref{thm:weak_minimality}
mimics the proof of \Cref{thm:weak_minimality_yel} with the extra challenge of dealing with divergence.
We refer the reader to 
%\noappendix{the extended version of the paper~\cite{semarxiv}}
\Cref{app:vec_grn}
%\noappendix{, 
%available in the extended version of the paper~\cite{semarxiv},} 
for the full proofs of 
Theorems~\ref{thm:weak_inductive} and \ref{thm:weak_minimality}.

Furthermore, we observe that $\semweaksetgrn{\cdot}$ is strictly finer than
$\semweaksetyel{\cdot}$. 
This relation is depicted in the lattice in \Cref{fig:granularity-lattice}, 
and proven in 
%\noappendix{the extended version of the paper~\cite{semarxiv}}.
\Cref{app:vec_compare}.
%\noappendix{\xspace of the extended version of the paper~\cite{semarxiv}}.

% \input{4expressivity}
\section{Compositional Hoare-style Rules} \label{sec:proof_systems}

We have argued that when a semantics is compositional, 
it is easier to design formal systems for reasoning about sets of programs.
To validate this claim, we return to the idea of designing a Hoare-style proof system for inductively defined sets of programs. 
First, we show that in the same way there does not exist
a compositional characterization for program-agnostic semantics of sets of programs, 
there is no sound and relatively complete compositional Hoare-style While-rule for 
{automatically} proving properties of sets of programs with program-agnostic triples (\Cref{thm:no-compo-prog-ag-while}).
% I would put back the subsection for the first result.} \sn{I am fine with either referencing the thm here directly or re-adding the subsection header.}
% 
Next, we show that for the two vector-state semantics 
$\semweaksetyel{\cdot}$ and $\semweaksetgrn{\cdot}$,
% for which 
% we have proven 
% that both variants 
which both admit compositional characterizations,
there \emph{do} exist sound and relatively complete compositional Hoare-style While-rules (\Cref{sec:vs_have_while_rules}). 
Combined with previous work on Unrealizability Logic~\cite{uls,ulw}---a 
relatively complete, compositional Hoare-style proof system for sets of loop-free programs---these 
rules provide the first
sound, relatively complete, compositional Hoare-style logics for reasoning about sets of programs.

%\removable{Because previous work~\cite{uls, ulw} has already given compositional 
%Hoare-style rules for programs without loops, 
%we will focus our discussion on developing/proving the non-existence of a 
%Hoare-style while-loop rule.}
We start by defining what it means for a while-loop rule to be compositional 
for a Hoare-style proof system for sets of programs under a specific semantics.
% 
% The following terminology reflects what is used in Unrealizability Logic~\cite{uls,ulw}.
% 
Given a semantics $\abstrsnovars{A}$ where the signature 
$A(\tau)$ has type $\tau = X \rightarrow X$ for some $X$, 
an \emph{unrealizability triple} 
$\utripleof{P}{S}{Q}{A}$ denotes that $\abstrsapp{A}{S}(P) \subseteq Q$ holds.
When writing proofs by induction (i.e., for inductively defined sets of programs), 
one uses a set $\Gamma$ to hold their inductive hypotheses.
When $\Gamma$ is a set of such triples, we 
use $\Gamma \vdash \utripleof{P}{S}{Q}{A}$ 
to mean that $\utripleof{P}{S}{Q}{A}$ holds, assuming that
all triples in $\Gamma$ hold.
%We will not use $\Gamma$ here, but it is included for completeness.

\begin{definition} [Compositional While-Rule] \label{Def:compositional_while_rule}
    Given a semantics $\abstrsnovars{A}$, a While-rule is compositional if it has the following form, where
    $\varphi(v_1, \cdots, v_n)$ is
    a formula over $v_1, \cdots, v_n$:
    \begin{prooftree}
        \AxiomC{$\Gamma \vdash \utripleof{P_B}{B}{Q_B}{A}$}
        \AxiomC{$\Gamma \vdash \utripleof{P_S}{S}{Q_S}{A}$}%\RightLabel{Wishful While}
        \BinaryInfC{$\Gamma \vdash
        \utripleof{P}{(\Ewhile{B}{S})}{\varphi(\varfun(B),\varfun(S),P,P_B,P_S,Q_B,Q_S)}{A}$}
    \end{prooftree}
\end{definition}
\Cref{Def:compositional_while_rule} disallows one from referring directly to $B$ and $S$ 
in the postcondition; only information that can be captured by a compositional semantics (e.g., 
the input and output sets $P_B$ and $Q_B$ for $B$) can be referenced.
This restriction captures compositionality for inference rules.

Hoare-style logics often also contain ``structural'' rules that do not correspond to a single constructor (e.g., the weakening rule in standard Hoare logic~\cite{hoare}).
%which allows a post-condition to be weakened
%Because structural rules may provide extra power that a compositional characterization of semantics does not admit, structural rules can allow triples $\utripleof{P}{\progset}{Q}{A}$ for a semantics $\abstrsnovars{A}$ to be derived, even if it would not be possible to derive $\abstrsapp{A}{\progset}(P) \subseteq Q$ in a compositional fashion.
An important structural rule
we will consider in this paper is the $\grmdisj$ rule 
from Unrealizability Logic~\cite{uls, ulw}, 
which allows one to partition a set of programs into a finite number of subsets.

\begin{definition} [The Grammar-Disjunction Rule~\cite{uls}] \label{def:grmdisj}
    Given a semantics $\abstrsnovars{A}$, the grammar-disjunction rule $\grmdisj$ is defined as follows:
    \[
      \infer[\grmdisj]
      {\Gamma \vdash \utripleag{P}{(\progset_1 \cup \progset_2)}{Q}}
      {
        \Gamma \vdash \utripleag{P}{\progset_1}{Q} \quad
        \Gamma \vdash \utripleag{P}{\progset_2}{Q}
      }
    \]
%    and the Weaken rule is defined as follows:
%    \begin{prooftree}
%        \AxiomC{$\Gamma \vdash \utripleag{P}{\progset}{Q}$}
%        \AxiomC{$\Gamma \vdash P\rightarrow P' \wedge Q'\rightarrow Q$}\RightLabel{Weaken}
%        \BinaryInfC{$\Gamma \vdash \utripleag{P'}{\progset}{Q'}$}
%    \end{prooftree}
\end{definition}

$\grmdisj$ allows a proof system more power than a strictly compositional 
characterization would: a proof system with $\grmdisj$ allows
one to partition sets of loops into smaller sets, while a proof system without 
$\grmdisj$ is denied this ability.
When proving or disproving the existence of a compositional While-rule, we wish to 
account for this additional power.
Thus, we show that 
there exists no While-rule even if $\grmdisj$
(and the standard weakening rule)
are
allowed---which in turn shows that 
Unrealizability Logic over a program-agnostic semantics is fundamentally incomplete.
%GrmDisj, Weaken, and the While-rule are the only rules that can match the signature of a while constructor---i.e., the only rules that can be used to start a proof about while-programs.

% \loris{todo: describe weaken like you did for grmdisj}
%  Weaken (which lets $\utriple{P}{\progset}{Q}$ imply $\utriple{P'}{\progset}{Q'}$ when $P \rightarrow P'$ and $Q' \rightarrow Q$).
% In the next section, we show that even such an additional rule does not salvage the noncompositionality of $\semsingset{\cdot}$.

%\mypar{No Program-Agnostic While-Rule} \label{sec:no_sing_state_while_rule}
%\removable{There is no sound, relatively complete, compositional Hoare-style While-rule 
%for program-agnostic semantics, even when we allow proofs to have
%finitely many applications of the GrmDisj rules.}
%\footnote{
%  We also consider the standard weakening rule as an additional structural rule, 
%  but this rule almost does not have any effect on the proof of \Cref{thm:no-compo-prog-ag-while}.
%}
% 
%Thus, Unrealizability Logic is fundamentally incomplete when considering a program-agnostic semantics.
% 
\begin{restatable} [No Program-Agnostic While-Rule]{thm}{nocompoprogagwhile}
\label{thm:no-compo-prog-ag-while}
    There is no sound and relatively complete compositional While-rule 
    % for Unrealizability Logic
    over triples $\utripleag{P}{\progset}{Q}$, even when finitely many applications of $\grmdisj$ (and the standard weakening rule) are allowed.
\end{restatable}

The proof,
%(\noappendix{given in the extended version of the paper~\cite{semarxiv}})
given in \Cref{app:ul_proofs},
%\noappendix{\xspace of the extended version of the paper~\cite{semarxiv}},
is similar to the proof of \Cref{thm:sing_noninductive}, except that it considers 
an infinite family of sets of guards to account for  the $\grmdisj$ rule.

%The proof is given in \Cref{app:ul_proofs}, submitted as supplementary material. 
%It is similar to the proof of \Cref{thm:sing_noninductive}, 
%but we consider an infinite family of sets of guards rather than just 
%$B_1$ and $B_2$ to account for the GrmDisj inference rule.

% \Cref{thm:no-compo-prog-ag-while} shows that no compositional While rule that is compositional can be complete, even if GrmDisj is allowed. 

\subsection{Hoare While-Rules for Vector States} \label{sec:vs_have_while_rules}

In contrast with program-agnostic semantics, the compositionality of vector-state semantics makes 
building compositional While-rules easy---we just encode the compositional characterizations.
% follow the proofs of compositionality.
%\removable{When combined with rules similar to those in \cite{uls} or \cite{ulw}, the While-rules we present result in sound and relatively complete proof systems.}

% \subsubsection{Strong Vector-State While Rule}

We begin with a While-rule for the \yellow
vector-state semantics $\semweaksetyel{\cdot}$.

\begin{restatable}
     [\yellowCaps Vector-State While-Rule]{thm}{vsyelwhilerule}
     There is a sound and relatively complete compositional While-rule for Unrealizability Logic for the semantics $\semweaksetyel{\cdot}$.
\end{restatable}
\begin{proofsketch}
    We can adapt $f^\setweak_{\textit{while}}$---the function that determines $\semweaksetyel{\Ewhile{B}{S}}$ from $\semweaksetyel{B}$ and $\semweaksetyel{S}$ in the proof of \Cref{thm:weak_compositional_yel}---to construct an inference rule of the following form:
    \[
    \infer[\mathsf{While_{\yellow}}]{\Gamma \vdash \utriplevs{P}{(\Ewhile{B}{S})}{Q}}{\Gamma \vdash \utriplevs{x=z}{B}{Q_B} \quad \Gamma \vdash \utriplevs{x=z}{S}{Q_S}}
    \]
    The full proof 
    %\noappendix{is available in the extended version of the paper~\cite{semarxiv}}.
    is available in \Cref{app:ul_proofs}.
    %\noappendix{, available in the extended version of the paper~\cite{semarxiv}}.
\end{proofsketch}

In this way, compositionality is a useful property in defining sound and complete compositional systems of reasoning. 
A similar rule works for 
$\semweaksetgrn{\cdot}$, as stated in \Cref{thm:green-vector-rule}.

\begin{restatable}[\greenCaps Vector-State While-Rule]{thm}{vsgrnwhilerule}
\label{thm:green-vector-rule}
    There is a sound and relatively complete compositional While-rule for Unrealizability Logic for the semantics $\semweaksetgrn{\cdot}$.
\end{restatable}
\begin{proofsketch}
The full rule is given in 
%\noappendix{the extended version of the paper~\cite{semarxiv},}
\Cref{app:ul_proofs},
%\noappendix{\xspace (available in the extended version of the paper~\cite{semarxiv})}, 
    which merely applies $f^\setweakgrn_{\textit{while}}$ from the proof of \Cref{thm:weak_inductive}. The rule differs from the rule $\mathsf{While_{\yellow}}$ in that it
    introduces a variable $i \in \N$ to simulate making countably many queries to $\semweaksetgrn{S}$:
    % , and it uses a shortcut to check for admissible guards by collecting the semantics of all guards in $B$:
    $$\infer[\mathsf{While_{\green}}]{\Gamma \vdash 
    \utriplevs{P}{\Ewhile{B}{S}}{Q}
    }{\Gamma \vdash \utriplevs{P_B}{B}{Q_B(x)} \quad \Gamma \vdash \utriplevs{P_S(i)}{S}{Q_S(i)}}$$
\end{proofsketch}

Both rules can be augmented with rules like those in \cite{uls} or \cite{ulw} 
to yield the first sound and relatively complete proof 
systems for inductively defined sets of programs in $\gimp$. 
We observe that, 
although complex predicates are required in the 
$\mathsf{While}$ rules to achieve 
completeness,
in practice one can often
find a proof
using simpler predicates 
as invariants.
A similar phenomenon can be found in standard Hoare logic, 
which requires complex invariants to achieve completeness in theory,
but where in practice one can often 
find a proof
using simpler invariants.

\section{Related Work} \label{Se:related}

\paragraph{Program Synthesis and Unrealizability}
The inputs to the semantics discussed in this paper align 
well with
% \removable{the inputs of}
program-synthesis problems, which often 
assume a syntactic search space constrained by 
a regular tree grammar 
or some other recursively defined language
(\cite{sygus, semgus, rosette, conflict},
to single out just a few).
As discussed in \Cref{sec:weak-vector-sem}, 
the results on single-input 
properties discussed 
in this paper generalize almost immediately to 
multiple-input properties, which in turn capture 
PBE specifications used in many 
synthesizers~\cite{sketch, semgus, flashfill, flashmeta}.

In particular, our work provides a formal result
about
the difficulty of proving
% it is to reason about the 
\emph{unrealizability}
of synthesis problems (i.e., the non-existence of a solution).
Two recent publications \cite{uls,ulw} introduced 
\textit{Unrealizability Logics}, 
Hoare-style logics for reasoning about inductively defined sets of programs. 
Both logics are sound, but not relatively complete when loops are allowed (see \Cref{note1}).
The main semantics for sets of programs used in~\citet{ulw} is coarser than 
$\semweakset{\cdot}$; 
% and this paper
\Cref{thm:weak_minimality_yel} proves that a compositional rule for while loops 
under such semantics does not exist. 
In~\citet{uls} (and in the appendices of~\citet{ulw}), 
a vector-state semantics finer than $\semweaksetyel{\cdot}$ is used. 
That semantics is 
% actually
compositional, 
but the inference rules are written with the assumption that a program-agnostic 
rule like the one disallowed by \Cref{thm:no-compo-prog-ag-while} exists. 
Vector-states are only added to 
% allow the logic to describe certain kinds of properties
discuss $k$-input extensional properties; 
their role in enabling compositionality is not 
% acknowledged or 
exploited. 
%Thus, neither work presents a proof system that is relatively complete for sets of loops.

\citet{ulw} designed a proof synthesizer for an incomplete Unrealizability Logic for infinite vector-states.
Their implementation suggests that reasoning over infinite
vector-states is tractable, at least in simple cases,
giving hope 
for tractable verification techniques to be derived from $\semweaksetyel{\cdot}$.

% IR
\paragraph{Outcome Logic and Hyper Hoare Logic}
\label{sec:hhl}\label{sec:ol}
An Outcome Logic is a program logic for uninterpreted nondeterministic programs expressed in a Kleene algebra. In an Outcome Logic, program semantics are monadic, and the monad is equipped with the operations of a partial commutative monoid to split monadic objects.
\citet{outcomelogic} give an instantiation of Outcome Logic for the set monad equipped with a set-union monoid operation.
Hyper Hoare Logic is a program logic for imperative nondeterministic programs in which pre- and postconditions are second-order predicates over sets of states, rather than first-order predicates over states~\cite{hyperhoarelogic}. 

The logic from \Cref{sec:proof_systems} is not subsumed by existing logics, although Outcome
Logic
and Hyper Hoare Logic share some similarities.
Because both logics reason about sets of states explicitly, and a vector-state is simply an indexed set of states, triples about a single program that are expressible in our \yellow vector-state logic are expressible in both Outcome and Hyper Hoare Logic. Triples about a single program in \green logic cannot be expressed in Hyper Hoare Logic because there is no treatment of divergence there; they may be expressible in an Outcome Logic if one can define an $\vsvars$ monad, where $\vsvars$ is the set of vector-states over which 
$\semweaksetgrn{\cdot}$ operates.
Note that neither logic expresses properties of sets of programs. 

Because Outcome Logic is expressed over an uninterpreted Kleene algebra, 
one might wonder whether we can simply interpret elements of the Kleene algebra 
as sets of programs to obtain a complete Outcome Logic for sets of programs \
revision {(i.e., can we represent a set of programs as a single nondeterministic program 
and use an Outcome Logic to reason about its behavior on vector-states?)}. 
Unfortunately, this reduction is not possible in general:
the traces of infinite sets of while loops, such as the one in our proof of 
\Cref{thm:no-compo-prog-ag-while}, do not always form regular languages and thus are not expressible by
Kleene algebras 
(e.g., by the Myhill-Nerode Theorem~\cite{myhillnerodetheorem}).

\paragraph{Incompleteness of Hoare-Style logics.}
Clarke~\cite{clarkel4} studies programming languages for which one cannot obtain 
a relatively complete ``Hoare-style logic'', 
which Clarke 
(and 
others, most notably~\citet{lipton})
interprets 
as a \emph{syntax-directed proof system} involving axiomatic rules that 
capture the semantics of a program~\cite{hoarecharacterization, 50years}.
Clarke proves that such logics cannot exist for certain 
combinations of language features 
(e.g., coroutines with parameterless recursive procedures).

Such ``Hoare-style logics'' 
are not necessarily compositional as defined in this paper.
In particular, such ``Hoare-style logics'' may contain as 
part of their proof system a separate procedure that reduces deriving a triple 
into some other kind of decidable problem (e.g., enumeration in the case of finite interpretations), 
instead of constructing a proof tree.
Clarke observes that these ``Hoare-style logics''
do not necessarily lead to 
practical proof systems~\cite[p11]{hoarecharacterization}.

Our work defines compositionality as an additional desideratum to support 
practical reasoning; 
one might view compositionality 
as a formal criterion of 
``natural Hoare-style logics''
mentioned by Clarke~\cite{hoarecharacterization}.
We are unaware of any non-compositional logics that see major use.

\paragraph{Variational Analysis.}
Variational analysis (VA) encompasses a broad set of techniques 
for tracking the behavior of finite sets of programs,
where each program is the result of some combination of `tags.'
A combination of tags encodes
the usage of different features, permissions, etc. 
For example, in a set of differently configured email clients, 
the tags $(\mathtt{encryption} = \mathtt{RSA}, \mathtt{spam}=\Et)$ would denote
the unique configuration that uses RSA encryption and a spam filter. 
The main aim of VAs
is to construct a map from tag combinations (i.e., specific programs)
to program properties or behaviors.

For example, the semantics of the choice calculus, 
used in the analysis of variable software, 
builds a map from
\rone 
combinations of finitely many tags describing choices about program syntax 
to 
\rtwo
program behavior~\cite{choicecalculus}.
Model checking-based VA approaches, useful in product-line analysis, 
build more compact maps from
\rone
propositional formulas
that represent
finite sets of configurations to
\rtwo
program behaviors~\cite{plccs}.
% Propositional formula maps, used in product-line analysis, 
% build similar label-to-behavior maps but introduce a more compact 
% representation for highly non-injective functions~\cite{propformmaps}.
Faceted execution, used in dynamic enforcement of privacy policies, 
builds maps from
\rone
users to
\rtwo
program states
to 
% represent the effects of 
model
different users executing the same code~\cite{facetedexecution}. 
%Other VA techniques follow the same basic idea.

We emphasize that VA is a family of specific 
techniques including the above, and does not provide a framework for semantic analysis or derivation of our key results 
(e.g., \Cref{thm:weak_minimality_yel}). 
One might wonder
whether a VA technique (extended to operate over infinite sets) could, 
at least in theory,
efficiently verify extensional properties for sets of programs.
Our theory suggests that the answer is no, for two reasons. 
First, 
a complete, compositional VA technique capable of 
verifying extensional properties will induce a semantics strictly 
finer than $\semweaksetyel{\cdot}$ due to the need to track 
tag combinations, which correspond to 
syntactic information.\footnote{To be precise, to construct
a semantics for VA techniques, one must extend $\gimp$ to include tags. 
This extension is necessary because VA techniques construct 
different maps for differently tagged sets of programs.}
Second, 
expressing maps from tag combinations to semantic behaviors 
requires non-Cartesian predicates that explicitly relate program syntax and semantics; these predicates are complex and best avoided for infinite sets of programs
(see \Cref{subsec:interpreter}).

Finally, we remark that \citet{va_semantic_analysis} define
a notion of
denotational semantics for variability languages
to compare
the expressive power of different VA formalisms. Both our work and theirs use denotational semantics to describe the expressivity of verification objectives, albeit for very different purposes. We believe that this strategy 
% \remove{of capturing the expressivity of verification objectives with denotational semantics}
has potential to advance the analysis and design of verification techniques, as both works evidence.
% the idea of capturing verification languages as semantics to understand properties of verification techniques is a common theme that we believe has a lot of potential.

\paragraph{
% Reasoning about Sets of Programs.
Compilers and Self-Modifying Code.
}
A compiler is
the canonical example of a tool that creates a set of programs.
Generally, ``compiler correctness'
% the question of ``proof of compiler correctness'' 
is concerned with whether each (single) source program is 
translated correctly to an appropriate target program.
The scheme for establishing 
compiler correctness
% such a property
is due to Burstall 
\cite{DBLP:journals/cj/Burstall69}
%Essentially the same scheme was 
and followed in most of the 14 passes in 
the CompCert compiler, 
which translated CLight to PowerPC assembly language
\cite{DBLP:journals/cacm/Leroy09}.

A compiler-correctness problem closer to the questions addressed 
in this paper would be, e.g., 
to prove that a compiler can only emit target programs 
that satisfy a memory-safety property 
(i.e., the \emph{set} of programs emittable
by the compiler's translation function are memory-safe).
Myreen \cite{myreen2010verified} addressed a version of this problem: 
he developed a machine-code Hoare logic to support 
reasoning about self-modifying x86 machine code, 
which was applied to the verification of JIT compilers.

% , establishing the following property \cite{myreen2010verified}:
% \begin{description}
%   \item[$\quad$]
%   ``$\ldots$ the JIT compiler correctly implements $\xrightarrow{\mathit{exec}}$: given a stack $(\mathit{xs}, l)$, a string representing the encoding of a bytecode program \textit{cs}, and enough space to fit the generated x86 code, this JIT compiler will always terminate in a state where the stack has been updated according to $\xrightarrow{\mathit{exec}}$."
% \end{description}

Other work on self-modifying code includes Gerth \cite{YCS:Gerth91} and Cai et al.\ \cite{PLDI:CSV07}.
The latter is based on an extension of Hoare logic, but to support self-modification, it treats program instructions as any other mutable data structure---an issue that we do not face in our work.
Gerth's earlier work also adopts a data-centric view of the world: there is no code, only data.
Gerth's proof system was based on Manna and Pnueli \cite{DBLP:conf/popl/MannaP83} 
and proves properties expressed in linear temporal logic.

While these problems are beyond the scope and capabilities of our work, the semantics and Hoare-style rules studied in our paper provide a new way of reasoning about sets of programs, including languages (domain-specific or produced by compilers) and self-modifying code.

\section{Conclusion} \label{Se:conclusion}

This paper set out to understand why existing compositional program logics for sets of programs~\cite{uls, ulw}
fail to be relatively complete, and whether such logics can be modified to achieve completeness.
By formalizing verification relative to a denotational semantics for sets of programs, we were able to prove that no complete compositional proof system exists for the kind of logical triples proposed in prior work.
In light of this result,
we define a more expressive kind of triple
and give \textit{the first (two) sound and relatively complete Hoare-style logics} for inductively defined sets of programs (even for programs that contain loops), solving an open problem in the literature~\cite{uls, ulw}.

\begin{acks}
This work was supported, in part, by
a \grantsponsor{00002}{Microsoft Faculty Fellowship}{},
a gift from \grantsponsor{00001}{Rajiv and Ritu Batra}{},
\grantsponsor{00003}{NSF}{https://www.nsf.gov/}
under grants
\grantnum{00003}{CCF-1918211},
\grantnum{00003}{CCF-2023222},
\grantnum{00003}{CCF-2211968},
\grantnum{00003}{CCF-2212558},
\grantnum{00003}{CCF-2402833},
\grantnum{00003}{CCF-2422214}, 
and \grantnum{00003}{CCF-2446711}, 
and a grant from the Korea Foundation of Advanced Studies.
This material is based upon work supported by the National Science Foundation Graduate 
Research Fellowship Program under Grant No. \grantnum{00003}{DGE-2038238}.
Any opinions, findings, and conclusions or recommendations
expressed in this publication are those of the authors,
and do not necessarily reflect the views of the sponsoring
entities.
\end{acks}

\section*{Data-Availability Statement}
Proofs of the paper's theorems are available in the 
full version of this paper~\cite{semarxiv}. 
Our paper presents only theoretical results, so there is no accompanying artifact.

\newpage
\bibliography{reference}

\clearpage

\appendix
\section{Proofs}\label{app:proofs}
In this appendix, we include the proofs of theorems that were stated but not included in the main text.

\subsection{Compositionality of \yellowCaps Weak Vector-State Semantics}\label{app:weak_compositionality_yel}

% \begin{theorem} [Compositionality of $\semweaksetyel{\cdot}$]
%     \yellow weak vector-state semantics over sets of programs $\semweaksetyel{\cdot}$ admits a compositional characterization.
% \end{theorem}
\yelwkcomp*

\begin{proof}
    Figure~\ref{fig:weak_compositional} gives a compositional characterization of $\semweaksetyel{\cdot}$ over unions and all constructors except while. The only cases worth discussing are if-then-else and while. 

    Note that $\semweakset{\progset}$ is determined by its behavior on singleton sets of vector-states, so we restrict our focus to inputs of the form $\{v\} \subset \vsvars$ when outlining the next two rules. The versions over arbitrary sets of vector-states $V$ are given in \Cref{fig:weak_compositional}.

    \paragraph{If-Then-Else} 
        
        Consider sets of statements $S_1$ and $S_2$ and a set of guards $B$. We would like $\semweaksetyel{\Eifthenelse{B}{S_1}{S_2}}(v)$ to, for each guard $b$, evaluate $\semweaksetyel{S_1}$ on a vector of the entries of $v$ where $b$ is true, and evaluate $\semweaksetyel{S_2}$ on a vector of the entries of $v$ where $b$ is false. Then, we'd like to stitch these results back together. 
        For example, $\semweaksetyel{\Eifthenelse{x_1 \neq 1}{S_1}{S_2}}([\statelist{\stateentry{x_1}{1}},\statelist{\stateentry{x_1}{4}}])$ contains $[\state_a, \state_b]$ iff
        $[\state_a] \in \semweaksetyel{S_2}([\statelist{\stateentry{x_1}{1}}])$
        and
        $[\state_b]$ $\in \semweaksetyel{S_1}([\statelist{\stateentry{x_1}{4}}])$.

        For a conditional $\Eifthenelse{B}{S1}{S2}$, each vector state
        $v^b \in \semweaksetyel{B}(v)$ splits $v$ into two subarrays:
        one array over the indices on which $v^b$ is true
        ($v^{1} = \filter{v}{v^b}$), and one over the indices on which $v^b$ is false ($v^{2} = \filter{v}{\neg v^b}$).
        Because $v$ and $v^b$ are finite, $v_1$ and $v_2$ are also finite, so we can determine $\semweaksetyel{S1}(v^{1})$ and $\semweaksetyel{S2}(v^{2})$. For each $u^{2} \in \semweaksetyel{S1}(v^{1})$ and $u^{2} \in \semweaksetyel{S2}(v^{2})$, we can assemble a finite vector $u$ by combining $u^{1}$ and $u^{2}$ according to $v^b$ (we say $u = \mathit{interleave}(u^{1}, u^{2}, v^b)$).
        The set of such $u$ is precisely $\semweaksetyel{\Eifthenelse{B}{S1}{S2}}(v)$.

    \paragraph{While}
        We describe 
        $f^{\setweak}_{while}$ as follows:
        
        Let a set of programs $\Ewhile{B}{S}$ and a vector state $v \in \State^*$ be given. Note that when $a$ is a vector, we will use $\subvec{a}{i}{j}$ to denote $[a_i, \cdots, a_j]$. Vectors like $[\subvec{a}{i}{j}, b]$ will be understood to be flattened---$[a_i, \cdots, a_j, b_1, \cdots, b_n]$.
        % So, we can ignore input vectors that end in $\divg$.

        For each possible semantics of the loop guard, we identify traces through the body of the loop that would cause the loop to converge ($T_{\state, \state'}$) from start state $\state$ to an end state $\state'$, and we ask $\semweaksetyel{S}$ whether any $s \in S$ can produce these traces.
    
        For example, suppose that we want to determine $\semweaksetyel{\Ewhile{x_1 > 3}{S}}([\statelist{\stateentry{x_1}{4}}])$, given a set of statements $S$ where $\varfun(S) = \{x_1\}$. We can consider all finite ``traces'' of executions of a body of this loop. Thus, $[\statelist{\stateentry{x_1}{4}}, \statelist{\stateentry{x_1}{12}}, \statelist{\stateentry{x_1}{1}}]$ would represent a loop body that, on the first iteration, sent $\statelist{\stateentry{x_1}{4}}$ to $\statelist{\stateentry{x_1}{12}}$ and, on the second iteration, sent $\statelist{\stateentry{x_1}{12}}$ to $\statelist{\stateentry{x_1}{1}}$. We can determine $\semweaksetyel{\Ewhile{x_1 > 3}{S}}([\statelist{\stateentry{x_1}{12}}])$ by querying $\semweaksetyel{S}$ on all such traces.
           
        To start, suppose that we want $\semweaksetyel{\Ewhile{B}{S}}([\sigma])$ for $\sigma \in \State$.
        
        For each $\sigma' \in \State$, call $T_{\state, \state'}$ the set of vectors of states (``traces'') $[t_1, \cdots, t_n] \in \State^*$ so that:
        \begin{itemize}
            \item $t_1 = \state$
            \item $t_n = \state'$
        \end{itemize}
        The set $T_{\state, \state'}$ represents potential converging traces of the loop from the start state $\state$ to the end state $\sigma'$. In order for the trace $[t_1, \cdots, t_n] \in T_{\state, \state'}$ to be realizable by some loop in the set, it must be the case that:
        \begin{itemize}
            \item $[\Et, \Et, \cdots, \Et, \Ef] \in \semweaksetyel{B}([t_1, \cdots, t_n])$ -- there is a guard which permits the trace.
            \item $[t_2, \cdots, t_n] \in \semweaksetyel{S}([t_1, \cdots, t_{n-1}])$ -- there is a loop body which permits the trace.
        \end{itemize}
        Thus, $[\state'] \in \semweaksetyel{\Ewhile{B}{S}}([\sigma])$ if and only if there exists a trace $[t_1, \cdots, t_n] \in T_{\state, \state'}$ satisfying both conditions (which are checkable from $\semweaksetyel{B}$ and $\semweaksetyel{S}$).
    
        Given $\semweaksetyel{B}$ and $\semweaksetyel{S}$, we can determine $\semweaksetyel{\Ewhile{B}{S}}$ on vectors of length 1. We can extend our construction to determine $\semweaksetyel{\Ewhile{B}{S}}$ on finite vectors $[\state_1, \cdots, \state_m]$ of arbitrary length by concatenating the traces from $T_{\state_j, \state'_j}$ in our queries to $\semweaksetyel{B}$ and $\semweaksetyel{S}$.

        For each potential output vector $[\state'_1, \cdots, \state'_m]$, define $T = \times_{j \leq m} T_{\state_j, \state_j'}$. Each element in $T$ is a concatenation of finitely many finite traces $[\subvec{t^{1}}{1}{n_1}, \subvec{t^{2}}{1}{n_2}, \cdots, \subvec{t^{2}}{1}{n_m}]$. This collection of traces is realizable by some loop in the set if and only if:
        \begin{itemize}
            \item $[\Et, \Et, \cdots, \Et, \Ef, \cdots, \Ef] \in \semweaksetyel{B}([\subvec{t^{1}}{1}{n_1-1}, \cdots, \subvec{t^{2}}{1}{n_m - 1}, \sigma'_1, \cdots, \sigma'_m])$ -- there is a guard which permits all the traces, accepting all non-final states and rejecting all final states $\sigma'_j$.
            \item $[\subvec{t^{1}}{2}{n_1}, \cdots, \subvec{t^{2}}{2}{n_m}] \in \semweaksetyel{S}([\subvec{t^{1}}{1}{n_1-1}, \cdots, \subvec{t^{2}}{1}{n_m - 1}])$ -- there is a loop body which permits the trace.
        \end{itemize}
        Again, $[\state'_1, \cdots, \state'_m] \in \semweaksetyel{\Ewhile{B}{S}}([\state_1, \cdots, \state_m])$ if and only if there exists a concatenation of traces $[\subvec{t^{1}}{1}{n_1}, \cdots, \subvec{t^{2}}{1}{n_m}] \in T$ satisfying both conditions (which are again checkable from $\semweaksetyel{B}$ and $\semweaksetyel{S}$).

        Thus, from $\semweaksetyel{B}$ and $\semweaksetyel{S}$, we can uniquely identify $\semweaksetyel{\Ewhile{B}{S}}$.
\end{proof}

\subsection{Coarseness of \yellowCaps Weak Vector-State Semantics}\label{app:weak_coarseness_yel}

\yelwkmin*

\begin{proof}
    Clearly, $\semweaksetyel{\progset}$ determines $\semsingsetyel{\progset}$ because $\semsingsetyel{\progset}(\sigma) = \{u[1] \mid u \in \semweaksetyel{\progset}([\sigma])\}$. 
    This gives a function mapping each $\semweaksetyel{\progset}$ to $\semsingsetyel{\progset}$. So $\semweaksetyel{\cdot}$ is at least as fine as $\semsingsetyel{\cdot}$.

    Let us show $\semweaksetyel{\cdot}$ is the coarsest such semantics. 
    Suppose an inductively defined set of statements
    $S$ and a compositional semantics $\abstrs{A}$ at least as fine as program-agnostic semantics $\semsingset{\cdot}$ are given. 
    We will construct set of loops $W_\alpha$ so that $\abstrsapp{A}{S}$ gives each $\abstrsapp{A}{W_\alpha}$, 
    each $\abstrsapp{A}{W_\alpha}$ gives $\semsingset{W_\alpha}$, 
    and the $\semsingset{W_\alpha}$ together give $\semweakset{S}$. This will show $\abstrsapp{A}{S}$ gives $\semweakset{S}$.
    
    Given $v \in \State^*$, we can determine $\semweaksetyel{S}(v)$ as follows.\\

    Recall $\varfun(S)$ is provided by $\abstrsapp{A}{S}$. Let $j$ be a variable not in $\varfun(S)$.

    Consider $v = [v_1, \cdots, v_n] \in \State^*$, and consider a desired vector of outputs $u = [u_1, \cdots, u_n] \in \State^*$. Clearly, if $u$ does not agree with $v$ off of $\varfun(S)$, then $u \notin \semweakset{S}(v)$. So suppose they do agree.
    
    We can write a set of loops $W_{v,u}$ {that contains one loop for each $s \in S$. For each $s \in S$, the corresponding loop in $W_{v,u}$ runs $s$} on every $v_j$, breaking when the wrong output is produced and incrementing a counter ($j \in \vars \setminus \varfun(S)$) for each $v_j$ correctly mapped to $u_j$. {
    Thus,
    each member of $W_{v,u}$ has the following form:}
\begin{lstlisting}[linewidth=\linewidth]
j := 1;
while (0 < j <= n) {
    (*@$\Eassign{\state}{v_j\restriction_{\varfun(S)}}$; @*) 
    s;
    if (*@$\sigma == {u_j}\restriction_{\varfun(S)} $@*) {
        j := j+1
    } else {
        j := 0
    }
}
\end{lstlisting}

    Then $u \in \semweaksetyel{S}(v)$ if and only if there is a state $\sigma'$ with $j = n+1$ in $\semsingsetyel{W_{v,u}}(\state)$ for arbitrary $\state \in \State$. By investigating various $u$, we can fully determine $\semweaksetyel{S}(v)$ from $\semsingsetyel{W_{v,\cdot}}$. Moreover, because $\abstrs{A}$ is compositional, $\abstrsapp{A}{S}$ determines each $\abstrsapp{A}{W_{v,u}}$. Since $\abstrs{A}$ is at least as fine as $\semsingsetyel{\cdot}$, $\abstrsapp{A}{W_{v,u}}$ determines $\semsingsetyel{W_{v,u}}$. Thus, $\abstrsapp{A}{S}$ determines $\semweaksetyel{S}$.\\

    For integer expressions $E$ (or Boolean expressions $B$), we can replace $S$ with $x_1 := E$ (or $\Eifthenelse{B}{x_1:=0}{x_1:=1}$) and the argument proceeds identically.

    So, $\abstrs{A}$ is at least as fine as $\semweaksetyel{\cdot}$.
\end{proof}

\subsection{Proofs for \Cref{sec:proof_systems}}\label{app:ul_proofs}

\nocompoprogagwhile*

\begin{proof}
    % In the proof of \Cref{thm:sing_noninductive}, we gave two inductively defined sets of guards $B_1$ and $B_2$ with identical program-agnostic semantics, and two inductively defined sets of loops $W_1$ and $W_2$ using these guards with identical bodies. We saw that $\semsingset{B_1} = \semsingset{B_2}$, but $\semsingset{W_1} \neq \semsingset{W_2}$. 
    % % 
    % In that construction, one can break each $W_j$ (or alternatively each $B_j$) into two pieces and analyze the semantics of each piece compositionally. If there is always a way to finitely partition troublesome sets of programs so that each piece can be handled compositionally, we might be happy with a semantics that is not perfectly compositional. This strategy could be implemented using the $\mathsf{GrmDisj}$ inference rule from \cite{uls} or by some clever rewriting of a troublesome grammar.
    % Thus, 
    {Rather than the proof sketched in \Cref{subsec:setlogic}, we give a more constructive proof following the 
    approach used in the
    proof of \Cref{thm:sing_noninductive}. This
    approach
    will allow us to account for the $\mathsf{GrmDisj}$ rule.
    
    Formally, we demand that, with the exception of $\mathsf{GrmDisj}$ and the $\mathsf{While}$ rule, 
    any rule that can prove conclusions of the form $\Gamma \vdash \utripleag{P}{W}{Q}$ 
    for a set of loops $W$ can only use $W$ by copying it 
    in into triples in the rule's premises. 
    That is, such rules' hypotheses can use other triples about $W$, 
    but they cannot not ``unbox'' $W$ into $B$ and $S$, 
    encode the syntax of $W$ into a first-order formula, 
    or otherwise circumvent the need for the $\mathsf{While}$ rule to 
    prove nontrivial properties of sets of loops. 
    We will prove that, if the only rules that can ``see'' 
    the syntax of $W$ are $\mathsf{GrmDisj}$ and a compositional $\mathsf{While}$ rule, 
    then the proof system cannot be both sound and relatively complete. 
    
    Our proof 
    {will proceed via proof by contradiction as follows: 
    First, we construct a true triple $\utripleag{P}{W}{Q}$ 
    over a specific infinite set of loops $W$. 
    Second, we show that, if the proof system is sound and relatively complete, then
    any proof of $\emptyset \vdash \utripleag{P}{W}{Q}$ 
    must contain at least one (nontrivial) application of the $\mathsf{While}$ rule to an 
    inductively definable, infinite subset 
    $\widetilde{W} \subseteq W$ of loops (\Cref{lem:inf_while_app}). 
    Third, we show that no sound compositional $\mathsf{While}$ rule can be applied in such a way, 
    so a proof of $\emptyset \vdash \utripleag{P}{W}{Q}$ cannot exist. 
    We conclude that the proof system cannot be sound and relatively complete.
    As a matter of notation, when $N$ is a set of expressions, 
    we write $E_N$ to denote the set of integers to which terms in $L(N)$ evaluate.}} 

    Consider $B ::= x_1 < N$ where $N ::= 0 \mid N + 2$, and say $W ::= \Ewhile{B}{x_1 := x_1 + 1}$ as before. 
    % Observe that $\semsingset{B} = \semsingset{\{x_1 < n \mid n \in M\}}$ whenever $M \subseteq \N$, $\min(M) = 0$, and $\abs{M} = \infty$.
    % Note that there are countably many possible {$M$} because {$M$ must be} inductively defined.    
    When $\state[x_1] = 0$, {then} $\semsingset{W}(\state) = 
    \{\state\subs{n}{x_1} \mid n \in {2\N}\}$ 
    (i.e., $\utripleag{\state[x_1] = 0}{W}
    {\state[x_1] \in {2\N}}$ 
    is a valid triple, where ${2\N}$ denotes the set of even numbers). 
    If our logic is sound and relatively complete, 
    there must be a proof deriving 
    ${\emptyset \vdash}
    \utripleag{\state[x_1] = 0}{W}{\state[x_1] \in {2\N}}$, 
    {and the proof must obey \Cref{lem:inf_while_app} below.}
    % {Recall that, to handle sets of loops, our} 
    % options are to apply the While-rule, 
    % % to weaken, or to apply $\mathsf{GrmDisj}$
    % {to apply $\mathsf{GrmDisj}$, or to apply rules that do not unpack $W$}. 
    % Suppose that we apply $\mathsf{GrmDisj}$ $0$ or more times to finitely partition $W${. Since the triple is not true of all sets of loops, we must eventually} apply the While-rule to each element of the partition {under some context $\Gamma$ containing no false triples}. 
    
    {
    \begin{lemma} [Nontrivial While-rule Application to an Infinite Set of Loops] \label{lem:inf_while_app}
        Suppose we have a sound and relatively complete proof system for judgments $\Gamma \vdash \utripleag{P}{\progset}{Q}$ in which only the GrmDisj rule and a compositional While rule can ``see'' the syntax of sets of loops as described above. Then
        any proof of 
        $\emptyset \vdash \utripleag{\state[x_1] = 0}{W}
        {\state[x_1] \in 2\N}$ in the proof system
        must contain at least one application of the $\mathsf{While}$ rule whose conclusion 
        $\Gamma \vdash \utripleag{P}{\widetilde{W}}{Q}$ meets the following conditions:
        \begin{itemize}
            \item $\widetilde{W}$ is a set of loops $\widetilde{W} ::= \Ewhile{x_1 < \widetilde{N}}{\Eassign{x_1}{x_1+1}}$ for some infinite, inductively definable $\widetilde{N} \subseteq N$.
            \item Let $\widetilde{n} = \min(E_{\widetilde{N}})$, let $\widetilde{N}' ::= \widetilde{n} \mid \widetilde{N}' + 1$, and let $\widetilde{W}' ::= \Ewhile{x_1 < \widetilde{N}'}{\Eassign{x_1}{x_1+1}}$. Then the judgment $\Gamma \vdash \utripleag{P}{\widetilde{W}'}{Q}$ is false. Equivalently (because the logic is sound and relatively complete), $\Gamma \vdash \utripleag{P}{\widetilde{W}'}{Q}$ cannot be proven. 
        \end{itemize} 
    \end{lemma}
    
%     , for some infinite 
%     $\widetilde{N} \subseteq L(N)$, 
%     our proof must apply the While-rule to prove some claim 
%     $\Gamma \vdash \utripleag{P}{\widetilde{W}}{Q}$ 
%     that does not hold over the set of loops 
%     $W' = \Ewhile{x_1 < [min(\widetilde{N}), \infty)}
%     {\Eassign{x}{x+1}}$ 
%     (where $\widetilde{W} = \Ewhile{x_1 < \widetilde{N}}{\Eassign{x}{x+1}}$).

\noindent
{
\begin{proof}
    Assume, for the sake of deriving a contradiction, that
    no application of the $\mathsf{While}$ rule satisfied these conditions.
    We will show that one could then modify a proof of $\emptyset \vdash \utripleag{\state[x_1] = 0}{W}
    {\state[x_1] \in 2\N}$ to prove a false claim.

    Let $W_1, \cdots, W_k$ be all the finite subsets of $W$ appearing in triples in the proof, 
    and let $N_1, \cdots, N_k$ be the corresponding sets of expressions in their guards. 
    Observe that $E_{\mathit{fin}} = \bigcup_{j=1}^k E_{N_j}$ is bounded. 
    Let $n = \min(E_N \cap (\max(E_{\mathit{fin}}), \infty))$ be the least integer in $E_N$ 
    greater than every integer in $E_{\mathit{fin}}$ (or $0$ if $E_{\mathit{fin}}=\emptyset$). 
    Intuitively, we will add loops to $W$ whose guard's bound is greater than $n$.
    Because
    only the loops in $W_1, \cdots, W_k$ are checked ``closely'', 
    adding these loops will not affect the correctness of the proof-rule 
    applications, but will result in an invalid conclusion.  
    
    Set $N'$ as an inductively defined set of expressions for which
    $E_{N'} = E_N \cup [n, \infty)$. 
    Now, we can modify our proof of $\emptyset \vdash \utripleag{\state[x_1] = 0}{W}{\state[x_1] \in 2\N}$ 
    to prove the false claim 
    $\emptyset \vdash \utripleag{\state[x_1] = 0}{\Ewhile{x_1 < N'}{\Eassign{x_1}{x_1+1}}}{\state[x_1] \in 2\N}$. 
    Recall that our proof may partition $W$ into many pieces using the GrmDisj rule. We will keep every finite set of loops the same. However, we will replace every infinite set of loops appearing in the proof as follows: If $W_a ::= \Ewhile{x_1 < N_a}{\Eassign{x}{x+1}}$ is an infinite set of loops, where    $L(W_a) \subseteq L(W)$, replace it with the set of loops $W_a' ::= \Ewhile{x_1 < N_a'}{\Eassign{x}{x+1}}$,
    % At every application of $\mathsf{GrmDisj}$ that partitions a subset of $W$, 
    % we keep the finite partition 
% \twrchanged{
%     of
% }
%     elements the same and change each 
%     infinite set of loops of the form
    % $W_a ::= \Ewhile{x_1 < N_a}{\Eassign{x}{x+1}}$ 
    % to the set of loops $W_a' ::= \Ewhile{x_1 < N_a'}{\Eassign{x}{x+1}}$, 
    where $N_a'$ is an inductively defined set of expressions
    for which $E_{N_a'} = E_{N_a} \cup [\max(\min(E_{N_a}), n), \infty)$. 
    Because every finite set of loops remains the same and every judgment 
    proven over infinite sets of loops ($W_a$) by the $\mathsf{While}$ rule can be proven over the modified sets ($W_a'$) by assumption, 
    the proof of 
    $\emptyset \vdash \utripleag{\state[x_1] = 0}{\Ewhile{x_1 < N'}{\Eassign{x_1}{x_1+1}}}{\state[x_1] \in 2\N}$ 
    goes through.
    But we have assumed that our proof system is sound, so no proof of the false judgment $\emptyset \vdash \utripleag{\state[x_1] = 0}{\Ewhile{x_1 < N'}{\Eassign{x_1}{x_1+1}}}{\state[x_1] \in 2\N}$ can exist. We have arrived at a contradiction.
    We conclude that any proof of $\emptyset \vdash \utripleag{\state[x_1] = 0}{W}{\state[x_1] \in 2\N}$ 
    must contain an application of the $\mathsf{While}$ rule
    that meets the conditions given in the statement of the lemma.
\end{proof}}
% \twrchanged{
%     If this condition on $\Gamma \vdash \utripleag{P}{\widetilde{W}}{Q}$ did not hold,
% }
%     then we could copy the proof of 
%     $\emptyset \vdash \utripleag{\state[x_1] = 0}{W}
%     {\state[x_1] \in 2\N}$ and
%     replace $N ::= 0 \mid N+2$ 
%     with the inductively definable set 
%     $N'' ::= N \cup N'$.
% \twr{Again, you are mixing syntax (i.e., $N$) with semantics (i.e., $N'$ is referred to as an inductively definable set [presumably of numbers], and $[n,\infty)$ represents a set of numbers.
% Because $N$ is a nonterminal, you should not use $N'$ for the set of numbers you are defining.
% I'll use $\mathcal{N}$ instead.  You are trying to define
% \[
%   \mathcal{N} = \{ \sem{m} \mid m \in L(N) \} \cup [n,\infty).
% \]
% }
%     where $n$ is the maximum element 
%     of the guard bounds of the finitely many finite sets of loops 
%     to which the While-rule is applied. 
    % The result would be
    % a proof of $\emptyset \vdash \utripleag{\state[x_1] = 0}
    % {\Ewhile{x_1 < N''}{\Eassign{x_1}{x_1+1}}}{\state[x_1] \in 2\N}$, which is false,
    % so the proof
    % of $\emptyset \vdash \utripleag{\state[x_1] = 0}{W}
    % {\state[x_1] \in 2\N}$
    % must contain an application of the While-rule as described.

    Now consider the application of the $\mathsf{While}$ rule described in \Cref{lem:inf_while_app}. Call $\widetilde{B} ::= x_1 < \widetilde{N}$, and call $\widetilde{B}' ::= x_1 < \widetilde{N}'$. Both sets are inductively defined, and $\semsingset{\widetilde{B}} = \semsingset{\widetilde{B}'} = \lambda P. \{\state\subs{\Et}{b_t} \mid \state \in P\} \cup \{\state\subs{\Ef}{b_t} \mid \state \in P \land \state[x_1] \geq \widetilde{n}\}$. Thus, any hypothesis we write about $\widetilde{B}$ is also true of $\widetilde{B}'$! In particular, if our logic is relatively complete, then we can prove the conclusion $\Gamma \vdash \utripleag{P}{\widetilde{W}'}{Q}$. But this conclusion is false from the statement of \Cref{lem:inf_while_app}!
% \twr{I don't see why the conclusion is false by \emph{assumption}.
% The conclusion is presumably false because of your assumptions about $\Gamma \vdash \utripleag{P}{\widetilde{W}}{Q}$, but you made no assumption about $\Gamma \vdash \utripleag{P}{\Ewhile{x_1 < N'}{\Eassign{x}{x+1}}}{Q}$ itself.}
    So, the logic cannot be both sound and relatively complete.}

\end{proof}

\vsyelwhilerule*

\begin{proof}
    We can adapt $f_{while}$ -- the function which determines $\semweaksetyel{\Ewhile{B}{S}}$ from $\semweaksetyel{B}$ and $\semweaksetyel{S}$ in the proof of \Cref{thm:weak_compositional_yel} -- to construct the following inference rule:
    \[
    \infer[\mathsf{While_{\yellow}}]{\Gamma \vdash \utriplevs{P}{(\Ewhile{B}{S})}{Q}}{\Gamma \vdash \utriplevs{x=z}{B}{Q_B} \quad \Gamma \vdash \utriplevs{x=z}{S}{Q_S}}
    \]
    The postcondition $Q$ is defined as:
    {\small
    \begin{align*}
        Q &= \exists t_1, ..., t_n, y, y' \in \State^*. (x' = [t_{1,1}, t_{2,1}, \cdots, t_{n, 1}] \rightarrow P[x'/x])\\
        & \bigwedge (Q_B[y/x] \land z = t_1{\pend}t_2{\pend}\cdots{\pend}t_n \land (\bigwedge\limits_{0 < i \leq n} (\bigwedge\limits_{j < \abs{t_i}} y_{(\sum\limits_{k < i} \abs{t_k}) + j}[b_t] = \Et) \land y_{(\sum\limits_{k \leq i} \abs{t_k})} = \Ef)) \\
        & \bigwedge (Q_S[y'/x] \land z = t_1{\pend}t_2{\pend}\cdots{\pend}t_n \land (\bigwedge\limits_{0 < i \leq n;\,\,\,0 < j < \abs{t_i}} y'_{(\sum\limits_{k < i} \abs{t_k}) + j} = z_{(\sum\limits_{k < i} \abs{t_k}) + j + 1})) \\
        & \bigwedge x = [t_{1,\abs{t_1}}, t_{2,\abs{t_2}}, \cdots, t_{n, \abs{t_n}}]\\
        % Q &= (\forall y_1. y_1 \in P_B \leftrightarrow y_1 = [\state_1, \state_2, \cdots]) \\&\land (\forall y_2. y_2 \in P_S \leftrightarrow y_2 = [\state_1, \state_2, \cdots]) \\&\land \exists v^b \in Q_B, v^s \in Q_S, p \in P. \forall i. x[i] = \sem{\Ewhile{v^b}{v^s}}(p[i]).
    \end{align*}
    }

    In the definition of $Q$, ${\pend}$ denotes concatenation of finite vectors. 
    $Q$ expresses that a vector-state $x$ is a possible outcome if there are $n$ traces 
    ($t_1, \cdots, t_n$) whose start states satisfy the precondition $P$ 
    (1) so that each trace is realizable by a guard (2) and a body (3) and the traces end in final states given by $x$ (4).
    % This construction may also be viewed as a generalization of the invariant constructed 
    % when proving the relative completeness of standard Hoare logic~\cite{winskel}, 
    % for vector-states and sets of programs.
    % $\sigma_j$ is a computable enumeration of states that covers every state over 
    % $\varfun(B) \cup \varfun(S)$. 
    % Also, the while loop in the conclusion's postcondition $Q$ queries $v^b$ and $v^s$ 
    % to simulate running the loop. 
    % That the semantics of programs can be encoded in first-order logic in this way is 
    % a well-known result \cite{winskel}.
    % \sn{I have more explicit proofs of soundness and completeness commented out below, but I don't think they add value.}
    Soundness and relative completeness of this rule follow immediately from the proof of Theorem~\ref{thm:weak_compositional_yel}. 
    % Our While-rule above follows exactly the implied construction.

    % Given sound premises in $Q_B$ and $Q_S$, the conclusion $\utriple{P}{\Ewhile{B}{S}}{Q}$ holds true. If $x \in \semstrongset{\Ewhile{B}{S}}(P)$, then for some $b \in B$, $s \in S$, and $p \in P$, $x = \semstrongset{\Ewhile{b}{s}}(p)$. Since $\sem{b} \in Q_B$ and $\sem{s} \in Q_S$ by hypothesis, $x \in Q$. So $\semstrongset{\Ewhile{B}{S}}(P) \subseteq Q$.  

    % Moreover, for any $P$ and $\Ewhile{B}{S}$, $Q$ is precisely $\semstrongset{\Ewhile{B}{S}}(P)$ when $Q_B$ and $Q_S$ are precise. Specifically, if $Q_B = \semstrongset{B}$ and $Q_S = \semstrongset{S}$, then a vector-state $x$ satisfies $Q$ if and only if there is a loop that yields $x$ when run on a vector-state in $P$ whose guard and body have semantics derivable from $B$ and $S$. This is precisely what it means to say $x \in \semstrongset{\Ewhile{B}{S}}(P)$. Thus $Q = \semstrongset{\Ewhile{B}{S}}(P)$. For any $Q'$ so that $\semstrongset{\Ewhile{B}{S}}(P) \subseteq Q'$, we have $Q \rightarrow Q'$. Thus $\utriple{P}{\Ewhile{B}{S}}{Q}$ implies $\utriple{P}{\Ewhile{B}{S}}{Q'}$.

\end{proof}

\vsgrnwhilerule*

\begin{proof}
    The rule is given below. It differs from the rule $\mathsf{While_{\yellow}}$ in that it
    introduces a variable $i \in \N$ to simulate making countably many queries to $\semweaksetgrn{S}$, and it uses a shortcut to check for admissible guards by collecting the semantics of all guards in $B$:
    $$\infer[\mathsf{While_{\green}}]{\Gamma \vdash 
    \utriplevs{P}{\Ewhile{B}{S}}{Q}
    }{\Gamma \vdash \utriplevs{P_B}{B}{Q_B} \quad \Gamma \vdash \utriplevs{P_S(i)}{S}{Q_S(i)}}$$
    The postcondition $Q$ is defined as:
    \begin{align*}
        Q &= (\forall y. y \in P_B \leftrightarrow y = [\state_1, \state_2, \cdots]) \\
        &\land x \in \Red(\{x' \mid \exists i \in \N, [\state'_1, \cdots] \in VS_{vars}, v^b \in Q_B, p \in P.\\
        &\quad(\forall y'. (y' \in P_S(i) \leftrightarrow y' = [\state'_1, \state'_2, \cdots]) \land Z)\}).
        % \forall i. x[i] = \sem{\Ewhile{\sem{b}}{\sem{s}}}(p[i]))
    \end{align*}

    $Z$ is shorthand for a formula that indicates that one of the countably many computable vectors 
    $[\state'_1, \cdots]$ gives a concatenation of traces over the input vector $p$ that results in $v^s \in Q_S$, a concatenation of output traces, that yields $x'$.\footnote{
      We refrain from explicitly constructing $Z$ here because the construction is very complex and tedious.
    }
    The structure of the rule $\mathsf{While_{\green}}$
    follows exactly the proof of compositionality of \green vector-state semantics (Theorem~\ref{thm:weak_inductive}).
    Intuitively, we capture $\semweaksetgrn{S}(v)$ for all the countably many $v \in \vsvars$ in $Q_s(i)$, and (as in Theorem~\ref{thm:weak_inductive}) we evaluate $\semweaksetgrn{S}$ on finite and infinite traces to determine $\semweaksetgrn{\Ewhile{B}{S}}(p)$. 
    
    %As in the proof of compositionality, $\state_j$ is a computable enumeration of $\State$. 
    % 
    %We do not write $Z$ out explicitly because its form would be tedious and unenlightening. We make no remarks about the assertion language necessary to support such a predicate.

    The soundness and relative completeness of this rule follow from the proof of compositionality 
    of $\semweakset{\cdot}$ (\Cref{thm:weak_inductive}).     

    % Note that, for all necessary predicates to be expressible, we do not strictly require a first-order theory over infinite arrays because all the vectors we use are computable. We can encode our vectors as programs, though this would be cumbersome and circular.
\end{proof}

\subsection{Proofs of Additional Results}
\label{app:small-proofs}
Here, we provide a few short proofs elided from the main text.

% \singlessthanaware*

% \begin{proof}
% To see that $\semfullset{\cdot} \succeq\semsingsetgrn{\cdot}$,
% simply take \snchanged{$f(\semfullset{\progset}) = \bigcup_{\semsingproggrn{\prog} \in \semfullset{\progset}} \semsingproggrn{\prog}$}.

% To see that $\semfullset{\cdot} \not \preceq \semsingsetgrn{\cdot}$, recall the two sets 
% of statements $S_1$ $= \{ \Eifthenelse{x=0}{\Eassign{x}{1}}{\Eassign{x}{2}}$, $\Eifthenelse{x=0}{\Eassign{x}{2}}{\Eassign{x}{1}} \}$ and $S_2$ $= \{ \Eassign{x}{2}$, $\Eassign{x}{1} \}$.
% $\semsingsetgrn{\cdot}(S_1) = \semsingsetgrn{\cdot}(S_2)$, but because 
% $\semfullset{S_1} \neq \semfullset{S_2}$, there cannot exist $f$
% where
% $f_T(\semsingsetgrn{\progset}) = \semfullset{\progset}$ for all $\progset$.
% \end{proof}

\yellessthangrn*

\begin{proof}
To see that $\semsingsetgrn{\cdot} \succeq\semsingsetyel{\cdot}$,
define $g : 2^\DState \rightarrow 2^\State$ as $g(X) = \begin{cases}
    X & \divg \notin X\\
    X \setminus \{\divg\} & \divg \in X
\end{cases}$. Then take $f(\semsingsetgrn{C}) = g \circ \semsingsetgrn{C}$.

To see that $\semsingsetgrn{\cdot} \not \preceq \semsingsetyel{\cdot}$, consider the set 
of programs $\progset = \{\Ewhile{\Et}{\Eassign{x}{x}}\}$. Then
$\semsingsetyel{\progset} = \semsingsetyel{\emptyset}$ because $\semsingsetyel{\progset}(X) = \semsingsetyel{\emptyset}(X) = \emptyset$. However, 
$\semsingsetgrn{\progset} \neq \semsingsetgrn{\emptyset}$ because $\semsingsetgrn{\progset}(\state) = \{\divg\}$ and $\semsingsetgrn{\emptyset}(X) = \emptyset$. Thus, there cannot exist $f$
where
$f_T(\semsingsetgrn{\progset}) = \semfullset{\progset}$ for all $\progset$.
\end{proof}

\begin{theorem}[$\semweaksetgrnnc{\cdot}$ is Noncompositional]\label{thm:nc_noncomp}
    The semantics $\semweaksetgrnnc{\cdot}$ introduced in \Cref{sec:vs_grn_naive} does not admit a compositional characterization.
\end{theorem}

\begin{proof}
    Recall the $\semweaksetgrnnc{\cdot}$ defined in \Cref{sec:vs_grn_naive}. Define $s_1$ as the program $\Eassign{x_1}{1}$ and $s_2$ as the program $\Eassign{x_1}{2}$. Define $s_1'$ as the program $(\Ewhile{x_1 = 1}{skip});\Eassign{x_1}{1}$ and $s_2'$ as the program $(\Ewhile{x_1 = 1}{skip});\Eassign{x_1}{2}$. Then define $S = \{s_1, s_2, s_1', s_2'\}$ and $S' = \{s_1, s_1', s_2'\}$.

    One can observe $\semweaksetgrnnc{S} = \semweaksetgrnnc{S'}$. Let $v$ be a vector-state. If all programs in $S$ converge on $v$, then $\semweaksetgrnnc{S}\{v\} = \semweaksetgrnnc{S'}(\{v\} = \{[\statelist{\stateentry{x_1}{1}}, \cdots ], [\statelist{\stateentry{x_1}{2}}, \cdots] \}$. Otherwise, $v$ contains a state where $x_1 = 1$, so $\semweaksetgrnnc{S}\{v\} = \semweaksetgrnnc{S'}(\{v\} = \{\divg, [\statelist{\stateentry{x_1}{1}}, \cdots ], [\statelist{\stateentry{x_1}{2}}, \cdots]\}$.

    Because $\semweaksetgrnnc{\progset}(V) = \bigcup_{v \in V} \semweaksetgrnnc{\progset}(\{v\})$, $\semweaksetgrnnc{S} = \semweaksetgrnnc{S'}$.

    Now consider the set of loops $W_S = \Eassign{x_1}{0}; \Eassign{x_2}{0}; \Ewhile{x_2 < 5}{x_2++; S}$. This loop runs $5$ times. Note that, when $s_1 \in S$ is chosen, the loop $\Ewhile{x_2 < 5}{x_2++; s_1}$ diverges because the first iteration sets $x_1=1$, and $s_1$ diverges on the second iteration. However, the loop $\Ewhile{x_2 < 5}{x_2++; s_2}$ converges with $x_1 = 2$ because the state $\statelist{\stateentry{x_1}{1}}$ is never encountered. Thus, $\divg \in \semweaksetgrnnc{W_S}(\{v\}) \setminus \semweaksetgrnnc{W_{S'}}(\{v\})$ for all vector-states $v$.

    Then despite $\semweaksetgrnnc{S} = \semweaksetgrnnc{S'}$, we see $\semweaksetgrnnc{W_S} \neq \semweaksetgrnnc{W_{S'}}$. Thus, $\semweaksetgrnnc{\cdot}$ cannot be compositional.  
\end{proof}

\section{Granularity Induces a Lattice}
\label{app:lattice}
In this appendix, we prove that the granularity preorder induces a lattice. Most results are obvious if one imagines granularity as a measure of the fineness of the partition $(\abstrs{A})^{-1}$.

\begin{definition} [Granularity]
    Let a semantics $\abstrs{A}$ be given. The granularity of $\abstrs{A}$ is defined as $[\abstrs{A}] = \{\abstrs{B} \mid \abstrs{A} \preceq \abstrs{B} \preceq \abstrs{A}\}$, the set of semantics with the same granularity as $\abstrs{A}$.
\end{definition}

Note that the granularity preorder induces a partial order ($\preceq$) over the set of granularities.

\begin{theorem} [Granularity Partial Order is a Complete Lattice]
    The set of granularities equipped with the partial order $\preceq$ is a complete lattice.
\end{theorem}

\begin{proof} We show that upper and lower bounds always exist.
    \begin{itemize}
        \item Given an index set $\mathcal{A}$ and granularities $[\abstrs{A_\alpha}]$ for each $\alpha \in \mathcal{A}$, define $\abstrs{\mathcal{A}_\lor}$ so that, for each $\progset$, $\abstrsapp{\mathcal{A}_\lor}{\progset} = \times_{\alpha \in \mathcal{A}} \abstrsapp{A_\alpha}{\progset}$.
        Clearly, $[\abstrs{\mathcal{A}_\lor}]$ is above each $[\abstrs{A_\alpha}]$. Moreover, any semantics $\abstrs{B}$ which determines every $\abstrs{A_\alpha}$ also determines $\times_{\alpha \in \mathcal{A}} \abstrs{A_\alpha}$ and thus $\abstrs{\mathcal{A}_\lor}$. So, $[\abstrs{\mathcal{A}_\lor}]$ is the least upper bound of $\{[\abstrs{A_\alpha}] \mid \alpha \in \mathcal{A}\}$.

        \item Given an index set $\mathcal{A}$ and granularities $[\abstrs{A_\alpha}]$ for each $\alpha \in \mathcal{A}$, define a relation $R_{\mathcal{A}_\land}$ so that, for all $\progset_1$ and $\progset_2$, $\progset_1 R_{\mathcal{A}_\land} \progset_2$ if and only if $\exists \alpha \in \mathcal{A}. \abstrsapp{A_\alpha}{\progset_1} = \abstrsapp{A_\alpha}{\progset_2}$. Define $\equiv_{\mathcal{A}_\land}$ as the transitive closure of $R_{\mathcal{A}_\land}$. Now, define the semantics $\abstrs{\mathcal{A}_\land}$ so that $\abstrsapp{\mathcal{A}_\land}{\progset} = [\progset]_{\equiv_{\mathcal{A}_\land}}$.

        Then $[\abstrs{\mathcal{A}_\land}]$ is below every $[\abstrs{A_\alpha}]$ because $\abstrsapp{A_\alpha}{\progset_1} = \abstrsapp{A_\alpha}{\progset_2}$ implies $\abstrsapp{\mathcal{A}_\land}{\progset_1} = \abstrsapp{\mathcal{A}_\land}{\progset_2}$. 
        
        To prove $[\abstrs{\mathcal{A}_\land}]$ is a greatest lower bound, consider some $[\abstrs{B}]$ below every $[\abstrs{A_\alpha}]$. Define $\equiv_B$ as the equivalence relation for which $\progset_1 \equiv_B \progset_2$ if and only if $\abstrsapp{B}{\progset_1} = \abstrsapp{B}{\progset_2}$. For each $\alpha$, whenever $\abstrsapp{A_\alpha}{\progset_1} = \abstrsapp{A_\alpha}{\progset_2}$, it must be the case that $\abstrsapp{B}{\progset_1} = \abstrsapp{B}{\progset_2}$ (and hence $\progset_1 \equiv_B \progset_2$) since $[\abstrs{B}]$ is below $[\abstrs{A_\alpha}]$. Then $R_{\mathcal{A}_\land} \subseteq \equiv_B$ as relations. Since $\equiv_B$ is transitively closed, $\equiv_{\mathcal{A}_\land} \subseteq \equiv_B$. Then, for all $\progset_1$ and $\progset_2$, if $\abstrsapp{\mathcal{A}_\land}{\progset_1} = \abstrsapp{\mathcal{A}_\land}{\progset_2}$ then $\abstrsapp{B}{\progset_1} = \abstrsapp{B}{\progset_2}$. So, $\abstrsapp{B}{\progset}$ is determined by $\abstrsapp{\mathcal{A}_\land}{\progset}$ for all $\progset$. Thus, $[\abstrs{B}] \preceq [\abstrs{\mathcal{A}_\land}]$.
    \end{itemize}
\end{proof}
Note also that:
\begin{itemize}
    \item  The minimal granularity is $[\abstrs{\bot}]$ where $\abstrsapp{\bot}{\progset} = \emptyset$ for all $\progset$. Given any semantics $\abstrs{A}$, we can always determine $\abstrsapp{\bot}{\progset}$ from $\abstrsapp{A}{\progset}$ because $\abstrs{\bot}$ is constant.

    \item The maximal granularity is $[\abstrs{\top}]$ where $\abstrsapp{\top}{\progset} = \progset$ for all $\progset$. Given any semantics $\abstrs{A}$, we can always determine $\abstrsapp{A}{\progset}$ from $\abstrsapp{\top}{\progset}$ by applying $\abstrs{A}$ to $\abstrsapp{\top}{\progset} = \progset$.
\end{itemize}
\section{Verfiying Sets of Programs by Analyzing an Interpreter}
\label{app:interpreter}

In \Cref{subsec:interpreter} we explore Hoare-style analysis of an interpreter over a set of input programs $\progset$ as a technique to analyze properties of $\progset$. On the basis of our semantic theory, we assert that, for any correct recursive implementation of $\itrp: \Prog \times \State \rightarrow \State$, there exist Cartesian triples  of the form $\triple{t \in \progset \land x \in P}{y = \itrp(t,x)}{y \in Q}$ which cannot be proven from a Hoare-style analysis of $\itrp$ using only Cartesian triples. Moreover, we claim the need for non-Cartesian triples is alleviated by lifting the signature to $\itrp_v: \Prog \times \State \rightarrow \State$. In this section, we prove both results.

\paragraph{Recursive Interpreters}
First, we define what is meant by a recursive interpreter. A recursive interpreter $\itrp(t,x)$ is one which the only accesses to $\mathtt{t}$ or its subprograms are:
\begin{itemize}
    \item On the first line of the implementation, a match statement on the top-level program constructor of $\mathtt{t}$.
    \item Recursively calling $\itrp(\prog, ?)$ where $\prog$ is a program constructed from $\mathtt{t}$ or its available subprograms  (e.g., if $t = \Ewhile{b}{s}$ is determined by a match, then the call $\itrp(\Eifthenelse{b}{s;t}{\Eskip}, x)$ would be allowed).
    % \snchanged{such that each subprogram appears at most once} (i.e., if $t = \Ewhile{b}{s}$ is determined by a match, then $\itrp(\Eifthenelse{b}{s}{skip}, x)$ would be allowed, but $\itrp(\Eifthenelse{b}{s;t}{skip}, x)$ would not be allowed because $b$ and $s$ appear twice each).
\end{itemize}
Examples include \Cref{fig:interpreter} and \Cref{fig:interpreter2}.
% The restriction that each subprogram appears at most once in each recursive call of $\itrp$ may seem strange, but it is necessary to preserve the intention of using Cartesian preconditions. If we do not include this requirement, we could copy our program $t$ to the program $t;t$, call 

\paragraph{Cartesian Triples}
% outline
% Goal -- to have reasoning not require reasoning about applying progs as functions outside what interpret does. (+ and - examples)
% To do this -- no semantic constraints in C. Essentially means C is syntactic combo of subprogs of t.
% -- no linking between C and P by subprogs of t. Fine to link to other subprogs. That's just defining a set.
% Definition
% A simple analysis is one that only uses these preconditions right before recursive calls to itrp
% Consider the example. Non-Cartesian-ness hid at level 1 but appeared at level 2!
% So how far do we unroll to check for this? Formally, a proof uses only Cartesian preconditions if it only uses Cartesian preconditions for all unrolling depths.

% IR
If a Hoare-style analysis of an interpreter is to be tractable, we should avoid reasoning about the relationship between program syntax and semantics inside our preconditions. After all, this is the relationship we are performing the analysis to determine! Thus, we do not want a precondition to a recursive call to $\itrp(t', x')$ to look like \Cref{eq:interpret_precondition_two} or the examples below (where $S$ is an inductively defined set of programs):
\begin{align*}
    (t' \in S \land \sem{t}\statelist{\stateentry{x_1}{0}}[x_1] = 0) \land (x' \in \State)\\
    (t' \in S) \land (x' \in \State \land \sem{t}\statelist{x'}[x_1] = 0)
\end{align*}

In these examples, the predicate has to track both the syntax of $s$ and the semantics of $s$. In other words, the predicate is already reasoning about interpretation! To forbid such predicates in our Hoare-style analysis, we formalize the idea of a ``Cartesian'' precondition so that a proof using only Cartesian preconditions before recursive calls to interpret does not relate syntax and semantics explicitly:
\begin{definition}[Cartesian Precondition]
    Suppose one is conducting a Hoare-style analysis of $\itrp(t,x)$ with the initial condition $t \in \progset \land x \in X$ for some inductively defined set of programs $\progset$ and $X \subseteq \State$. A Cartesian precondition of a recursive invocation $\itrp(t', x')$ is a predicate of the form 
    \begin{equation}
        t' = f(t_1, \cdots, t_j) \land t = f'(t_1, \cdots, t_n) \in \progset \land P(x', x, t_{j+1}, \cdots, t_n) \land R(a_1, \cdots, a_m, x, t_{j+1}, \cdots, t_n)
    \end{equation} where $t_1, \cdots, t_n$ are interpreter program variables denoting subprograms of $t$, $f$ and $f'$ are compositions of program constructors, $a_1, \cdots, a_m$ are program variables of the interpreter, and $P$ and $R$ are predicates.
\end{definition}
Examples of Cartesian preconditions include:
\begin{align*}
    t' &= \Ewhile{b}{\Eassign{x_1}{1};s} \land t = \Ewhile{b}{s} \in \progset \land x' = x \land a_1=5\\
    t' &= \Eifthenelse{b}{s_1;\Ewhile{b}{s_1}}{\Eskip} \land t = \Ewhile{b}{s_1};s_2 \in \progset 
    \land x' = \sem{s_2}(x) \land s_2=\Eskip
\end{align*}

Note that a Cartesian precondition of the top-level $\itrp(t, x)$ is simply a predicate of the form $t \in \progset \land x \in X$ for some inductively defined set of programs $\progset$ and $X \subseteq \State$ (because no other program variables exists).

Next, we should clarify what it means for a Hoare-style analysis of $\itrp$ to use only Cartesian triples. Consider a Hoare-style proof for \Cref{fig:interpreter2}, where $W = \Ewhile{B}{S}$ is a set of loops and $P$ and $Q$ are sets of states.
\begin{example}[Singly Unrolled Hoare Analysis]
    \begin{align*}
        \{t \in W \land x \in P\}\\
        \text{match t with ... case While b do s:}\\
        \{b \in B \land s \in S \land x \in P\}\\
        \itrp(\Eifthenelse{b}{s;t}{\Eskip}, x)\\
        \{\textrm{ret} \in Q\}
    \end{align*}
\end{example}
 Here, it seems that we do not need for non-Cartesian preconditions to have a precise proof! However, if we unroll the analysis of the recursive call $\itrp(\Eifthenelse{b}{s;t}{\Eskip}, x)$, we identify non-Cartesian preconditions:
\begin{example} [Doubly Unrolled Hoare Analysis] Note how, in the call $\itrp(s1, x)$, the set of (syntactic) guards $b$ that can appear in $s_1$ is restricted by the guards' (semantic) output on state $x$.
    \begin{align*}
        \intertext{$\{t \in W \land x \in P\}$}
        &\qquad\text{match t with ... case While b do s}
        \intertext{$\{b \in B \land s \in S \land x \in P\}$}
        &\qquad\itrp(\Eifthenelse{b}{s;t}{\Eskip}, x)
        \intertext{$\{t' = \Eifthenelse{b}{s;t}{\Eskip} \land b \in B \land s \in S \land t = \Ewhile{b}{s} \land x \in P\}$}
        &\qquad\quad \text{match t' with ... case if b then s1 else s2}
        \intertext{$\{b \in B \land s \in S \land t = \Ewhile{b}{s} \land s1 = s;t \land s2=\Eskip \land x \in P\}$}
        &\qquad\qquad b_1 = \itrp(b, x)
        \intertext{$\{b \in B \land s \in S \land t = \Ewhile{b}{s} \land s1 = s;t \land s2=\Eskip \land x \in P \land b_1 = \sem{b}(x)\}$}
        &\qquad\qquad \textrm{if } (b_1.b_t)
        \intertext{$\{{\color{red}b} \in B \land s \in S \land t = \Ewhile{{\color{red}b}}{s} \land s1 = s;t \land s2=\Eskip \land x \in P \land {\color{red}\mathsf{true} = \sem{b}(x)}\}$}
        &\qquad\qquad\quad \itrp(s1, x)\\
        &\qquad\qquad\quad\cdots
        \intertext{$\{\textrm{ret} \in Q\}$}
    \end{align*}
\end{example}
Thus, the non-Cartesian preconditions are still necessary, but they only appear when unrolling the recursive invocation of $\itrp$. In our first ``Cartesian'' proof, we simply offloaded the difficulty to whatever technique analyzes our recursive calls (e.g., the function summary in recursive Hoare-logic~\cite{nipkow}).

Thus, when we say that a Hoare-style analysis of $\itrp(t,x)$ uses only Cartesian triples, we mean that, no matter how much recursive applications of $\itrp$ are unrolled as above, the precondition preceding a recursive call $\itrp(t', x')$ will always be Cartesian. We also demand that, for every triple that is not over a recursive call to $\itrp$, the triple holds if we replace $\progset$ with any other set of programs $\progset'$ (i.e., we disallow triples $\triple{t \in \progset \land x \in X}{\Eskip}{t \in \progset \land x \in X \land \sem{t}(x) \in Q}$ that use the form of $\progset$ to extract semantic properties at the logical level).

\begin{figure}
  \vspace{-2mm}
  \begin{lstlisting}[linewidth=\linewidth]
State interpret (t: Prog, x: State) {
  match t with 
    ...
    case While b do s: {
      return interpret((*@if@*) b (*@then@*) s;t (*@else@*) skip, x);
    }
    case (*@if@*) b (*@then@*) s1 (*@else@*) s2: {
      State b1 = interpret(b,x);
      if (b1.b_t) {
        return interpret(s1, x);
      }
      else {
        return interpret(s2, x);
      }
    }
}
  \end{lstlisting}
  \vspace{-2mm}
  \caption{We give an alternate implementation of $\itrp: \Prog \times \State \rightarrow \State$ that suppresses non-Cartesian analysis to the second unrolling of $\itrp$.
  }
  \label{fig:interpreter2}
\end{figure}

\paragraph{Incompleteness of Cartesian Triples}
As stated in \Cref{subsec:interpreter}, we will show that, for any correct recursive implementation of $\itrp: \Prog \times \State \rightarrow \State$, there are Cartesian triples that cannot be proven via a Hoare-style analysis of $\itrp$ using only Cartesian triples. 

To prove this fact, we first define a semantics induced by the true Cartesian triples of $\itrp$. Understanding this semantics will be key to proving the necessity of non-Cartesian triples.

\begin{definition}[Interpreter Semantics $\abstrs{\itrp_a}$]
    Recall from \Cref{subsec:interpreter} that our analysis is based on a semantics that assigns to each $\progset$ the set of true Cartesian triples about it, eliding $\progset$ in the triple:
    {\small
\begin{equation*}
    \abstrsapp{\itrp_a}{\progset} = \set{\triple{\mathtt{t} \in\ ? \wedge P(\mathtt{x})}{y = \itrp(t, x)}{Q} \mid \text{the triple is true when $\progset$ is substituted for $?$}}
\end{equation*}
}
\end{definition}

\begin{proposition}[$\abstrs{\itrp_a} \equiv \semsingsetyel{\cdot}$ Are Noncompositional]\label{prop:itrp_noncomp}
    The semantics $\abstrs{\itrp_a}$ and $\semsingsetyel{\cdot}$ are equivalently granular. Moreover, both semantics are noncompositional over while and compositional over all other constructors.
\end{proposition}
\begin{proof}
Because a triple $\triple{\mathtt{t} \in \progset \wedge \mathtt{x} \in P}{y = \itrp(t, x)}{Q}$ is equivalent to $\semsingsetyel{\progset}(P) \subseteq Q$, $\abstrs{\itrp_a} \preceq \semsingsetyel{\cdot}$. Because $\triple{\mathtt{t} \in \progset \wedge \mathtt{x} \in P}{y = \itrp(t, x)}{\semsingsetyel{\progset}(P)}$ is a true triple for every $P$, $\abstrs{\itrp_a} \succeq \semsingsetyel{\cdot}$. So $\abstrs{\itrp_a}$ and $\semsingsetyel{\cdot}$ are equivalently granular. Since compositionality over each constructor is preserved between equivalently granular semantics, the noncompositionality over while of $\semsingsetyel{\cdot}$ implies the noncompositionality over while of $\abstrs{\itrp_a}$. Similarly, the compositionality of $\semsingsetyel{\cdot}$ over all other constructors implies the compositionality of $\abstrs{\itrp_a}$ over all non-while constructors.
\end{proof}

We are now ready to state and prove the necessity of non-Cartesian triples:
\begin{theorem}[Cartesian Reasoning is Incomplete]\label{thm:itrp_noncart}
    Let $\itrp: \Prog \times \State \rightarrow \State$ be a correct recursive interpreter. Then there exists a Cartesian triple $\triple{t \in \progset \land x \in P}{y = \itrp(t, x)}{y \in Q}$ which cannot be proven by any Hoare-style analysis that uses only Cartesian triples.
\end{theorem}

\begin{proof}
    Suppose we have a correct recursive interpreter $\itrp: \Prog \times \State \rightarrow \State$ for which every Cartesian triple can be proven using only Cartesian triples. We will show that this implies compositionality of $\abstrs{\itrp}$ over while, a contradiction of \Cref{prop:itrp_noncomp}.

    Let $B$ and $S$ be given, and let $W = \Ewhile{B}{S}$. We will determine $\abstrsapp{\itrp_a}{W}$ from $\abstrsapp{\itrp_a}{B}$ and $\abstrsapp{\itrp_a}{S}$. Let a Cartesian triple $\triple{t \in ? \land x \in P}{y = \itrp(t, x)}{y \in Q}$ be given. 
    
    If the triple is true, then by hypothesis there is a Hoare-style analysis of $\itrp$ using only Cartesian triples that proves it. 
    Suppose we unroll the recursive calls of $\itrp$ in this proof, stopping only at calls of the interpreter on the guard and/or body of $t$ -- i.e., $\itrp(b, \cdot)$ and $\itrp(s, \cdot)$.
    The resulting unrolled proof, which may be infinite, is such that every recursive invocation of $\itrp$ is either:
    \begin{itemize}
        \item Of the form $\itrp(b, \cdot)$ or $\itrp(s, \cdot)$.
        \item Of the form $\itrp(\prog, \cdot)$ for a constant program $\prog$.
        \item Determined by the results of further recursive calls to $\itrp$.  
    \end{itemize}
    This is easy to justify. If there were a call $\itrp(f(b,s), x')$ with non-constant $f$ that does not recursively invoke interpret, then it knows the behavior of the program $f(b,s)$ without looking at said program! But this is impossible, because every non-constant program constructor can produce functions that differ at a given $x'$. 

    It helps to visualize our Hoare-style analysis as a proof tree. 
    Unrolling a proof as in \Cref{fig:interpreter} would lead to an infinite tree of unrolled calls that the proof explores (we abbreviate $\itrp$ as $\mathtt{itrp}$):
    {\tiny 
        \begin{prooftree}
            \AxiomC{$\triple{...}{b_1 = \mathtt{itrp}(b, x)}{...}$\kern-2.5em}
            \AxiomC{$\triple{...}{r = x}{...}$\kern-2.5em}
            \AxiomC{$\triple{...}{x1 = \mathtt{itrp}(s, x)}{...}$\kern-20em}
            \AxiomC{$\triple{...}{b_1 = \mathtt{itrp}(b, x)}{...}$\kern-2.5em}
            \AxiomC{$\triple{...}{r = x}{...}$\kern-2.5em}
            \AxiomC{$\triple{...}{x1 = \mathtt{itrp}(s, x)}{...}$\kern-2.5em}
             \AxiomC{$\vdots$}
             \UnaryInfC{$\triple{...}{\mathtt{itrp}(\Ewhile{b}{s}, x_1)}{...}$}
            \QuaternaryInfC{$\triple{...}{\text{match t with ... While b do s}}{...}$}
             \UnaryInfC{$\triple{...}{\mathtt{itrp}(\Ewhile{b}{s}, x_1)}{...}$}
            \QuaternaryInfC{$\triple{...}{\text{match t with ... While b do s}}{...}$}
            \UnaryInfC{$\triple{t \in W \land x \in P}{y = \mathtt{itrp}(t, x)}{y \in Q}$}
        \end{prooftree}
    }
    
    To check correctness of the proof tree, we must check that all axioms are correct and that all non-axioms follow from their hypotheses. It is easy to check that every non-axiom triple in the proof tree is implied by its hypotheses -- this requires no knowledge of $W$ except that it is of the form $\Ewhile{B}{S}$. Correctness of the axioms can be checked with only $\abstrsapp{\itrp_a}{B}$ and $\abstrsapp{\itrp_a}{S}$ because every axiom is either \rone a Cartesian triple on $\itrp(\prog, \cdot)$ for a constant $\prog$, which requires no knowledge of $W$, \rtwo a triple about a program that is not a recursive call to $\itrp$ and thus holds if we replace $W$ by any $\progset$, requiring no knowledge of $W$, or \rthree a Cartesian triple on $\itrp(b, \cdot)$ or $\itrp(s, \cdot)$ which is determined by $\abstrsapp{\itrp_a}{B}$ or $\abstrsapp{\itrp_a}{S}$.
    
    Thus, if $\triple{t \in W \land x \in P}{y = \itrp(t, x)}{y \in Q}$ is provable using only Cartesian triples, there exists an unrolled proof of this form whose correctness we can certify. Obviously, if the triple is not true then no correct proof will exist. 

    Now, the procedure is simple. Given $\triple{t \in ? \land x \in P}{y = \itrp(t, x)}{y \in Q}$, $\abstrsapp{\itrp_a}{B}$, and $\abstrsapp{\itrp_a}{S}$, we search for a correct unrolled proof over $\itrp$ of the sort described. If we find one, then $\triple{t \in ? \land x \in P}{y = \itrp(t, x)}{y \in Q} \in \abstrsapp{\itrp_a}{\Ewhile{B}{S}}$. If we do not, then $\triple{t \in ? \land x \in P}{y = \itrp(t, x)}{y \in Q} \notin \abstrsapp{\itrp_a}{\Ewhile{B}{S}}$. In this way, we can determine $\abstrsapp{\itrp_a}{\Ewhile{B}{S}}$ from $\abstrsapp{\itrp_a}{B}$ and $\abstrsapp{\itrp_a}{S}$. So, $\abstrs{\itrp_a}$ is compositional over while, a contradiction of \Cref{prop:itrp_noncomp}. 

\end{proof}
While the inability of Cartesian reasoning to prove all Cartesian triples does not mean that conducting a Hoare-style analysis of an interpreter over an inductively defined set of programs containing loops is impossible, it does mean that a certain amount of reasoning about the semantic meaning of syntactic programs will need to occur at the logical level (i.e., in a non-Cartesian fashion).

\paragraph{Completing $\itrp$ for Cartesian Triples}
% Outline 
% We showed Cartesian incomplete bc ag noncompositional
% If we use compositional semantics, can we get Cartesian complete?
% Yes. Change the signature to vs. Implementation below:
% Cartesian analysis of this interpreter is complete.
% Explain why.

We have seen that Hoare-style analyses of $\itrp: \Prog \times \State \rightarrow \State$ that use only Cartesian triples cannot prove all Cartesian triples. This insufficiency holds because the signature of $\itrp$ induces a noncompositional semantics. A natural question is whether we can extend the signature of $\itrp$ to induce a compositional semantics, and whether Cartesian reasoning over such an extended interpreter might be complete. As it turns out, we can.

Recall that the \yellow vector-state semantics $\semweaksetyel{\cdot}$ is a compositional semantics finer than the \yellow program-agnostic semantics $\semsingsetyel{\cdot}$. Since the \yellow program-agnostic semantics corresponds to the signature $\itrp: \Prog \times \State \rightarrow \State$, perhaps extending the signature to $\itrp_v: \Prog \times \State^* \rightarrow \State^*$ will induce a semantics equivalent to $\semweaksetyel{\cdot}$. As it turns out, this is the case. 

\begin{definition}[Vector-State Interpreter Semantics $\abstrs{\itrp_v}$]
    A correct interpreter $\itrp_v: \Prog \times \State^* \rightarrow \State^*$ induces the semantics:
\begin{equation*}
    \abstrsapp{\itrp_v}{\progset} = \set{\triple{\mathtt{t} \in\ ? \wedge P(\mathtt{x})}{y = \itrp(t, x)}{Q(y)} \mid \text{the triple is true when $? = \progset$}}
\end{equation*}
where $x$ and $y$ are vector-states, and $P$ and $Q$ are predicates over vector-states.
\end{definition}

The correspondence between $\abstrs{\itrp_v}$ and $\semweaksetyel{\cdot}$ is immediate:
\begin{proposition}[$\abstrs{\itrp_v} \equiv \semweaksetyel{\cdot}$ Are Compositional]
    The semantics $\abstrs{\itrp_v}$ and $\semweaksetyel{\cdot}$ are equivalently granular. Moreover, both semantics are compositional.
\end{proposition}
\begin{proof}
Because a triple $\triple{\mathtt{t} \in \progset \wedge \mathtt{x} \in P}{y = \itrp(t, x)}{Q}$ is equivalent to $\semweaksetyel{\progset}(P) \subseteq Q$, $\abstrs{\itrp_v} \preceq \semweaksetyel{\cdot}$. Because $\triple{\mathtt{t} \in \progset \wedge \mathtt{x} \in P}{y = \itrp(t, x)}{\semweaksetyel{\progset}(P)}$ is a true triple for every $P$, $\abstrs{\itrp_v} \succeq \semsingsetyel{\cdot}$. So $\abstrs{\itrp_v}$ and $\semsingsetyel{\cdot}$ are equivalently granular. Since compositionality is preserved between equivalently granular semantics, the compositionality $\semweaksetyel{\cdot}$ implies the compositionality of $\abstrs{\itrp_v}$.
\end{proof}

\begin{figure}\label{fig:interpretervs}
 \begin{lstlisting}[linewidth=\linewidth]
State^* interpret_vs (t: Term, x_1: State, \cdots, x_n: State) {
 match t with 
   ...
   case While b do s: {
     State^n r;
     (State^*)^n executions = (nondet_finite_vec_starting_with(x_1), ..., nondet_finite_vec_starting_with(x_n));
     (State^*)^n guard_executions = (interpret(b, executions_1), ..., interpret(b, executions_n));
     if (\forall j, i. guard_executions[j][i].b_t[:-2] = t && guard_executions[j][i].b_t[-1] = f) {
       (State^*)^n body_executions = (interpret(s, executions_1), ..., interpret(s, executions_n));
       if (\forall j, i. body_executions[j][i] = body_executions[j][i+1]) {
           return (body_executions[1][-1], ..., body_executions[n][-1]);
        }
        else {
            abort;
        }
     }
     else {
       abort;
     }
   }
   ...
}
 \end{lstlisting}
 \caption{Part of an interpreter that interprets $\Ewhile{B}{S}$ 
 for the language described in Figure~\ref{fig:gimp}.
 We assume that $\mathtt{State}$ and $\mathtt{Term}$ that 
 respectively capture the types of states and terms.
 }\label{fig:vs_interpret}
\end{figure}

An implementation of such an interpreter is given in \Cref{fig:interpretervs}. The interpreter's implementation closely follows the proof that $\semweaksetyel{\cdot}$ admits a compositional characterization (\Cref{thm:weak_compositional_yel}). Because the while operator $f_{while}$ in the proof of compositionality checks for the existence of a satisfactory finite vector-state in $\State^*$, our implementation is nondeterministic, choosing a vector-state and returning only if the vector-state is satisfactory. Specifically, for an input $(\Ewhile{b}{s}, [x_1, \cdots, x_n])$, the algorithm nondeterministically chooses $n$ sequences of states that could model the execution of $\Ewhile{b}{s}$ on each $x_j$ (Line 6). Then, the algorithm checks that, for each possible execution, $b$ is true on every state except the last one where the loop converges (Line 7 and 8). If this fails then the executions are incorrect, so we abort. Else, we check whether, for each possible execution, $s$ sends each state in the execution to the next state in the execution (Lines 9 and 10). If this fails then the executions are incorrect, and we abort. If both checks succeed, then the nondeterministically chosen executions do in fact describe how the behavior of the loop on $[x_1, \cdots, x_n]$, and we return the appropriate output.\footnote{Note that the existential nondeterminism can actually be determinized by interleaving computations of each nondeterministic choice and adding a ``number of steps'' parameter to the interpreter to avoid diverging on irrelevant inputs. The construction is straightforward but fairly tedious.}

Over such an interpreter, Cartesian analysis is complete.
\begin{theorem}[Cartesian Reasoning is Complete over $\itrp_v$]\label{thm:itrp_vs_cart}
   Every true Cartesian triple $\triple{t \in \progset \land x \in P}{y = \itrp(t, x)}{y \in Q}$ can be proven by a Hoare-style analysis of $\itrp_v$ that uses only Cartesian triples.
\end{theorem}
\begin{proof}
    The proof is a simple induction argument. Suppose we want to prove the Cartesian triple $\triple{t \in \progset \land P(x)}{y = \itrp_v(t, x)}{y \in Q}$ where $\progset = \op(C_1, \cdots, C_n)$.
    
    If $\itrp_v$ does not invoke $\itrp_v$ recursively, then one can check that the most precise proof uses only Cartesian triples.
    
    Otherwise, after one unrolling we have some recursive invocations of $\itrp$. One observes that the most precise precondition at each invocation is in fact Cartesian, \emph{so long as the subprograms are independent} (e.g., $b$ and $s$ do not depend on one another in $\Ewhile{b}{s}$). For example, in the While case, each call to $\itrp(b, executions_j)$ on Line 7 is captured precisely by the triple
    \begin{align*}
    \{t = \Ewhile{b}{s} \in \Ewhile{B}{S} \land executions_j \in \State* \land executions_j[0]=x_j\}\\{guard\_executions[j] = \itrp_v(b, executions_j)}\\\{guard\_executions[j] \in \semweaksetyel{B}(executions_j)\}
    \end{align*}
    The most precise triples on Line 9 are similarly Cartesian.
    Critically, the recursive calls in our implementation do not entangle subprograms. Thus, we do not make a recursive call like $\itrp(\Eifthenelse{b}{s;t}{\Eskip})$ like in \Cref{fig:interpreter2} which entangles the if-then-else guard with the then branch. Since the analyses of the recursive calls are Cartesian triples of calls that do not entangle subprograms, the proofs over the recursive calls are Cartesian by our inductive hypothesis, so the entire proof of $\triple{t \in \progset \land P(x)}{y = \itrp_v(t, x)}{y \in Q}$ is Cartesian as desired.
\end{proof}

While this interpreter supports complete Cartesian analysis, one pays a price by reasoning over unbounded vector-states and nondeterminism. Again, the ability to reason compositionally requires an increase in expressivity.
\section{\greenCaps Vector-State Semantics}
\label{app:VectorStateSemantics}

Next, we describe a \green vector-state semantics $\semweaksetgrn{\cdot}$, the minimal compositional semantics that is at least as fine as $\semsingsetgrn{\cdot}$. While similar to $\semweaksetyel{\cdot}$, $\semweaksetgrn{\cdot}$ must be able to reason about nontermination. This means that it must be able to explicitly represent divergence, and it must also be able to reason about infinitely many executions of loop bodies. Unlike $\semweaksetyel{\cdot}$, the vector states $\semweaksetgrn{\cdot}$ operates over may be infinite (but still computable) or contain an occurrence of $\divg$ (i.e., $[\statelist{\stateentry{x_1}{1}}, \statelist{\stateentry{x_1}{2}}, \statelist{\stateentry{x_1}{3}}, \cdots]$ and $[\statelist{\stateentry{x_1}{1}}, \statelist{\stateentry{x_1}{2}}, \divg]$ are vector-states we may wish to consider). 

\subsection{\texorpdfstring{$\semweaksetgrn{\cdot}$}{}: Divergence-Aware Vector-State Semantics} \label{app:vec_grn}

The tricky part the construction of $\semweaksetgrn{\cdot}$ is handling divergence---i.e., nontermination.
Observing divergence of carefully crafted loops ($\divg \in \semsingsetgrn{\Ewhile{b}{\cdots \progset \cdots}}$) will never give enough information to determine whether any $\prog \in \progset$ diverges 
on two distinct inputs. 
Similarly, inspecting the content of $\semweaksetgrn{\progset}$ cannot help us conclude whether a single program diverged on two inputs---i.e., we cannot take a program $\prog$ which diverges on $\statelist{\stateentry{x_1}{1}}$ and $\statelist{\stateentry{x_1}{2}}$ and determine that $\semweaksetgrn{\{\prog\}}([\statelist{\stateentry{x_1}{1}}, \statelist{\stateentry{x_1}{2}}]) = [\divg, \divg]$ by querying the program-agnostic semantics of sets of loops containing $\prog$. This issue inspires the first restriction on our vector-state semantics---outputs are truncated after the first occurrence of $\divg$.

Furthermore, consider a program $s_1$ that diverges only when $x_1 = 1$, 
a program $s_2$ diverges only when $x_1 = 2$, 
and a program $s_3$ that never diverges;  otherwise all three behave as 
identity. Then the program-agnostic semantics of the three sets of loops
$W_j$ ($j = 0, 1, 2$) that each run the set of bodies 
$S = \{s_1, s_2, s_3\}$ on the sequence 
$[\statelist{\stateentry{x_1}{0}}, \cdots, \statelist{\stateentry{x_1}{j}}]$ give us enough information to identify that there is a program in $S$ with the behavior $[\statelist{\stateentry{x_1}{0}}, \divg]$ (i.e., because $\semsingsetgrn{W_0}$ always converges $\statelist{\stateentry{x_1}{0}}$ on the first input state but $\semsingsetgrn{W_j}$ may diverge).
However, these $\semsingsetgrn{W_j}$ ($j = 0, 1, 2$) will never give us enough information to see that
$[\statelist{\stateentry{x_1}{0}}, \statelist{\stateentry{x_1}{1}}, \divg]$ is a possible outcome 
because although $\semsingsetgrn{W_1}$ and $\semsingsetgrn{W_2}$ may both diverge, we cannot determine whether there is divergence in $\semsingsetgrn{W_2}$ from loop bodies which do not cause divergence in $\semsingsetgrn{W_1}$. We say that the divergence of $s_1$ \emph{occludes} the later divergence of $s_2$.
One can see this by observing that the three $\semsingset{W_j}$ 
are unchanged if we remove $s_2$.
This occlusion constitutes half of our second restriction---captured 
formally in \Cref{def:reduction}.

A similar phenomenon of occlusion happens when we want to test the behavior of $\progset$ 
on an infinite sequence of states. 
In this case, checking the program-agnostic semantics of various sets of loops
does not allow us to distinguish divergence due to divergence in the body from 
divergence due to the body being run infinitely many times. 
Thus, if a loop $(\Ewhile{\Et}{s;x_1:=x_1+1})$ diverges on a given input, 
we cannot distinguish divergence of $s$ from convergence of $s$ on infinitely many states 
because the loop diverges in both cases.
This second type of occlusion is the other half of our second restriction, captured formally 
in \Cref{def:reduction}.

% \twr{The three restrictions covered in the paragraphs above might suggest to the reader that something is wrong: in defining a semantics for a set of programs, an intuitive notion is that one is trying to capture (somehow) the behaviors of all of the possible programs.
% But if some behaviors are occluded by other behaviors, then a client of the semantics has no chance to make use of the occluded behaviors.
% (Could the behavior of some program $P$ \emph{always} be occluded by one or more other programs? If so, then the client never ``sees'' $P$.)
% }

% \twr{
% This situation sounds fishy: how do we know that occluded behaviors are not a problem for some client of our semantics.
% Perhaps the answer is that we have only one intended client, namely Hoare-style compositional reasoning.
% Perhaps the intuitive explanation has to give up at this point with something like Jinwoo's phrase, ``From the perspective of trying to define a compositional semantics, this feature has the characteristic that it allows the semantics to ‘see’ where the first occurrence of divergence occurs. It turns out this characteristic is a
% minimal requirement to compute the semantics of sets of loops in
% a compositional fashion.''}

Because we are searching for the coarsest semantics satisfying our criteria, the restrictions imposed on $\semweaksetgrn{\cdot}$ naturally limit its expressivity. The cost of the truncation we describe is that we cannot identify when there is a single program in a set that diverges on two or more particular inputs. The cost of our occlusion is that, although we can tell when a set of loops may diverge, we can see only the most immediate cause of divergence. 
It turns out that the ability to see the most immediate cause of divergence is a minimal requirement to compute the semantics of a set of loops in a compositional fashion.
Also, the lost behaviors are not required for a compositional semantics finer than $\semsingsetgrn{\cdot}$.
Those interested in such behaviors should refer to the program-aware semantics for sets of programs defined in \Cref{def:full_sem_set}.

% \sn{I think this next point can wait until the proof sketch later.}
% As a final note, one may wonder how such a semantics could be compositional. How can we determine the semantics of a while loop over a vector if we need vectors to understand behavior on even a single input? The short answer is that, because we only need to understand behavior on the input vector $v$ up to the first diverging input $v_k$, we can concatenate the vectors describing converging behavior on $v_1, \cdots, v_{k-1}$ with vectors capturing diverging behavior on $v_k$. 
% This step
% produces a single vector, albeit a longer one, on which the behavior of the loop body can be captured by $\semweakset{\cdot}$. The details are spelled out further in the proof of \Cref{thm:weak_inductive}.

\subsubsection{Formalization of $\semweaksetgrn{\cdot}$}  \label{sec:weak_vs_formalization}
  Building upon the intuition just described, we now 
  formalize vector-states and $\semweaksetgrn{\cdot}$.
In the following definitions, we use $\State^*$ and $\State^\omega$ to denote the set of finite and infinite sequences over $\State$, respectively.
We use  $\State^*\cdot \{\divg\}$ to denote the set of finite sequences where all but the last element are states, and the last element is $\divg$; we sometimes refer to such sequences as \textit{diverging}.
We call an infinite sequence of states $v = [{\state_1, \state_2, \cdots}]$ \emph{computable} if the function mapping $(i, v) \in \N \times \vars$ to the values of $v[i][v]$ is computable.
We use $\vinf \subseteq \State^\omega$ to denote the set of computable infinite sequences of states.
% \sn{Note: I changed this definition slightly. Before, we required the map $\N \rightarrow (V \rightarrow vals^n)$ for $V \subset \vars$ [$\abs{V} = n$] to be computable. Now, we ask for $\N \times \vars \rightarrow vals$ to be computable which is equivalent but easier to understand.}

\begin{definition} [$\vsvars$] \label{def:vs_vars}     
    An \textit{extended vector state} is an element of the set $$\vsvars = \State^* \cup (\State^* \cdot \{\divg\}) \cup V_{Inf}$$    
    % We call a weak vector state in $\State^* \cdot \{\divg\}$ \textit{diverging}.
\end{definition}
Informally an extended vector state can be \rone a finite sequence of states, \rtwo a finite sequence of states followed by a diverging state, or \rthree an infinite, but still computable sequence of states.

As in \Cref{sec:weak_vs_intuition}, a vector-state 
$v \in \vsvars$ will sometimes \emph{occlude} another vector-state $u \in \vsvars$.
For example, consider $v = [\statelist{\stateentry{x_1}{1}}, \divg]$, and $u = [\statelist{\stateentry{x_1}{1}}, \statelist{\stateentry{x_1}{2}}, \divg]$.
Intuitively, the interpretation of $v$ is that entries further than $v[0]$ are 
\emph{unknown}, because the semantics have determined that values after 
$v[0]$ may (not must!) diverge.
From this interpretation, it follows that $v$ actually corresponds to many concrete vector-states: 
for example, $[\statelist{\stateentry{x_1}{1}}, \statelist{\stateentry{x_1}{2}}, \divg]$, $[\statelist{\stateentry{x_1}{1}}, \statelist{\stateentry{x_1}{2}}, \statelist{\stateentry{x_1}{3}}, \divg]$, and so on, as well as infinite vectors starting with $\statelist{\stateentry{x_1}{1}}$ (e.g., $[\statelist{\stateentry{x_1}{1}}, \statelist{\stateentry{x_1}{2}}, \cdots]$).
One may also say that $[\statelist{\stateentry{x_1}{1}}, \divg]$ can be interpreted as describing a class of vector-states.

We capture this notion of occlusion in Definition~\ref{def:reduction} with the function $\Red$. Given a set of vector-states $V \subseteq \vsvars$, 
$\Red(V)$ collapses this $V$ into the coarsest classes possible, throwing out information we are not ``allowed'' to see. Thus, the elements in $\Red(V)$ are those elements unoccluded in $V$.
% We define the equivalence class such that, e.g., $[x = 1, \divg]$ acts as an equivalence class for 
% \rone all vectors that start with $x = 1$, then may diverge at some point (e.g., $[x = 1, x = 2, \divg]$; 
% \rtwo all vectors that start with $x = 1$, and are \emph{infinite} (but need not necessarily diverge on an element, e.g., $[x = 1, x = 2, \cdots]$). \sn{``These are the two cases which will cause loop divergence.'' seems important to mention}
% Note the equivalence class does not contain \emph{finite} vectors that start with $x = 1$, such as 
% $[x =1, x =2, x = 3]$.
% The reason we define equivalence classes in this way is to obtain a semantics that is both 
% compositional and as coarse-grained as possible; the concept of attributing vector-states to an
% equivalence class is essentially what makes $\semweakset{\cdot}$ `coarse'.

% \snchanged{``With this picture in mind, it is helpful, given a set of vector states $V \subseteq \vsvars$, to have a way to discard (or quotient out) information that weak vector-state semantics is not allowed to see (i.e., causes of divergence of a loop beyond the first observable one, as in \S~\ref{sec:weak_vs_intuition}) \snchanged{-- e.g., reduce $\{[\statelist{x_1 = 1}, \divg], [\statelist{x_1 = 1}, \statelist{x_1 = 2}, \divg], [\statelist{x_1 = 1}, \statelist{x_1 = 2}, \cdots]\}$ to $\{[\statelist{x_1 = 1}, \divg]\}$ because divergence in $[\statelist{x_1 = 1}, \divg]$ occludes the other two states}. 
% A construct $\Red$ to perform this reduction is given below with an example:''}

% \sn{I changed ``divergent prefixes'' to ``occluders''.}
For the purposes of our definition,
we say a finite sequence $[v_1,\cdots,v_k,\divg]$ is an 
occluder
of any sequence $u=[v_1,\cdots,v_k, v_{k+1},\cdots] \in \State\cdot\{\divg\} \cup \vinf$.
For example, $[\statelist{\stateentry{x_1}{1}}, \divg]$ is an 
occluder
of itself and of both $[\statelist{\stateentry{x_1}{1}}, \statelist{\stateentry{x_1}{2}}, \divg]$ and  $[\statelist{\stateentry{x_1}{1}}, \statelist{\stateentry{x_1}{2}}, \statelist{\stateentry{x_1}{3}},\cdots]$, but not of $[\statelist{\stateentry{x_1}{1}}, \statelist{\stateentry{x_1}{4}}, \statelist{\stateentry{x_1}{3}}]$.

\begin{definition} [Reduction of a Set of Vector-States] \label{def:reduction}    
    Given a set of extended vector states $V \subseteq \vsvars$ we define the \textit{reduced set of states} $\Red(V)$ as 
    the set obtained by replacing each sequence $u \in V\cap (\State\cdot\{\divg\} \cup \vinf)$ with its shortest 
    occluder
    in $V$.
\end{definition}
$\Red(V)$ induces an equivalence relation on $V$.
Intuitively, every vector-state in $\State^*$ is its own ``class'', and every vector-state in $\State\cdot\{\divg\} \cup \vinf$ falls into the ``class'' described by its shortest 
occluder in $V$. 

\begin{example}[Reduction]
Consider the following set of vector-states:
  \[
  \begin{split}
  V_1 = \{ &\mathcolor{dgreen}{[\statelist{\stateentry{x_1}{1}}, \divg]}, 
          \mathcolor{dgreen}{[\statelist{\stateentry{x_1}{1}}, \statelist{\stateentry{x_1}{2}}, \divg]}, 
          \mathcolor{cyan}{[\statelist{\stateentry{x_1}{2}}, \divg]} \\
         & \mathcolor{dred}{[\statelist{\stateentry{x_1}{1}}, \statelist{\stateentry{x_1}{2}}]}, 
         \mathcolor{dgreen}{[\statelist{\stateentry{x_1}{1}}, \statelist{\stateentry{x_1}{2}}, \statelist{\stateentry{x_1}{3}}, \cdots]}
         \}
  \end{split}
  \]
  The equivalence classes induced by $\Red(V_1)$ are color-coded.
  $\mathcolor{dgreen}{[\statelist{\stateentry{x_1}{1}}, \divg]}$ is the shortest 
  occluder
  % diverging prefix
  of both $\mathcolor{dgreen}{[\statelist{\stateentry{x_1}{1}}, \statelist{\stateentry{x_1}{2}}, \divg]}$ and $\mathcolor{dgreen}{[\statelist{\stateentry{x_1}{1}}, \statelist{\stateentry{x_1}{2}}, \statelist{\stateentry{x_1}{3}}, \cdots]}$.
  % : the former diverges with a longer prefix while the latter is infinite.
  % $\mathcolor{cyan}{[x = 2, \divg]}$ is in its own equivalence class: there are no vector-states that start with $x = 2$.
  % $\mathcolor{dred}{[x = 1, x = 2]}$ is also in its own equivalence class: it is finite.
  Thus it follows that:
  \[
  \Red(V_1) = \set{\mathcolor{dgreen}{[\statelist{\stateentry{x_1}{1}}, \divg]}, \mathcolor{cyan}{[\statelist{\stateentry{x_1}{2}}, \divg]}, \mathcolor{dred}{[\statelist{\stateentry{x_1}{1}}, \statelist{\stateentry{x_1}{2}}]}}
  \]
\end{example}

We can now give a formal definition of $\semweaksetgrn{\cdot}$.
%The point is illustrated further in Figure~\ref{fig:weak_sem} for weak vector-state semantics.
% \twr{This operation is tantamount to insisting on a deterministic semantics, isn't it?  That is, $\{ [1,\divg], [1,2,\divg] ... \}$ would be a possible behavior of a non-deterministic program.}
% \sn{Not exactly. $\{[1, \divg], [1, 2, 3]\}$ would be a valid outcome.}

\begin{definition} [\greenCaps Vector-State Semantics of a Program] \label{def:weak_sem_prog}
    The \emph{\green vector-state semantics} for programs $\semweakproggrn{\cdot}$ over the signature $\weak(2^\vsvars \rightarrow 2^\vsvars)$ is defined as follows.
    We first define the semantics for singleton vector states.
    For a program $\prog \in T$ and extended vector state $[v_1,\ldots] \in \vsvars$, the semantics 
    $\semweakproggrn{\prog}(\{[v_1,\ldots]\}))$ is defined as follows:
    \begin{enumerate}
        \item If for every $v_i$ the program-agnostic semantics does not diverge---i.e., $\semsingproggrn{\prog}(v_i)\neq~\divg$---then 
        \[\semweakproggrn{\prog}(\{[v_1,\ldots]\}) = \{[\semsingproggrn{\prog}(v_1), \cdots]\}.\]
        \item If $k$ is the smallest index such that $\semsingproggrn{\prog}(v_i)=\,\divg$, then 
        \[\semweakproggrn{\prog}(\{v\}) = \{[\semsingproggrn{\prog}(v[1]), \cdots, \semsingproggrn{\prog}(v[k-1]), \divg]\}.\]
    \end{enumerate}

    Note that the entries of the vectors are technically singleton sets; we unpack them as states.
    
    For a set of vector states $V \subseteq \vsvars$, we define $\semweakproggrn{\prog}(V) = \Red(\bigcup\limits_{u \in V}\semweakproggrn{\prog}(\{u\}))$.
\end{definition}

Intuitively, $\semweakproggrn{\prog}(\{v\})$ (often written as $\semweakproggrn{\prog}(v)$) gives the output of $\prog$ on $v$ for each entry of $v$, up to the first instance of divergence and $\semweakproggrn{\prog}(V)$ computes the reduced set of vector states obtained by running every input in $V$.

% Weak vector-state semantics are identical to strong-vector state semantics except that, when a program diverges on a single input in a vector, the rest of the vector is unspecified. Thus in strong vector-state semantics, we can reason about multiple inputs that cause a program to diverge. In weak vector-state semantics we can only reason about one divergent input per program.

\begin{example} [\greenCaps Vector-State Semantics is Imprecise]
\label{ex:weak_vs}
  For the program $w = (\Ewhile{x > 3}{\Eassign{x}{x + 1}})$, we have
  $\semweakproggrn{w}([1, 2, 3, 4]) = \{[1, 2, 3, \divg]\}$ 
  and $\semweakproggrn{w}([1, 2, 4, 8]) = \{[1, 2, \divg]\}$.
  Therefore, $\semweakproggrn{w}(\{[1, 2, 3, 4], [1, 2, 4, 8]\}) = \{[1, 2, \divg]\}$.

  As discussed, 
  the application of $\Red$ \emph{loses precision} 
  because it merges both $[1, 2, 3, \divg]$ and $[1, 2, \divg]$ 
  into $[1, 2, \divg]$; 
  The symbol $\divg$ appearing at the end of a vector-state indicates that the behavior beyond this point in the sequence is unknown.
  (indeed, the semantics suggests that, e.g., $[1, 2, 4, \divg]$ also might have been possible).  
  % \sn{Previously : ``It is in this sense that weak vector-state semantics are ``weak''''. This is no longer accurate with 2 weak semantics, so we might want to change the name, especially if we are ditching the strong semantics.}
\end{example}

The extension to a semantics over sets of programs retains the ideas of stopping at the first instance of divergence and occlusion. 
The semantics takes the usual union across all programs in the set and applies the reduction $\Red$ afterwards.

\begin{definition} [\greenCaps Vector-State Semantics of a Set of Programs] \label{def:weak_sem_set}
    The \textit{\green vector-state semantics over set of programs} $\semweaksetgrn{\cdot}_T$ over the signature
    $\setweak(2^\vsvars \rightarrow 2^\vsvars)$
    is defined as follows.
    For a set $V \subseteq \vsvars$, we have 
    \[\semweaksetgrn{\progset}(V) = \Red(\bigcup\limits_{\prog \in \progset} \semweakproggrn{\prog}(V)).\] 
\end{definition}
% In other words, we set $\semweakset{\progset}(V) = \bigcup\limits_{\prog \in \progset} \semweakprog{\prog}(V)$, but we throw away infinite vectors and diverging vectors if we already have diverging prefixes of them.

% \input{Figures/WeakVectorStateSemantics}

\begin{example}
    For all $v \in \vsvars$, $\semweaksetgrn{\{x:=x\}}(\{v\}) = \{v\}$.

    Let $s_j$ denote a program that behaves like the identity on all inputs except when $x_1 = j$, in which case it diverges.
    
    Then $\semweakset{\{s_1,s_2\}}(\{[\statelist{\stateentry{x_1}{4}}, \statelist{\stateentry{x_1}{2}}, \statelist{\stateentry{x_1}{1}}, \statelist{\stateentry{x_1}{5}}]\}) = \{[\statelist{\stateentry{x_1}{4}}, \divg]\}$.
    
    For $S_{id} = \{s_n \mid n \in \N\}$ and for every finite sequence $u$ that does not end in $\divg$, $\semweakset{S_{id}}(u) = \{[\divg], u\}$. 
    When $u$ is infinite or ends in $\divg$, $\semweakset{S_{id}}(u) = \{[\divg]\}$. 
\end{example}

% Why is this the right semantics to capture program-agnostic semantics? In order to determine the outcome of a set of loops ``$\Ewhile{b}{S}$'' on a single input $\state_1$, we would like to be able to determine whether:

% \begin{itemize}
%     \item [i] There is a sequence $[\state_1, \cdots, \state_n]$ so that $\semsingprog{b}(\state_j)[b_t] = true$ for $j < n$, $\semsingprog{b}(\state_j)[b_t] = false$, and for some $s \in S$, $\semsingprog{s}(\state_j) = \state_{j+1}$ on all $j$. In this case the loop terminates in state $\state_n$.
%     \item [ii] There is a sequence $[\state_1, \cdots, \state_n]$ so that $\semsingprog{b}(\state_j)[b_t] = true$ for $j < n$, and for some $s \in S$ and $m < n-1$, $\semsingprog{s}(\state_j) = \state_{j+1}$ when $j < m$ but $\semsingprog{s}(\state_m) =\divg$. In this case, the loop executes $s$ finitely many times until $s$ diverges on $\state_m$.
%     \item [iii] There is an infinite (computable) sequence $[\state_1, \cdots]$ so that $\semsingprog{b}(\state_j)[b_t] = true$ for all $j$, and for some $s \in S$, $\semsingprog{s}(\state_j) = \state_{j+1}$ forall $j$. In this case, the loop executes $s$ infinitely many times (i.e., it diverges).
%     Note that if a prefix of $[\state_1, \cdots]$ satisfies (ii), then (iii) is not needed to ascertain divergence.
% \end{itemize}
\subsubsection{Properties of $\semweaksetgrn{\cdot}$} \label{sec:weak_vs_properties}

  Next, we show that \green vector-state semantics over sets of 
  programs is compositional~(Theorem \ref{thm:weak_inductive})
  and as coarse as possible~(Theorem \ref{thm:weak_minimality}).

%Weak vector-state semantics gives us exactly the ability to answer the three questions posed at the beginning of this section. \loris{edit after editing earlier} The ability to answer these questions is not only sufficient to determine the program-agnostic semantics over sets of looping programs, but also necessary for compositionality (because we can write a loop that runs its body on every entry in a computable vector-state) (Theorem~\ref{thm:weak_minimality}).

%\input{Figures/weak_semantics}

\grnwkcomp*

The proof of this theorem closely mirrors the proof of Theorem~\ref{thm:weak_compositional_yel} with the slight differences that \rone as a shortcut, we can determine the semantics of guards in $B$ directly by checking $\semweaksetgrn{B}(x)$ where $x$ is an enumeration of $\State\restriction_{\varfun(B)}$ and \rtwo we must handle infinite (and diverging) traces in addition to finite, nondiverging traces. 

\begin{proof}
    \Cref{fig:weak_compositional} gives the compositional characterizations of $\semweakset{\cdot}$ over unions and all constructors except while. 
    The only cases worth discussing are if-then-else and while.

    Note that $\semweakset{\progset}$ is determined by its behavior on singleton sets of vector-states, so we restrict our focus to inputs of the form $\{v\} \subset \vsvars$ when outlining the next two rules. The versions over arbitrary sets of vector-states $V$ are given in \Cref{fig:weak_compositional}.
    % \loris{changed to sets in fig?}

    % \begin{itemize}
%         \item \textit{Expressions} - Weak vector-state semantics are compositional over expressions and Boolean expressions for the same reasons as strong vector-state semantics.
% \twr{An ordering problem: strong vector-state semantics is not introduced until the next section.}
%         For each constructor, we can evaluate a given vector on each child (because expressions don't diverge) and combine the results in each entry. For example, $\semweakset{E_1 + E_2}(v) = \semweakset{E_1}(v) + \semweakset{E_2}(v)$ where addition is performed elementwise, and we take addition of sets to mean $X_1 + X_2 = \{x_1 + x_2 \mid x_1 \in X_1 \land x_2 \in X_2\}$.
%         \item \textit{Assign} - For Assign, $$\semweakset{x := E}(v) = \semweakset{E}(v)\subs{e_t}{x}$$ where the value of $x$ in each vector entry is replaced with the value of $e_t$.
%         \item \textit{Seq} - For Seq, $$\semweakset{S1;S2}(v) = \semweakset{S2}(\semweakset{S1}(x))$$ If either $S_1$ or $S_2$ diverges on an input, then $S_1;S_2$ diverges on that input as well.
        \paragraph{If-Then-Else} 
        
        Consider sets of statements $S_1$ and $S_2$ and a set of guards $B$. We would like $\semweakset{\Eifthenelse{B}{S_1}{S_2}}(v)$ to, for each guard $b$, evaluate $\semweakset{S_1}$ on a vector of the entries of $v$ where $b$ is true, and evaluate $\semweakset{S_2}$ on a vector of the entries of $v$ where $b$ is false. Then, we'd like to stitch these results back together. 
        For example, $\semweakset{\Eifthenelse{x_1 \neq 1}{S_1}{S_2}}([\statelist{\stateentry{x_1}{1}},\statelist{\stateentry{x_1}{4}}])$ contains $[\state_a, \state_b]$ if
        $[\state_a] \in \semweakset{S_2}([\statelist{\stateentry{x_1}{1}}])$
        and
        $[\state_b]$ $\in \semweakset{S_1}([\statelist{\stateentry{x_1}{4}}])$.
        Similarly, if $[\divg]$ $\in \semweakset{S_2}([\statelist{\stateentry{x_1}{1}}])$, then the output diverges and $[\divg] \in \semweakset{\Eifthenelse{x \neq 1}{S_1}{S_2}}([\statelist{\stateentry{x_1}{1}}, \statelist{\stateentry{x_1}{4}}])$.
        
        For a conditional $\Eifthenelse{B}{S1}{S2}$, each vector state
        $v^b \in \semweakset{B}(v)$ splits $v$ into two subarrays:
        one array over the indices on which $v^b$ is true
        ($v^{1} = \filter{v}{v^b}$), and one over the indices on which $v^b$ is false ($v^{2} = \filter{v}{\neg v^b}$).
        Because $v$ and $v^b$ are computable, $v_1$ and $v_2$ are also computable, we can determine $\semweakset{S1}(v^{1})$ and $\semweakset{S2}(v^{2})$. For each $u^{2} \in \semweakset{S1}(v^{1})$ and $u^{2} \in \semweakset{S2}(v^{2})$, we can assemble a computable vector $u$ by combining $u^{1}$ and $u^{2}$ according to $v^b$, stopping at the first occurrence of $\divg$ if necessary (we say $u = \mathit{interleave}(u^{1}, u^{2}, v^b)$).
        The set of such $u$, after applying $\Red$, is precisely $\semweakset{\Eifthenelse{B}{S1}{S2}}(v)$.

        \paragraph{While} 
        We describe 
        $f^{\setweak}_{while}$ as follows:
        
        Let a set of programs $\Ewhile{B}{S}$ and a vector state $v$ be given. In what follows, treat diverging vectors $[a_1, \cdots, a_n, \divg]$, as $[a_1, \cdots, a_n]$, then append $\divg$ to the end of the outputs if necessary. Note that when $a$ is a vector, we will use $\subvec{a}{i}{j}$ to denote $[a_i, \cdots, a_j]$. Vectors like $[\subvec{a}{i}{j}, b]$ will be understood to be flattened---$[a_i, \cdots, a_j, b_1, \cdots, b_n]$.
        % So, we can ignore input vectors that end in $\divg$.

        For each possible semantics of the guard, we identify traces through the body of the loop that would cause the loop to converge ($T_{v^b, \state}$) or diverge ($TD_{v^b, \state}$) from start state $\state$, and we ask $\semweakset{S}$ whether any $s \in S$ can produce these traces.

        For example, suppose that we want to determine $\semweakset{\Ewhile{x_1 := 4}{S}}([\statelist{\stateentry{x_1}{12}}])$, given a set of statements $S$ where $\varfun(S) = \{x_1\}$. We can consider all computable ``traces'' of executions of a body of this loop. Thus, $[\statelist{\stateentry{x_1}{1}}, \statelist{\stateentry{x_1}{4}}]$ would represent a loop body that, on the first iteration, sent $\statelist{\stateentry{x_1}{12}}$ to $\statelist{\stateentry{x_1}{4}}$. Similarly, we can imagine diverging traces---i.e., $[\statelist{\stateentry{x_1}{12}}, \statelist{\stateentry{x_1}{13}}, \cdots]$ which would capture a body that is executed infinitely many times. We can determine $\semweakset{\Ewhile{x_1 := 4}{S}}([\statelist{\stateentry{x_1}{12}}])$ by querying $\semweakset{S}$ on all such traces.
           
        To start, suppose that we want $\semweakset{\Ewhile{B}{S}([\sigma])}$ for $\sigma \in \State$, and let $x$ be a computable vector of program states so that all distinct states over $\varspres = \varfun(B) \cup \varfun(S)$ appear. (Recall we have access to these sets through our semantics.)

        For each $v^b \in \semweakset{B}(x)$, call $T_{v^b, \state}$ the set of vectors of states (``traces'') $[t_1, \cdots, t_n]$ so that:
        \begin{itemize}
            \item $t_1 = \state$
            \item For each $j$, $v^b[t_j\restriction_{\varspres}][b_t] = \Et$ 
            \item $v^b[t_n\restriction_{\varspres}][b_t] = \Ef$
        \end{itemize}
        The set $T_{v^b, \state}$ represents the possible converging traces of the loop when the guard has the semantics of $v^b$. For each $t \in T_{v^b, \state}$, if $[t_2, \cdots, t_k, \divg] \in \semweakset{S}(\subvec{t}{1}{n-1})$ for any $k < n$, then the set of loops can diverge on $\sigma$. If $[t_2, \cdots, t_n] \in \semweakset{S}(\subvec{t}{1}{n-1})$, then $[t_n]$ is a possible outcome of the set of loops on $v$.

        Thus, if a loop in the set converges or diverges along a converging trace (one which ends falsifying the guard, $T_{v^b, \state}$), the semantics of $B$ and $S$ let us catch it.

        What if a loop diverges along an infinite trace (i.e., one which never falsifies $v^b$---i.e., a trace in $TD_{v^b, \state}$)?

        For each $v^b \in \semweakset{B}(x)$, call $TN_{v^b, \state}$ the set of computable, infinite vectors of states $[t_1, \cdots]$ so that:
        \begin{itemize}
            \item $t_1 = \state$
            \item For each $j$, $v^b[t_j\restriction_{\varspres}][b_t] = \Et$ 
        \end{itemize}
        If a loop body diverges along this trace, we will find $[\subvec{t}{2}{m}, \divg] \in \semweakset{S}(t)$ for some $m$.
        Otherwise, if no body diverges on this trace, but there is a body that realizes the trace, then $\subvec{t}{2}{\omega} \in \semweakset{S}(t)$.
        In both cases, we capture the divergence.

        Note that we only need to consider traces that are computable. Computable vectors suffice because any realized trace (that is, one that is produced by some $\Ewhile{b}{s}$) can be computed over $\varspres$ by simply running $\Ewhile{b}{s}$.

        In sum, $[\divg] \in \semweakset{\Ewhile{B}{S}}([\sigma])$ only if (i) there is some $v^b \in \semweakset{B}(x)$ and $t \in T_{v^b, \state}$ so that $[t_2, \cdots, t_k, \divg] \in \semweakset{S}(\subvec{t}{1}{n-1})$ for some $k < n$ or (ii) there is some $v^b \in \semweakset{B}(x)$ and $t \in TD_{v^b, \state}$ so that $[t_2, \cdots]$ (or an occluder thereof) is in $\semweakset{S}(t)$.

        Similarly, when $\sigma' \neq ~\divg$, then $[\sigma'] \in \semweakset{\Ewhile{B}{S}}([\sigma])$ only if there is some $v^b \in \semweakset{B}(x)$ and $t \in T_{v^b, \state}$ so that $[t_2, \cdots, t_n] \in \semweakset{S}(\subvec{t}{1}{n-1})$ and $t_n = \sigma'$. Thus, from $\semweakset{S}$ and $\semweakset{B}$, we can precisely determine $\semweakset{\Ewhile{B}{S}}$ on vectors of length $1$.\\

        We can extend the argument to longer input vectors by concatenating traces as follows:

        To check for output vectors without $\divg$, given input vector $v = [v_1,\cdots]$ replace $T_{v^b, \state}$ with $\times_{j \leq \abs{v}} T_{v^b, v_j}$. Thus $u = [u_1,\cdots] \in \semweakset{\Ewhile{B}{S}}([v_1,\cdots])$ only if $u$ and $v$ agree off of $\varspres$ and there is some $v^b \in \semweakset{B}(x)$ and $(t^{1},\cdots) \in \times_{j \leq \abs{v}} T_{v^b, v_j}$ computable so that $u_j = t^{j}_{n_j}$ for $j \leq \abs{u}$ and $[\subvec{t^{1}}{2}{n_1}, \subvec{t^{2}}{2}{n_2}, \cdots]$ is contained in $\semweakset{S}([\subvec{t^{1}}{1}{n_1-1}, \subvec{t^{2}}{1}{n_2-1},\cdots])$. Again, we only concern ourselves with computable vectors of traces because, if one exists, it can be computed by running the loop that produces it.

        To check for diverging output vectors $u$, given finite or infinite input $v = [v_1,\cdots]$, we identify vectors diverging in the $k$th slot by replacing $T_{v^b, \state}$ with $\times_{j \leq k} T_{v^b, v_j}$ and $TD_{v^b, \state}$ with $(\times_{j \leq k} T_{v^b, v_j}) \times TD_{v^b, v_k}$. Thus $u = [u_1, \cdots, u_{k-1}, \divg]$ or an occluder thereof is contained in $\semweakset{\Ewhile{B}{S}}([v_1,\cdots])$ if  either:
        \begin{enumerate}
            \item [(i)] there is some $v^b \in \semweakset{B}(x)$ and $(t^{1},\cdots, t^{k}) \in \times_{j \leq k} T_{v^b, v_j}$ so that $u_j = t^{j}_{n_j}$ for $j < k$ and $[\subvec{t^{1}}{2}{n_1}, \subvec{t^{2}}{2}{n_2}, \cdots, \subvec{t^{k-1}}{2}{n_{k-1}}, \subvec{t^{k}}{2}{m}, \divg]$ is contained within $\semweakset{S}([\subvec{t^{1}}{1}{n_1-1}, \subvec{t^{2}}{1}{n_2-1}, \cdots, \subvec{t^{k}}{1}{n_k-1}])$ for $m \leq n_k$.

            \item [(ii)] (i) does not occur, but there is some $v^b \in \semweakset{B}(x)$ and $(t^{1},\cdots, t^{k}) \in (\times_{j \leq k-1} T_{v^b, v_j}) \times TD_{v^b, v_k}$ so that $\semweakset{S}([\subvec{t^{1}}{1}{n_1-1}, \subvec{t^{2}}{1}{n_2-1}, \cdots, \subvec{t^{k-1}}{1}{n_{k-1}-1}, \subvec{t^{k}}{1}{\omega}])$ contains either $[\subvec{t^{1}}{2}{n_1},\subvec{t^{2}}{2}{n_2}, \cdots, \subvec{t^{k-1}}{2}{n_{k-1}}, \subvec{t^{k}}{2}{\omega}]$ or an occluder diverges after $t^{k-1}$.
        \end{enumerate}
        Thus, from $\semweakset{B}$ and $\semweakset{S}$, we can uniquely identify $\semweakset{\Ewhile{B}{S}}$.
\end{proof}

As an analog of $\semweaksetyel{\cdot}$, we have a result analogous to Theorem~\ref{thm:weak_minimality_yel} for $\semweaksetgrn{\cdot}$.
Note that the restriction of the semantics to computable vectors (i.e., through $V_{inf}$ in Def~\ref{def:vs_vars}) is not necessary to the semantics' definition or its compositionality above, but it is necessary for the proof that $\semweaksetgrn{\cdot}$ is the coarsest semantics that is finer than $\semsingsetgrn{\cdot}$ below.

% \loris{this sentence is a repeat?}
% The semantics $\semweakset{\cdot}$ is the \loris{avoid word minimal, use coarse/fine} minimal compositional semantics that \loris{captures is vague word} captures program-agnostic semantics.
% \loris{I propose replacing fine-grained coarse-grained everywhere with just finer/coarser}
\grnwkmin*

The proof that $\semweaksetgrn{\cdot} \restriction_{\State^* \cdot \{\epsilon, \divg\}} \preceq \abstrs{A}$ closely mirrors the proof of \Cref{thm:weak_minimality_yel}. The proof that $\semweaksetgrn{\cdot} \restriction_{\vinf} \preceq \abstrs{A}$ is carried out by first checking whether any loop body diverges on an infinite input vector-state and, if no body diverges, considering a set of loops which runs every entry in the (computable) input vector-state and diverges only when the loop body returns the wrong value.

\begin{proof}
    Clearly, $\semweakset{\progset}$ determines $\semsingset{\progset}$ because $\semsingset{\progset}(\sigma) = \{u[1] \mid u \in \semweakset{\progset}([\sigma])\}$. This gives a function mapping each $\semweakset{\progset}$ to $\semsingset{\progset}$. So $\semweakset{\cdot}$ is at least as fine as $\semsingset{\cdot}$.

    Let us show $\semweakset{\cdot}$ is the coarsest such semantics. Suppose an inductively defined set of statements $S$ and a compositional semantics $\abstrs{A}$ at least as fine as program-agnostic semantics $\semsingset{\cdot}$ are given. We will construct set of loops $W_\alpha$ so that $\abstrsapp{A}{S}$ gives each $\abstrsapp{A}{W_\alpha}$, each $\abstrsapp{A}{W_\alpha}$ gives $\semsingset{W_\alpha}$, and the $\semsingset{W_\alpha}$ together give $\semweakset{S}$. This will show $\abstrsapp{A}{S}$ gives $\semweakset{S}$.

    Given $v \in \vsvars$, we can determine $\semweakset{S}(v)$ as follows. As usual, we account for vectors ending in $\divg$ by ignoring the last entry and appending $\divg$ to the outputs if necessary.\\

    Recall $\varfun(S)$ is provided by $\abstrsapp{A}{S}$. Let $j$ be a variable not in $\varfun(S)$.

    Suppose $v = [v_1, \cdots, v_n]$ is finite, and consider a desired vector of outputs $u = [u_1, \cdots, u_n]$. Clearly, if $u$ does not agree with $v$ off of $\varfun(S)$, then $u \notin \semweakset{S}(v)$. So suppose they do agree.
    
    We can write a set of loops $W_{v,u}$ in which each loop runs a program in $S$ on every $v_j$, breaking when the wrong output is produced and incrementing a counter ($j \in \vars \setminus \varfun(S)$) for each $v_j$ correctly mapped to $u_j$:
\begin{lstlisting}[linewidth=\linewidth]
j := 1;
while (0 < j <= n) {
    (*@$\Eassign{\state}{v_j\restriction_{\varfun(S)}}$; @*) 
    S;
    if (*@$\sigma == {u_j}\restriction_{\varfun(S)} $@*) {
        j := j+1
    } else {
        j := 0
    }
}
\end{lstlisting}
    
    % \begin{itemize}
    %     \item $\Eassign{j}{1}$;
    %     \item while $0 < j <= n$:
    %     \begin{itemize}
    %         \item $\sigma := {v_j}\restriction_{\varfun(S)}$
    %         \item $S$
    %         \item $\Eifthenelse{\sigma == {u_j}\restriction_{\varfun(S)}}{j:=j+1}{j:=0}$
    %     \end{itemize}
    % \end{itemize}
    
    Then $u \in \semweakset{S}(v)$ if and only if there is a state $\sigma'$ with $j = n+1$ in $\semsingset{W_{v,u}}(\state)$ for arbitrary $\state \in \State$. If $\divg \in \semsingset{W_{v,u}}(\state)$, then $[u_1, \cdots, u_{n-1}, \divg]$ or some diverging prefix thereof is in $\semweakset{S}(v)$. We can determine that prefix by repeating this analysis for prefixes of $v$. By investigating various $u$ and prefixes of $v$, we can fully determine $\semweakset{S}(v)$ from $\semsingset{W_{v,\cdot}}$. Moreover, because $\abstrs{A}$ is compositional, $\abstrsapp{A}{S}$ determines each $\abstrsapp{A}{W_{v,u}}$. Since $\abstrs{A}$ is at least as fine as $\semsingset{\cdot}$, $\abstrsapp{A}{W_{v,u}}$ determines $\semsingset{W_{v,u}}$. Thus, $\abstrsapp{A}{S}$ determines $\semweakset{S}$ on finite inputs.\\

    Now suppose $v = [v_1, \cdots]$ is infinite.
    For each prefix of $v$, we can do the above to determine the divergent vectors in $\semweakset{S}(v)$. 
    The only question left is, given an infinite vector $u$, to determine whether $u \in \semweakset{S}(v)$. If $u$ has a divergent prefix in $\semweakset{S}(v)$, the answer is no; we have seen above how to check this. If $u$ is not computable, the answer is also no. Otherwise, we can determine whether $u \in \semweakset{S}(v)$ as follows:
    
    Because $v$ and $u$ are computable, we can write a set of loops $W$, each of which takes some $s \in S$ and runs it on each $v_j$ in turn, checking that the output matches $u_j$ at each step:
\begin{lstlisting}[linewidth=\linewidth]
j := 1;
while (0 < j) {
    (*@$\Eassign{\state}{v_j\restriction_{\varfun(S)}}$; @*) 
    S;
    if (*@$\sigma == {u_j}\restriction_{\varfun(S)} $@*) {
        j := j+1
    } else {
        j := 0
    }
}
\end{lstlisting}
    % \begin{itemize}
    %     \item $\Eassign{j}{1}$;
    %     \item while $0 < j$:
    %     \begin{itemize}
    %         \item $\sigma := v_j\restriction_{\varfun(S)};$
    %         \item $S$
    %         \item $\Eifthenelse{\sigma == u_j\restriction_{\varfun(S)}}{j:=j+1}{j:=0}$
    %     \end{itemize}
    % \end{itemize}

    We have established in our casework that $u$ has no divergent prefix in $\semweakset{S}(v)$. Thus, $\divg \in \semsingset{W}(\state)$ on arbitrary $\state$ if and only if there is an $s \in S$ so that $s$ maps each $v_j$ to $u_j$. In other words, $\divg \in \semsingset{W}(\state)$ if and only if $u \in \semweakset{W}(v)$.

    Thus, we can fully determine $\semweakset{S}(v)$ from the program-agnostic semantics of loops running $S$. As in the finite case, compositionality and granularity of $\abstrs{A}$ imply that $\semweakset{S}(v)$ is uniquely determined by $\abstrsapp{A}{S}$ for all infinite $v$.

    For integer expressions $E$ (or Boolean expressions $B$), we can replace $S$ with $x_1 := E$ (or $\Eifthenelse{B}{x_1:=0}{x_1:=1}$) and the argument proceeds in the same way.

    So, $\abstrs{A}$ is at least as fine as $\semweakset{\cdot}$.
\end{proof}

% \begin{proofsketch}
%     Clearly, $\semweakset{\progset}$ determines $\semsingset{\progset}$ because $\semsingset{\progset}(\sigma) = \{u[1] \mid u \in \semweakset{\progset}([\sigma])\}$. This gives a function mapping each $\semweakset{\progset}$ to $\semsingset{\progset}$. 
%     So $\semsingset{\cdot} \preceq \semweakset{\cdot}$.

%     We must then show that $\semweakset{\cdot}$ is the is the coarsest compositional semantics at
%     least as fine as $\semsingset{\cdot}$. To do so, we observe that
%     any compositional semantics $\abstrs{A}$ such that $\semsingset{\cdot} \preceq \abstrs{A}$ 
%     must determine sets of statements $S$ well enough that we can reason about the program-agnostic semantics of sets of loops in which $S$ appears in the body. Given a set of statements $S$, for each pair $(v,u)$ of input and output vector-states, we construct sets of loops $W_{v,u}$ that loop over the indices of our vectors and check the behavior of $S$ on each entry (i.e., $\Ewhile{b}{s_1;S;s_2}$, where $s_1$ sets the input state to $v[i]$, $s_2$ checks if the output state is $u[i]$ and breaks if not, and $b$ iterates through $v$). The proof concludes by showing the program-agnostic semantics of such sets of loops essentially determine whether $u \in \semweakset{S}(v)$. 
%     Thus, $\semweakset{\cdot} \preceq \abstrs{A}$ on sets of statements. A similar argument carries to expressions.
% \end{proofsketch}

\subsection{The Granularity of \texorpdfstring{$\semweaksetyel{\cdot}$}{} and \texorpdfstring{$\semweaksetgrn{\cdot}$}{}} \label{app:vec_compare}

Clearly $\semweaksetyel{\cdot} \not\preceq \semsingsetgrn{\cdot}$ following the proof of Theorem~\ref{thm:sing_noninductive}. Of course, $\semsingsetgrn{\cdot} \not\preceq \semweaksetyel{\cdot}$ because $\semweaksetyel{\emptyset} = \semweaksetyel{\{\Ewhile{\Et}{\Eassign{x}{x}}\}}$ but $\semweaksetgrn{\emptyset} \neq \semweaksetgrn{\{\Ewhile{\Et}{\Eassign{x}{x}}\}}$.

Now, comparing $\semweaksetyel{\cdot}$ and $\semweaksetgrn{\cdot}$, since $\semweaksetgrn{\cdot}$ is a compositional semantics finer than $\semsingsetyel{\cdot}$, $\semweaksetyel{\cdot} \preceq \semweaksetgrn{\cdot}$. Moreover, since $\semweaksetyel{\cdot}$ is incomparable with $\semsingsetgrn{\cdot}$, $\semweaksetyel{\cdot}$ and $\semweaksetgrn{\cdot}$ do not have equivalent granularity. Thus $\semweaksetyel{\cdot} \prec \semweaksetgrn{\cdot}$.

Interestingly, the strict $\semweaksetyel{\cdot} \prec \semweaksetgrn{\cdot}$ holds even over sets of nondiverging programs.

\begin{restatable} [$\semweaksetyel{\cdot} \prec \semweaksetgrn{\cdot}$]{thm}{yelwklessgrnwk}
    $\semweaksetyel{\cdot} \prec \semweaksetgrn{\cdot}$ over sets of nondiverging programs.
\end{restatable}
\begin{proof}
    As above, $\semweaksetyel{\cdot} \preceq \semweaksetgrn{\cdot}$ in general, so $\semweaksetyel{\cdot} \preceq \semweaksetgrn{\cdot}$ on sets of nondiverging programs. It remains to show $\semweaksetgrn{\cdot} \not\preceq \semweaksetyel{\cdot}$ on sets of nondiverging programs. We do so by borrowing the sets of programs introduced by \citet{uls} in their Example 5.1.

    Consider the set $S_1 = L(S)$ defined by the grammar with nonterminals $S::= \Eifthenelse{x_2 == N}{\Eassign{x_1}{N}}{S} \mid \Eassign{x_1}{N}$ and $N::= 0 \mid 1 + N$. Consider also $S_2 = S_1 \cup \{\Eassign{x_1}{x_2}\}$. One may observe that the behavior of $\Eassign{x_1}{x_2}$ on finitely many input states may be captured by some $s \in S_1$ by hardcoding the states. However, no program in $S_1$ has the semantics of $\Eassign{x_1}{x_2}$ because all programs in $S_1$ have finite range. Said another way, the function computed by $\Eassign{x_1}{x_2}$ is a limit point in the Baire topology of the set of functions computed by the programs in $S_1$.

    Observe that $\semweaksetyel{S_1} = \semweaksetyel{S_2}$ because, for any finite vector-state, there is a program $s \in S_1$ that matches the behavior of $\Eassign{x_1}{x_2}$ on it. However, $\semweaksetgrn{S_1} \neq \semweaksetgrn{S_2}$. To see this, one may consider an infinite vector-state $v = [\statelist{\stateentry{x_1}{0}, \stateentry{x_2}{1}}, \statelist{\stateentry{x_1}{0}, \stateentry{x_2}{2}}, \cdots]$. All output vector-states in $\semweaksetgrn{S_1}(v)$ will have an upper bound on the value of $x_1$, but $[\statelist{\stateentry{x_1}{1}, \stateentry{x_2}{1}}, \statelist{\stateentry{x_1}{2}, \stateentry{x_2}{2}}, \cdots] \in \semweaksetgrn{S_2}(v)$ has no bound on $x_1$.
\end{proof}

Of course, both vector-state semantics are strictly coarser than the program-aware semantics. In the case of the \green vector-state semantics, this is a consequence of blocking our ability to observe multiple instances of divergence.

\begin{restatable}
 [$\semweaksetgrn{\cdot}$ Strictly Coarser than $\semfullset{\cdot}$]{thm}{weaklessthanstrong}\label{thm:weak_less_than_strong}
    $\semweaksetgrn{\cdot}\prec \semfullset{\cdot}$.
    % Weak vector-state semantics is strictly coarser than program-aware semantics.    
\end{restatable}
\begin{proof}
    First, observe that program-aware semantics is at least as fine as vector-state semantics. We can determine $\semweaksetgrn{\progset}$ on singleton sets from $\semfullset{\progset} = \{\semsingproggrn{\prog} \mid \prog \in \progset\}$ by, on an input $v \in \vsvars$, determining $\{[\semsingproggrn{\prog}(v_1), \cdots, ] \mid \prog \in \progset\}$ and applying $\Red$. Of course, $\semweaksetgrn{\progset}$ is determined by its value on singleton sets.

    Second, we show vector-state semantics is strictly coarser than program-aware semantics. Consider the following four partial computable functions, where $\divg$ stands for $\Ewhile{\Et}{x_1:=x_1}$:
    \begin{itemize}
        \item $f(x_1) := x_1$
        \item $f_1(x_1) := \Eifthenelse{x_1==1}{\divg}{x_1}$
        \item $f_2(x_1) := \Eifthenelse{x_1==2}{\divg}{x_1}$
        \item $f_{1,2}(x_1) := \Eifthenelse{(x_1 == 1 \lor x_1 == 2)}{\divg}{x_1}$
    \end{itemize}

    Let $S_1 = \{f, f_{1,2}\}$ and $S_2 = \{f, f_1, f_2\}$. Evidently $\semfullset{S_1} = \{\semsingproggrn{f},\semsingproggrn{f_{1,2}}\}$ and $\semfullset{S_2} = \{\semsingproggrn{f}, \semsingproggrn{f_1}, \semsingproggrn{f_2}\}$. Since $\semsingproggrn{f_1} \neq \semsingproggrn{f}, \semsingproggrn{f_{1,2}}$, we see that $\semsingproggrn{f_1} \in \semfullset{S_1} \setminus \semfullset{S_2}$ and thus $\semfullset{S_1} \neq \semfullset{S_2}$.
    
    However, $\semweaksetgrn{S_1} = \semweaksetgrn{S_2}$. For any input vector $v$ not containing a state setting $x_1$ to $1$ or $2$, all functions behave identically, and the outputs of the semantics are equal. When $\stateentry{x_1}{1}$ and $\stateentry{x_1}{2}$ appear in $v$, if $\stateentry{x_1}{1}$ appears first, then $\semweakproggrn{f_1}(v) = \semweakproggrn{f_{1,2}}(v)$, and moreover $\semweakproggrn{f_1}(v)$ is a prefix of $\semweakproggrn{f_2}(v)$. So 
    \begin{align*}
        \semweaksetgrn{S_1}(v) &= \Red(\{\semweakproggrn{f_1}(v), \semweakproggrn{f_2}(v), \semweakproggrn{f}(v)\})\\ 
        &= \{\semweakproggrn{f_1}(v), \semweakproggrn{f}(v)\}\\ 
        &= \{\semweakproggrn{f_{1,2}}(v), \semweakproggrn{f}(v)\}\\ 
        &= \semweaksetgrn{S_2}(v).
    \end{align*} If $\stateentry{x_1}{2}$ appears first, the argument is analogous. When $v$ is infinite, the sets are still equal, but $\semweakproggrn{f}(v)$ does not appear in them.

    Thus, $\semweaksetgrn{S}$ does not determine $\semfullset{S}$ for every $S \in \Stmtset$.
    % In other words, $(\semweakset{\cdot})^{-1}$ is not a refinement of $(\abstrsapp{F}{\cdot})^{-1}$ when $\varfun(\progset)$ is fixed. 
    So $\semweaksetgrn{\cdot} \prec \semfullset{\cdot}$.
\end{proof}

While program-aware semantics is strictly finer than \green vector-state semantics for sets of \textit{diverging} programs, the two semantics are equivalent for sets of \textit{nondiverging} programs in which only finitely many program variables appear. This distinction suggests that even \green vector-state semantics is a rather expressive semantics. When $\Tset$ is a collection of sets of programs, call $ND(\Tset)$, the collection of sets $\progset \in \Tset$ such that $\progset$ contains only nondiverging programs and $\abs{\varfun(\progset)}<\infty$.

\begin{restatable} [$\semweaksetgrn{\cdot}$ 
Equally
Granular as $\semfullset{\cdot}$ over Sets of Nondiverging Programs]{thm}{grnwkeqconvaw}
    For sets of nondiverging programs, $\semweaksetgrn{\cdot}$ and $\semfullset{\cdot}$ have equal granularity.
    
    That is, there are a functions $f$ and $g$ so that, for all $\progset \in ND(\Stmtset) \cup ND(\Expset) \cup ND(\Boolset)$, $f(\semweaksetgrn{\progset}) = \semfullset{\progset}$ and $g(\semfullset{\progset}) = \semweaksetgrn{\progset}$.

    % That is, if $CStmt(vars)$ is the set of nondiverging statements in $L(\gimp)$ and $CP(vars) = 2^{CStmt(vars)} \cup 2^{Exp(vars)} \cup 2^{Bool(vars)}$, then there are a functions $f$ and $g$ so that, for all $S \in CP(vars)$, $f(\semweakset{S}) = \semstrongset{S}$ and $g(\semstrongset{S}) = \semweakset{S}$.
\end{restatable}

\begin{proof}
    Because $\semweaksetgrn{\cdot} \preceq \semfullset{\cdot}$ over 
    $\Tset \in \{\Stmtset, \Expset, \Boolset\}$, 
    it follows that $\semweaksetgrn{\cdot} \preceq \semfullset{\cdot}$ over $ND(\Tset)$. 
    One can use as $g$ the same function described in the proof of 
    Theorem~\ref{thm:weak_less_than_strong}.

    In the other direction, for $f$, consider some $\progset \in ND(\Tset)$. Let $\semweaksetgrn{\progset}$ be given, and let $v \in \vsvars$ be an infinite vector so that every state over $\varfun(\progset)$ appears in $v$.

    Because no programs in $\progset$ diverge, $\semweaksetgrn{\progset}(v) = \{[\semsingproggrn{\prog}(v[1]), \cdots] \mid \prog \in \progset\}$. Because $v$ covers all states over $\varfun(\progset)$, $[\semsingproggrn{\prog}(v[1]), \cdots ]$ determines $\semsingproggrn{\prog}$. Taking $T \in \{\Stmtvars, \Expvars, \Boolvars\}$ respectively, we see that for any $\prog \in \Tvars$,
    $\semsingproggrn{\prog} \in \semfullset{\prog}$ iff $[\semsingproggrn{\prog}(v[1]), \cdots] \in \semweaksetgrn{\progset}(v)$.

    Then $\semweaksetgrn{\cdot}$ determines $\semfullset{\cdot}$ over $ND(\Tset)$
\end{proof}

The three compositional semantics $\semweaksetyel{\cdot}$, $\semweaksetgrn{\cdot}$, and $\semfullset{\cdot}$ offer three different levels of granularity for different use cases. If someone wishes to design a compositional reasoning system for sets of programs, they may choose one of these semantics as the object their system reasons about.

\end{document}